\documentclass[12pt]{elsarticle}

\usepackage{geometry}
\usepackage{amsthm} 
\usepackage[hidelinks,colorlinks=false]{hyperref}
\newtheorem{theorem}{Theorem}
\newtheorem{lemma}{Lemma}

\newtheorem{definition}{Definition}
\usepackage{bm}

\usepackage{amssymb}
\usepackage{amsmath}

\usepackage{amssymb}
\usepackage{xcolor}
\usepackage{xspace}
\usepackage{color,soul}  
\usepackage{graphicx}
\usepackage{amsmath}
\usepackage{enumerate}
\usepackage{enumitem}
\usepackage{fancybox}
\usepackage{comment}
\usepackage{tikz}
\usepackage{framed}
\usepackage[bb=stix]{mathalpha}
\usepackage{amsfonts}
\usepackage{mathalfa}
\usepackage{stmaryrd}
\usetikzlibrary{shapes,fit}
\usetikzlibrary{arrows, positioning, calc}
\usetikzlibrary{trees}
\usepackage{pifont}
\usepackage{orcidlink}
\usepackage[none]{hyphenat}

\usetikzlibrary{positioning, fit, calc}   
\tikzset{block/.style={draw, thick, text width=8cm ,minimum height=0.75cm, align=center},   
line/.style={-latex}     
}  
\usepackage{hyperref} 
\usepackage{float}
\usepackage{algpseudocode}
\hypersetup{colorlinks=true}
\usepackage{url} 
\usepackage{enumitem}  
\usepackage{multicol}
\usepackage{framed}

\usepackage{subcaption}
 \usepackage{setspace}
\usepackage{array}
\usepackage{marvosym}

\newcommand{\crs}{\mathsf{crs}}
\newcommand{\LL}{\mathcal{L}}
\newcommand{\PP}{\mathsf{P}}
\newcommand{\VV}{\mathsf{V}}

\newcommand{\OO}{\mathcal{O}}
\newcommand{\uniform}{\xleftarrow{\$}
}
\newcommand{\bH}{\mathbf{H}}
\newcommand{\bT}{\mathbf{T}}

\newcommand{\Invert}{\mathsf{Invert}}

\newcommand{\GAME}{\mathbf{Game}}

\newcommand{\PVSS}{\mathsf{PVSS}}
\newcommand{\Dec}{\mathsf{Dec}}
\newcommand{\Key}{\mathsf{Key}}
\newcommand{\Gen}{\mathsf{Gen}}
\newcommand{\Trap}{\mathsf{Trap}}
\newcommand{\TrapGen}{\mathsf{TrapGen}}
\newcommand{\xmark}{\ding{55}}%

\newcommand{\Adversary}{\mathcal{A}}

\newcommand{\PKE}{\mathsf{PKE}}

\newcommand{\bff}{\mathbf{f}}
\newcommand{\bd}{\mathbf{d}}

\newcommand{\BadChallenge}{\mathsf{BadChallenge}}
\newcommand{\negl}{\mathsf{negl}}

\newcommand{\pk}{\mathsf{pk}}

\newcommand{\bv}{\mathbf{v}}
\newcommand{\bx}{\mathbf{x}}
\newcommand{\bb}{\mathbf{b}}
\newcommand{\bk}{\mathbf{k}}

\newcommand{\ba}{\mathbf{a}}
\newcommand{\sk}{\mathsf{sk}}

\newcommand{\RAND}{\mathcal{RAND}}

\newcommand{\tempcaption}{}

\newcommand{\condprob}[2]{\Pr\left[ \begin{array}{l}
		#1
	\end{array}\;\middle|\;\begin{array}{l}
		#2
	\end{array} \right]}

\newcommand{\Simulator}{\mathcal{S}}

\newcommand{\Share}{\mathsf{Share}}

\newcommand{\pp}{\mathsf{pp}}

\newcommand{\by}{\mathbf{y}}
\newcommand{\bt}{\mathbf{t}}
\newcommand{\bz}{\mathbf{z}}
\newcommand{\bh}{\mathbf{h}}
\newcommand{\bmm}{\mathbf{m}}

 \newcommand{\bs}{\mathbf{s}}
  \newcommand{\bu}{\mathbf{u}}

  \newcommand{\bw}{\mathbf{w}}
    
 \newcommand{\be}{\mathbf{e}}
  \newcommand{\bA}{\mathbf{A}}
    
        \newcommand{\bzero}{\mathbf{0}}
   \newcommand{\RR}{\mathcal{R}}
   \newcommand{\bB}{\mathbf{B}}

 \newcommand{\uniformly}{\stackrel{\$}{\leftarrow}}
  
  \newcommand{\br}{\mathbf{r}} 
 \newcommand{\bC}{\mathbf{C}}

 \newcommand{\msg}{\mathsf{msg}}
 \newcommand{\tr}{\mathsf{tr}}
   
  \newcommand{\Enc}{\mathsf{Enc}}

\newcommand{\SSS}{\mathsf{SSS}}

\newcommand{\Prove}{\mathsf{Prove}} 
 
\newcommand{\Verify}{\mathsf{Ver}}
 
\newcommand{\NIZK}{\mathsf{NIZK}} 
\newcommand{\Setup}{\mathsf{Setup}} 
 
\newcommand{\Combine}{\mathsf{Combine}} 
 
\newcommand{\bc}{\mathbf{c}} 
\newcommand{\Adv}{\mathbf{Adv}}

\newcommand{\CC}{\mathcal{C}}
\newcommand{\Correctness}{\mathsf{Correctness}}

\newenvironment{framedfigure}[1][\makered{Missing caption}]
{
	\def\myenvargumentI{#1} 
	\begin{figure*}[!h]
		\small
		\begin{tabular}{|p{\textwidth}|}
			\hline 
            \small
		}
		{
			\\\hline
		\end{tabular}
		\caption{\myenvargumentI}
		\vspace*{-1em}
	\end{figure*}
}

\newenvironment{xenumerate}%
{\begin{enumerate}}%
{\end{enumerate}}%

\journal{Computer Standards $\&$ Interfaces}
\begin{document}

\begin{frontmatter}



\title{Publicly Verifiable Secret Sharing: Generic Constructions and Lattice-Based Instantiations in the Standard Model}



\author[add1,add2]{Pham Nhat Minh}
\ead{pnminh.sdh232@hcmut.edu.vn}
\author[add3]{Khoa Nguyen}
\ead{khoa@uow.edu.au}
\author[add3]{Willy Susilo}
\ead{ wsusilo@uow.edu.au}
\author[add1,add2]{Khuong Nguyen-An\corref{cor1}}
\ead{nakhuong@hcmut.edu.vn}

\cortext[cor1]{Corresponding author.}
\affiliation[add1]{organization={Department of Computer Science, Faculty of Computer Science and Engineering,\\ University of Technology (HCMUT)},
             addressline={268 Ly Thuong Kiet Street, District 10},\\
             city={Ho Chi Minh City},
             postcode={70000},
             country={Vietnam}}

\affiliation[add2]{organization={Vietnam National University Ho Chi Minh City},
             addressline={Linh Trung Ward, Thu Duc City},\\
             city={Ho Chi Minh City},
             postcode={70000},
             country={Vietnam}}

\affiliation[add3]{organization={Institute of Cybersecurity and Cryptology, School of Computing and Information Technology, University of Wollongong},
             city={Wollongong, NSW 2522},
             country={Australia}}


\begin{abstract}
Publicly verifiable secret sharing (PVSS) allows a dealer to share a secret among a set of shareholders so that the secret can be reconstructed later from any set of qualified participants. In addition, any public verifier should be able to check the correctness of the sharing and reconstruction process. PVSS has been demonstrated to yield various applications, such as e-voting, decentralized random number generation protocols, and secure computation on distributed networks. Although many concrete PVSS protocols have been proposed, their security is either proven in the random oracle model or relies on quantum-vulnerable assumptions such as factoring or discrete logarithm. In this work, we propose a generic construction of a PVSS from any public key encryption scheme and non-interactive zero-knowledge arguments for suitable gap languages. We then provide lattice-based instantiations of the underlying components, which can be proven in the standard model.  
As a result, we construct the first post-quantum PVSS in the standard model, with a reasonable level of asymptotic efficiency.
\end{abstract}


\begin{highlights}
\item The paper proposes a generic construction for publicly verifiable secret sharing (PVSS) schemes from any IND-CPA-secure PKE scheme and a NIZK system for certain gap languages. 


\item The construction can be instantiated in the standard model based on the LWE assumption and can support Shamir's secret sharing scheme. To obtain the underlying NIZK system, we design dedicated trapdoor $\Sigma$-protocols for lattice-based relations and then make them non-interactive using recent techniques on Fiat-Shamir in the standard model.

\item The resulting lattice-based PVSS scheme is the first one achieving post-quantum security in the standard model.
\end{highlights}

\begin{keyword}
PVSS \sep generic construction \sep lattice-based cryptography\sep LWE \sep standard model\sep trapdoor $\Sigma$-protocols\sep NIZK arguments


\end{keyword}

\end{frontmatter}

\sloppy 
\section{Introduction}
Secret sharing scheme (SSS) \cite{Sha79} allows a dealer to share a secret among a committee of shareholders so that any qualified set of participants can recover the secret, while any unqualified set of participants learns nothing about it. Verifiable secret sharing~\cite{CGMA85} (VSS) allows the shareholder to verify the process of sharing and reconstruction against malicious dealers (who might distribute invalid shares) and participants (who might submit wrong shares). Publicly verifiable secret sharing (PVSS) takes one step further by allowing anyone, not just the shareholder, to publicly verify the correctness of the sharing and reconstruction process. Normally, PVSS requires at least two phases (and some also need a key generation phase at the beginning), but a PVSS can be considered non-interactive (for example,~\cite{GHL22,Sch99}) if all phases are non-interactive. PVSS has many vital applications, for example, e-voting \cite{Sch99}, e-cash \cite{YY00}, decentralized random number generation protocols (DRNGs) \cite{KRDO17,CD17,SJSW18,GSX20,CD20,DKI022,BSLKN21,CDSV23,BL23}, and is also potentially used for secure computation on distributed networks \cite{GHL22}. So far, many PVSS constructions have been proposed in literature, from group-based, Paillier encryption-based, to hybrid lattice and group-based \cite{Stadler96,FO98,B99,Sch99,YY00,FS01,AJ05,HV08,J11,JVS14,CD17,CD20,DKI022,GHL22,CDGK22,CDSV23,CD24}. However, existing constructions suffer from one of two drawbacks: They either i) require Fiat-Shamir heuristic and could only achieve security in the random oracle model (ROM) or ii) do not achieve post-quantum security due to the reliance on the hardness of either the factoring or the discrete logarithm problems. For the former drawback, \cite{CGH04,GK03} and recently \cite{KRS25} provided counterexamples of protocols that achieve security in the ROM but become insecure when instantiated with any hash functions, making these protocols insecure in the real world. The former two were ``artificial'' counterexamples, but the latter claims to be a realization of a more practical and natural proof system.
Nevertheless, these counterexamples might raise the concern (and this is an open problem) of whether most natural protocols using the Fiat-Shamir heuristic could achieve provable security in the real world (and whether more natural counterexamples could be found). The latter problem goes without saying. Today, quantum computers are being developed, and it is known that discrete log and factoring problems can be easily solved by Shor's quantum algorithm \cite{Shor99}. Thus, when quantum computers are ready to be deployed, the security of existing PVSSs will be compromised. Given the importance of both properties for security, it would be desirable to construct a PVSS that achieves post-quantum security and does not need to rely on the ROM.

Recently, using correlation intractable hash functions \cite{CCHLRRW19}, the authors of \cite{CCHLRRW19,PS19} finally solved the problem of realizing non-interactive zero-knowledge (NIZK) arguments following the Fiat-Shamir paradigm in the common reference string (CRS) model for all NP problems from standard \textit{lattice assumptions}. On the bright side, using the result of \cite[Subsection 3.3]{GMW86} together with the FLS compiler of~\cite{FLS99} to achieve multi-theorem NIZKs for all NP languages, we might be able to build a post-quantum PVSS against malicious participants, from standard post-quantum assumptions only. However, their constructions require reducing the underlying problem to the Graph Hamiltonicity problem, which would be extremely inefficient. It is highly desirable to obtain a reasonably efficient instantiation of the NIZKs without resorting to generic techniques, which can be used as building blocks for designing PVSS schemes. Therefore, this work aims to propose such NIZKs to fully realize a post-quantum secure PVSS using standard assumptions.

\subsection{Our Contribution}
We put forward a generic construction of PVSS from i) any IND-CPA scheme where each public key has a unique corresponding secret key; and ii) a NIZK for suitable \textit{gap languages}, which provides a generalized formal framework for PVSS schemes following the GMW approach of \cite{GMW86} such as \cite{GMW86,GHL22,CDGK22,CD24}. Note that while previous works, e.g.,~\cite{GMW86} or \cite{CDGK22}, have already proposed generic PVSS constructions, the work of \cite{GMW86} only described an informal construction, while the work of \cite{CDGK22} requires the encryption scheme to be linearly homomorphic. In addition, their constructions require NIZK with exact languages. On the other hand, we would like to capture the lattice setting, and thus decide to employ NIZK for gap languages, which can be considered as a generalization of exact languages. Hence, our construction is more generic than \cite{CDGK22}. The only restriction in the PKE is that it needs the public key to have a unique corresponding secret key pair. However, this is still generic enough to capture previous schemes where their public-secret key pair are either from group, Paillier, or LWE-based cryptosystems, which also have this property. We are the first to give a \textit{formal} and generic PVSS based only on an IND-CPA encryption scheme with unique corresponding public-secret key pairs and an NIZK for gap languages.

We then propose nontrivial instantiations of the NIZKs supporting the key generation, sharing, and decryption algorithms from the \textit{plain} LWE assumption. One major point is that the NIZKs in the CRS model are not the result of using Karp reduction to the Graph Hamiltonicity problem (hence nontrivial). Instead, we construct \textit{concrete} trapdoor $\Sigma$-protocols for correct key generation, sharing, and share decryption and use the compiler of \cite{LNPT20} to achieve NIZK with adaptive soundness and adaptive multi-theorem zero-knowledge. As analyzed in Appendix~\ref{appendix-comparison-with-generic-sols}, our NIZK will be \textit{more efficient} than those using the generic Karp reduction.

Finally, we combine the NIZK mentioned above to achieve the \textit{first} nontrivial instantiation of a lattice-based \textit{non-interactive} PVSS in the CRS model without relying on random oracles. PVSS constructions have already been made in the standard model, such as \cite{CD17}. However, previous constructions all require pairing-based assumptions, which are not post-quantum secure. Our construction is from plain LWE only and, therefore, is post-quantum secure. For a technical note, our construction is not as efficient as previous work, as our trapdoor $\Sigma$-protocols require binary challenges and must be repeated $\lambda$ times (where $\lambda$ is the security parameter) to make the soundness error exponentially small, while constructions such as \cite{Stadler96,Sch99,CD17,GHL22,CDGK22} employ challenges over $\mathbb{Z}_p$ (for an exponentially large $p$) and do not need to be repeated. However, this is a trade-off as our strong point is that we could achieve post-quantum security in the standard model, while the previous construction could not. We also achieve the asymptotically smallest required modulus among existing PVSSs so far by only requiring a modulus polynomially large in $\lambda$ and $n$, where $n$ is the number of users. To summarize, among known PVSS schemes, ours is the \textit{most} efficient one that achieves post-quantum security in the standard model. We remark that our PVSS requires a third party to generate the CRS. However, this setting has been considered in previous works such as \cite{GHL22} as well.

The primary contribution of this work is foundational: we establish the theoretical feasibility of the first post-quantum PVSS scheme that is provably secure in the standard model under lattice assumptions. Our focus is on the novel cryptographic constructions and the rigorous security arguments required to solve this long-standing open problem. While a full implementation, concrete parameter selection, and performance benchmarking are critical for practical deployment, such an endeavor constitutes a substantial engineering and research challenge in its own right. It is therefore considered beyond the scope of this foundational paper, which prioritizes establishing the theoretical underpinnings and asymptotic efficiency of our approach. We believe that this work sets the necessary foundations for future research dedicated to practical optimization and implementation.

\subsection{Technical Overview}\label{section-technical-overview}

We follow the (informal) framework of \cite[Subsection 3.3]{GMW86} to have a non-interactive PVSS. Based on the framework, we require four components: An IND-CPA public key encryption scheme (which we denote $\PKE$), a NIZK for participants to prove correct key generation, a NIZK for the dealer to prove the correctness of sharing, and a NIZK for participants to prove the correctness of decrypting the shares\footnote{The framework of \cite{GMW86} only mentions the sharing phase, where the dealer needs to prove the correctness of sharing. However, PVSS also requires a reconstruction phase, and this phase requires an NIZK of correct decryption, which is not mentioned in their construction. In the sharing phase, the authors only consider interactive ZKP protocols, which require $\lambda$ rounds. It would be natural to have the minimal round complexity; hence existing PVSSs decide to at least use NIZK instead so that all phases are noninteractive}. The flow of the PVSS is very simple: Participants generate the key pairs of the PKE and prove the validity of the public keys. When sharing the secrets, the dealer encrypts the shares and proves the validity of the encryption. Finally, the participants decrypt the shares and prove that the decrypted result is correct. For PKE, we will choose the lattice-based scheme of \cite{ACPS09}. Our remaining work is to construct the NIZKs without relying on the Random Oracle Model (ROM) while avoiding the inefficient Karp reduction to the Graph Hamiltonicity problem, and this is also the main technical challenge in this paper. 

To achieve this, we leverage the concept of trapdoor $\Sigma$-protocols (defined in~\cite{CCHLRRW19}). These allow us to design NIZKs for specific languages directly, without Karp reductions. We can then use the compiler of \cite{LNPT20} to transform these protocols into NIZKs with adaptive soundness and adaptive multi-theorem zero-knowledge properties. This approach reduces our task to the construction of appropriate $\Sigma$- protocols for three key languages. In the remainder of this section, we will show how to define the appropriate generic gap languages for the PVSS and provide the trapdoor $\Sigma$-protocol for these languages, instantiated with lattice-based schemes.\\

\noindent \textbf{Defining the Gap Languages for the Generic PVSS Construction.} Before going to the instantiation, we would like to construct a generic PVSS from NIZKs for \textit{gap languages} to capture possible instantiation from lattices. Hence, we first need to design suitable gap languages for key generation, sharing, and decryption, which we denote by $\LL^\Key,\LL^\Enc,\LL^\Dec$. Unlike the usual case of exact languages, defining gap languages involving public-key encryption systems to capture lattice-based instantiations (and other instantiations such as group-based) is quite non-trivial, as elaborated below.

For key generation, to capture lattice-based instantiation, we need to define an additional $\PKE.\Key\Verify$ algorithm to check whether a public-secret key pair is valid. The key verification algorithm should satisfy the following: If $(\pk,\sk) \leftarrow \PKE.\Key\Gen(\pp)$, then $\PKE.\Key\Verify(\pp,\pk,\sk)=1$. In addition, for any $(\pk,\sk)$ s.t. $\PKE.\Key\Verify(\pp,\pk,\sk)=1$\footnote{Here, the pair $(\pk, \sk)$ was not necessarily generated by $\PKE.\Key\Gen$, so this is a somewhat relaxed property.}, if $E$ is an encryption of $m$ using $\pk$, then $\PKE.\Dec(\pp,\pk,\sk,E)$ should return $m$\footnote{In particular, it holds for all $(\pk,\sk)$ honestly generated by $\PKE.\Key\Gen$. Thus we still capture the basic requirements of a PKE.}. An example for lattice-based schemes is that: For a public key $(\bA,\bb)$ and secret key $\bs$, then $\PKE.\Key\Verify$ checks whether~$||\bb-\bs^\top\cdot \bA||<B^{\Key\star}_\be$ for some bound $B^{\Key\star}_\be$, where $B^{\Key\star}_\be$ will be \textit{bigger} than the expected bound of $||\be||$. In the best case, an honest participant would prove that the keys are honestly generated, i.e., $(\pk,\sk)=\PKE.\Key\Gen(\pp,r)$ for some randomness $r$. In the worst case, we would require that: Even for dishonest participants who have passed verification, then $\PKE.\Key\Verify(\pp,\pk,\sk)=1$ so that encryption correctness is still ensured with such keys. (note that this does not necessarily mean that $(\pk,\sk)$ is generated from $\PKE.\Key\Gen$. Hence, there is a gap between the two sets of $\pk$. This situation is often seen in proving relaxed lattice-based relations, such as \cite{LNPT20} or \cite[Figure 2]{GHL22}). Thus we define the gap language $\LL^\Key$ s.t. $\LL^\Key_{zk}$ the set of is all tuples $(\pp,\pk)$ s.t. $(\pk,\sk)=\PKE.\Key\Gen(\pp,r)$ for some randomness $r$ and $\LL^\Key_{sound}$ is the set of all tuples $(\pp,\pk)$ s.t. $\PKE.\Key\Verify(\pp,\pk,\sk)=1$ for some secret key $\sk$. 

Next, we define the gap language $\LL^\Enc$ for correct sharing. At the very least, we would like to capture that: For any dealer that has passed verification, then all honest participants agree on some secret $s$. Now, define $\LL^\SSS_{n',t}$ to be the set of all valid shares $s_1,s_2,\dots,s_{n'}$ come from a Shamir secret sharing with threshold $t$ (see Subsection \ref{section-special-ss} for details). In the best case, we simply define $\LL^\Enc_{zk}$ to be the set of all $(\pp,n',t,(\pk_i,E_i)_{i=1}^{n'})$  s.t. i) $E_i=\PKE.\Enc(\pp,\pk_i,s_i,r_i)$, for some randomness $r_i$, ii) $(s_1~||~\dots~||~s_{n'}) \in \LL^\SSS_{n',t}$,  and iii) $(\pp,\pk_i) \in \LL^\Key_{sound}$ for all $i$. Although we define the last condition, the dealer does not need to prove that: The condition $(\pp,\pk_i) \in \LL^\Key_{sound}$ has been proved by participant $P_i$ previously. Thus, in the NIZK, we design it so that the dealer only needs to prove the former two conditions using witnesses $(s_i,r_i)_{i=1}^{n'}$. 

Now, defining $\LL^\Enc_{sound}$ will be the most challenging step. We need the language $\LL^\Enc_{sound}$ so that we can use the trapdoor to extract the randomness in the PKE to check if $E_i$ is a valid encryption of some message $m_i$. In our lattice-based instantiation (we will see later), while we have an easy trapdoor for key generation to directly extract $\sk$, we do not have the trapdoor to extract the randomness in the encryption process (even the relaxed ones). Hence, we need to define $\LL^\Enc_{sound}$ in another way, but still need to capture the properties we need in a PVSS. Fortunately, we can follow the idea of \cite{LNPT20}: In their definition, they define the language $\LL_{zk}$ to be the set of ciphertext that is honestly encrypted, while $\LL_{sound}$ to be the set of ciphertext s.t. decrypting them would provide valid message $m$ and witness~$f$ of small norm. Technically, while they cannot extract the randomness of the encryption process, they can still extract the witness $f$ in the decryption process. We might follow their idea to define our language  $\LL^\Enc_{sound}$: First, for participant~$P_i$ who passed the key verification process, we assume that~$(\pp,\pk_i) \in \LL^{\Key}_{sound}$ already and consider the secret keys $(\sk_i)_{i=1}^{n'}$ s.t. $((\pp,\pk_i),\sk_i) \in \RR^\Key_{sound}$ and assume they are unique. This is reasonable, as previously we forced participants to prove the existence of~$\sk_i$. The reason for uniqueness will be for the definitions in Subsection \ref{subsection-pvss-definition} to make sense and for the security proof. Now, we define $\LL^\Enc_{sound}$ to be the set of all $(\pp,n',t,(\pk_i,E_i)_{i=1}^{n'})$ such that: i) $(\pp,\pk_i) \in \LL^\Key_{sound}$ for all $i$  and ii) if we \textit{honestly} compute $s_i=\PKE.\Dec(\pp,\pk_i,\sk_i,E_i)$, then it holds that~$(s_1~||~\dots~||~s_{n'}) \in \LL^\SSS_{n',t}$, where $\sk_i$ is the corresponding witness of $(\pp,\pk_i)$ which can be extracted with a trapdoor as long as $(\pp,\pk_i) \in \LL^\Key_{sound}$. We now show that this captures our requirement in the worst case. Indeed, for a set of honest participants with public-secret keys $(\pk_i,\sk_i)$ and suppose $(\pp,n',t,(\pk_i,E_i)_{i=1}^{n'}) \in \LL^\Enc_{sound}$. Then recall that $s_i=\PKE.\Dec(\pp,\pk_i,\sk_i,E_i)$, then it holds that $(s_1~||~\dots~||~s_{n'}) \in \LL^\SSS_{n',t}$, or equivalently, $s_1,s_2,\dots,s_{n'}$ are valid shares of some secret $s$. Honest participants $P_i$ receive $s_i$ by executing $\PKE.\Dec$ and agree on $s$ due to the correctness of the secret sharing scheme. Thus, our definition of $\LL^\Enc_{sound}$ fully captures the fact that all honest participants agree on the same $s$ in the worst case.

Finally, we are left with the gap language $\LL^\Dec$ for correct decryption. For the motivation of our definition, we consider the LWE-based decryption protocol. To show correctness of decryption of such scheme given ciphertexts $(\bc_1,\bc_2,m)$, we must show that there exists $\bs,f$ s.t. $\bc_2-\bs^\top\cdot \bc_1=p\cdot m+f \pmod{q}$ for some small vector $f$. Hence, we have an additional witness $f$ when performing decryption and need to capture this situation. To do so, we define a set $\mathcal{W}^\Dec$ to capture the set of additional witnesses and let $\PKE.\Dec$ additionally output it. In other words, if~$E\leftarrow \PKE.\Enc(\pp,\pk,m)$, then it holds that $\PKE.\Dec(\pp,\pk,\sk,E)=(m,w)$ for some $w \in \mathcal{W}^\Dec$. The participant then uses $w$ as the witness to prove the validity of $m$. With this idea, we consider two sets $\mathcal{W}^\Dec_{zk}$ and $\mathcal{W}^\Dec_{sound}$ s.t. $\mathcal{W}^\Dec_{zk} \subseteq \mathcal{W}^\Dec_{sound}$. For $\LL^{\Dec}_{zk}$, we will let it to be the set of all $(\pp,\pk,E,m)$ such that there are~$(\sk,r,w)$ satisfying $(\pk,\sk)=\PKE.\Key\Gen(\pp,r)$, $(m,w)=\PKE.\Dec(\pp,\pk,\sk,E)$ and $w \in \mathcal{W}^\Dec_{zk}$. The language $\LL^{\Dec}_{sound}$ is also similar, except that we require~$\PKE.\Key\Verify(\pp,\pk,\sk)=1$ and $w \in \mathcal{W}^\Dec_{sound}$ instead. This definition of $\LL^\Dec$ does indeed generalize most of the existing well-known encryption schemes: For  RSA and group-based PKE such as ElGamal, then $\mathcal{W}^\Dec_{zk}=\mathcal{W}^\Dec_{zk}=\perp$, as there is no additional witness, while for LWE-based encryption scheme, $\mathcal{W}^\Dec_{zk}$ and $\mathcal{W}^\Dec_{sound}$ would be the set of all scalar $f$ s.t $|f|<B^{\Dec}_f$ and $B^{\Dec\star}_f$ respectively for some public bound $B^{\Dec}_f<B^{\Dec\star}_f$.  Thus, we have informally defined all the gap languages required for the PVSS. The formal definition of the gap languages, the PVSS construction, and its security proof will be given in Section \ref{section-generic-pvss}.\\

\noindent \textbf{Instantiating the NIZKs.} Now that we have described the gap languages, we will provide the instantiation for the languages. We choose the encryption scheme of \cite{ACPS09} over the one in \cite{GHL22} because to encrypt a scalar over $\mathbb{Z}_q$, the scheme of \cite{GHL22} must encode it into a vector in $\mathbb{Z}_q^\ell$ first, then encrypt the whole vector. Thus, the decryption complexity will be increased by a factor of $\ell$. Instead, we use the scheme of \cite{ACPS09}. It has the modulus $q=p^2$ and can directly encrypt a scalar over $\mathbb{Z}_p$ without any encoding process. Hence, the scheme of \cite{ACPS09} would be more efficient (see Subsection~\ref{subsection-pke} for details). In addition, while there are also many schemes outside \cite{GHL22,ACPS09}, only the particular modulus $q=p$ of \cite{GHL22} or $q=p^2$ of \cite{ACPS09} allows us to successfully design trapdoor $\Sigma$-protocols required for secret space $\mathbb{Z}_p$. Thus, we decided to choose the scheme of \cite{ACPS09} as our encryption scheme due to better efficiency. We also choose to directly encrypted the share $s_i$ by $E_i \leftarrow \PKE.\Enc(\pp,\pk_i,s_i)$ instead of using an amortized scheme like \cite{GHL22}. The reason will be discussed in Section \ref{subsection-final-remark}.

Now, we need to instantiate a trapdoor $\Sigma$-protocol for the generation of LWE keys for the scheme of \cite{ACPS09}. It has been described in \cite[Appendix G]{LNPT20}, but we will describe it as a warm-up and use its technique for the remaining two. More specifically, given a matrix $\bA \in \mathbb{Z}_q^{v\times u}$ and a vector $\bb \in \mathbb{Z}_q^u$, for $\LL^\Key_{zk}$, we would like to prove that there are $\bs \in \mathbb{Z}_q^v,\be \in \mathbb{Z}^u$ such that $\bb=\bs^\top\cdot \bA+\be^\top \pmod{q}$ and $||\bs||<B^\Key_{\bs},||\be||<B^\Key_{\be}$ for some public bounds  $B^\Key_{\bs},B^\Key_{\be}$. The trapdoor $\Sigma$-protocol for this is as follows. 
\begin{itemize}
\item Given $(\bA,\bb=\bs^\top\cdot \bA+\bff \pmod{q})$, the prover samples short vectors $\br,\bf$ and provide $\br^\top\cdot \bA+\be^\top \pmod{q}$ to verifier.
\item After receiving the challenge $c$ from verifier, prover provides the values $\bz=\br+c\cdot \bs$ and $\bt=\bff+c\cdot \be$ to verifier. 
\item Verifier accepts iff $\bz^\top\cdot \bA+\bt^\top=\bd+c\cdot \bb \pmod{q}$ and $||(\bz~||~\bt)||$ is short.
\end{itemize}
To generate the trapdoor, one uses the algorithm of \cite{MP12} to generate the trapdoor~$\bT$ and the matrix $\bA$. For the $\BadChallenge$ function, it uses the trapdoor $\bT$ to extract the witnesses $(\bs,\be)$ and returns the value $c$ s.t. $\bd+c\cdot \bb$ is not in the form of $\bz^\top\cdot \bA+\bt$ for some short $\bz,\bt$. In Appendix \ref{appendix-proof-of-decryption}, we will later prove that, if verifier accepts, then $(\bA,\bb) \in \LL^{\Key}_{sound}$ in the sense that: There exists $\bs,\be$ such that $\bb=\bs^\top\cdot \bA+\be^\top \pmod{q}$ and $||\bs||<B^{\Key\star}_{\bs},||\be||<B^{\Key\star}_{\be}$ for some public bound  $B^{\Key\star}_{\bs}>B^{\Key}_{\bs},B^{\Key\star}_\be>B^\Key_{\be}$.  The detailed construction of the protocol above will be presented in Subsection \ref{subsection-nizk-for-key-generation}.

 Our next step is to provide a trapdoor $\Sigma$-protocol for correct sharing. For a vector $\bmm=(m_1~||~m_2~||~\dots~||~m_{n'}) \in \mathbb{Z}_p^{n'}$ \textit{valid public keys} $(\bb_i)_{i=1}^{n'}$, we compute the encryptions $\bc_{1i}=\bA\cdot \br_i \pmod{q}$ and $\bc_{2i}=\bb_i\cdot \br_i+e_i+p\cdot m_i \pmod {q}$ for some short vectors $\br_i$ and small scalars $e_i$. In addition, we need that $m_1,\dots,m_{n'}$ are valid shares of some secret ($\bmm \in \LL^\SSS_{n',t}$). Now, the language $\LL^\Enc_{zk}$ is straightforward: It is the set of all $(\bA,n',t,(\bb_i,\bc_{1i},\bc_{2i})_{i=1}^{n'})$ s.t. $(\bA,\bb_i) \in \LL^{\Key}_{sound}$ and there exists~$(m_i,\br_i,e_i)_{i=1}^{n'}$ satisfy the above conditions. The language $\LL^\Enc_{sound}$ is straightforward as well: It is the set of all $(\bA,n',t,(\bb_i,\bc_{1i},\bc_{2i})_{i=1}^{n'})$ s.t $(\bA,\bb_i) \in \LL^{\Key}_{sound}$, and there exists $(m_i,f_i)_{i=1}^{n'}$ $\bc_{2i}-\bs_i^\top\cdot \bc_{1i}=p\cdot m_i+f_i \pmod{q}$, $|f_i|<B^{\Enc\star}_f$ and $\bmm \in \LL^\SSS_{n',t}$. Here $\bs_i$ is the corresponding secret key of $\bb_i$ given that $(\bA,\bb_i) \in \LL^\Key_{sound}$. By setting $B^{\Enc\star}_f$ big enough and let $\mathcal{W}^\Dec_{zk}$ to be the set of all $f$ s.t. $|f|<B^{\Enc\star}_f$, the description of $\LL^\Enc_{sound}$ fits the generic description we proposed previously when instantiated with the PKE of \cite{ACPS09}. We are now left with one final problem: Given a vector $\bm$, how do we check whether $\bmm \in \LL^\SSS_{n',t}$? Fortunately, for Shamir secret sharing scheme, it is known that there is a parity check matrix $\bH^{t}_{n'}$ such that $\bmm \in \LL^\SSS_{n',t}$ iff $\bmm^\top\cdot \bH^{t}_{n'}=\bzero \pmod{p}$  (see Subsection \ref{section-special-ss}). 
 
 With the language $\LL^\Enc$ above, we will design the trapdoor $\Sigma$-protocol based  on the technique of \cite{LNPT20} with minor modification to additionally check whether $\bmm^\top\cdot \bH^{t}_{n'}=\bzero \pmod{p}$. The scheme is straightforward as follows.
 \begin{itemize}
 \item Prover samples a short vector $(\bv~||~\bk)$ and a vector $~\bu=(u_1~||~u_2~||~\dots~||~u_n) \in \mathbb{Z}_p^n$ s.t. $\bu^\top\cdot \bH^{t}_n=\bzero \pmod{p}$. It parses $\bv=( \bv_1~|| \bv_2 ~||~ \dots ~||~ \bv_n)$, $\bk=( k_1~|| k_2 ~||~ \dots ~||~ k_n)$. Finally it  provides $\ba_{1i}=\bA\cdot \bv_i \pmod{q},~ \ba_{2i}=\bb_i\cdot \bv_i+k_i+p\cdot u_i \pmod{q}$ to verifier for all $1 \leq i \leq n$.
\item After receiving the challenge $c$ from verifier, prover provides the value $\bz_i=\bv_i+c \cdot \br_i,$ $h_i=k_i+c\cdot e_i,~$ $t_i=u_i+c\cdot m_i \pmod{p}$ to $\VV$ for all $1 \leq i \leq n$.
\item $\VV$ checks whether $t_i \in \mathbb{Z}_p$, $\bA\cdot \bz_i=\ba_{1i}+c\cdot \bc_{1i} \pmod{q},~$ $\bb_i\cdot \bz_i+h_i+p\cdot =\ba_{2i}+c\cdot \bc_{2i} \pmod{q}$. It also checks whether $||(\bz_i~||~h_i)||$ is small for all $1 \leq i \leq n$. Finally, check if $\bt^\top\cdot \bH^{t}_n=\bzero \pmod{p}$. Accept iff all checks pass.
 \end{itemize}
 The trapdoor $\bT$ is the same as the trapdoor $\Sigma$-protocol for key generation. For bad challenge function, we simply extract back $(\bs_i,\be_i)$ from the trapdoor $\bT$ to decrypt $(c\cdot \bc_{1i}+\ba_{1i},c\cdot \bc_{2i}+\ba_{2i})$ to receive $(m_i,f_i)$ for each $i \in \{0,1\}$. If $|f_i|>B^{\Enc\star}_f$ for some $i$ or $\bmm^\top \cdot \bH^{t}_n \neq  \bzero \pmod{p}$, we simply return $1-c$. Later, we will prove that, if $(\bA,\bb_i) \in \LL^\Key_{sound}$ with corresponding witness $\bs_i$ (extracted by $\bT$) for all $i \leq n'$, and if $(\bA,n',t,(\bb_i,\bc_{1i},\bc_{2i})_{i=1}^{n'})$ is accepted by verifier, then the honestly decrypted message $m_i$ from $(\bc_{1i},\bc_{2i})$ using $\bs_i$ must satisfy $\bmm \in \LL^\SSS_{n,t}$, which is what we want to capture in the worst case.  The detailed construction for the trapdoor $\Sigma$-protocol will be given in Subsection \ref{subsection-nizk-for-sharing} and its security proof will be given in Appendix \ref{appendix-proof-of-sharing}.

Finally, the trapdoor $\Sigma$-protocol for correct decryption is similar to key generation. For an instance $(\bA,\bb,\bc_1,\bc_2,m)$ one would have to prove the existence of $\bs,\be,f$ s.t. $\bb=\bs^\top\cdot \bA+\be^\top \pmod{q}$, $\bc_2-p\cdot m=\bs^\top\cdot \bc_1+f \pmod{q}$, and both $\be,f$ has small norm. The role of $\bc_2-p\cdot m$ is the same as $\bb$, the role of $\bc_1$ is the same as $\bA$, and the role of $\be$ is the same as $f$. Hence, one could easily design a trapdoor $\Sigma$-protocol for decryption based on the one for key generation. The language $\LL^\Dec_{zk}$ will be the set of all tuples $(\bA,\bb,\bc_1,\bc_2,m)$ s.t there exists~$\bs,\be,f$ above with $||\bs|| \leq B^\Dec_{\bs},||\be|| \leq B^\Dec_{\be}, |f| \leq B^\Dec_{f}$. The language~$\LL^\Dec_{sound}$ is similar, except that we require $||\bs|| \leq B^{\Dec\star}_{\bs},||\be|| \leq B^{\Dec\star}_{\be}, |f| \leq B^{\Dec\star}_{f}$ for some~$B^{\Dec\star}_{\bs}>B^{\Dec}_{\bs},B^{\Dec\star}_{\be}>B^{\Dec}_{\be}, B^{\Dec\star}_{f}>B^{\Dec}_{f}$. We see that this language definition easily captures the generic language $\LL^\Dec$ earlier when $f$ is the additional witness $w$, and $\mathcal{W}^\Dec_{zk}$ is the set of $f$ s.t. $|f|<B^{\Dec}_{f}=B^{\Enc\star}_{f}$, while $\mathcal{W}^\Dec_{sound}$ is the set of $f$ s.t.~$|f|<B^{\Dec\star}_{f}$. Due to the similarity of the trapdoor $\Sigma$-protocol for key generation, we will not describe the construction here. Instead, we refer the detailed construction to Subsection \ref{subsection-nizk-for-decryption}, and its security will be proven in Appendix \ref{appendix-proof-of-decryption}. 

The formal construction of the NIZKs will be described in Section \ref{section-nizk}. Having described the supporting NIZKs, we next plug the NIZK instantiations and the PKE into the generic PVSS to achieve a lattice-based PVSS in the standard model. This will be done in Section \ref{section-lattice-based-pvss}, with parameter setting in Subsection \ref{appendix-parameter-setting} and finally its complexities in Subsection \ref{section-complexity-analysis}. The security of the PVSS is implied by its generic version. Finally, in Appendix \ref{appendix-comparison-with-generic-sols}, we show that our construction is indeed more efficient than the generic technique of Karp reduction in \cite{GMW86}, hence we achieve the most efficient solution to the problem of post-quantum secure PVSS in the standard model so far.

\subsection{Related Works}
\noindent \textbf{Publicly Verifiable Secret Sharing.} Publicly verifiable secret sharing was initially introduced by Stadler to enable any party to publicly verify the correctness of both the sharing and reconstruction processes \cite{Stadler96}. This verification capability encompasses two key aspects: (i) during the sharing process, a public verifier, without knowledge of any shares or the secret, can validate the transcript produced by the dealer, and (ii) during the reconstruction process, a public verifier can verify the correctness of each participant's revealed share.

Most PVSS schemes, including those proposed in \cite{Stadler96,FO98,B99,Sch99,YY00,FS01,JVS14,CD20,GHL22,CDGK22,CDSV23,CD24} follow or derive from the GMW framework \cite{GMW86}, which involves encrypting the shares and employing Non-Interactive Zero-Knowledge (NIZK) arguments to prove: (i) the correctness of the encryption and (ii) that these encrypted shares constitute valid shares of some secret. During secret reconstruction, participants decrypt their shares and use additional NIZKs to prove decryption correctness. These NIZKs rely on the Fiat-Shamir heuristic and therefore achieve security in the random oracle model. 

Several constructions seek to minimize required assumptions by utilizing bilinear pairings or specific properties of the Paillier encryption scheme, thus avoiding dependence on random oracles \cite{AJ05,HV08,J11,CD17,DKI022}. Consequently, these PVSS schemes achieve security in the standard model. However, regardless of whether they avoid random oracles, all aforementioned constructions lack post-quantum security since their security relies minimally on discrete logarithm or factoring problems. 

Therefore, no existing PVSS construction simultaneously achieves both standard model security and post-quantum security. Our construction addresses this gap, and although it may not achieve the computational efficiency of previous constructions, it represents the first PVSS scheme proven secure in the standard model under post-quantum assumptions. We provide Table \ref{table:comparison-table-of-existing-work} to compare our construction with previous works. Note that the table does not include the costs of computing shares or Lagrange coefficients for secret reconstruction. Additionally, the computational costs presented represent the total time for secret sharing and subsequent reconstruction, including verification of the dealer's proof and $O(n)$ share decryptions. Although our construction may be less computationally efficient than previous schemes, its fundamental contribution lies in being the first PVSS protocol proven secure in the standard model under post-quantum assumptions.\\

\noindent \textbf{NIZK in the CRS Model and Trapdoor $\Sigma$-Protocols.} The majority of ZKP protocols employ the Fiat-Shamir heuristic to achieve NIZK in the random oracle model. However, \cite{GK03,KRS25} have shown that there exist protocols secure in ROM but insecure when the RO is instantiated with any standard hash function. Consequently, extensive research has been done to realize the Fiat-Shamir paradigm in the standard model. One such instantiation would be \textit{correlation intractable hash function} (CIHF), these hash functions are designed so that it would be infeasible to find inputs $x$ s.t. $(x,h(x)) $ belongs to specific relations.  A line of works \cite{KRR17,CCRR18,CCHJLRR18,HL18} focused on constructing CIHFs to soundly instantiate the Fiat-Shamir heuristic in the standard model. Canetti et al. \cite{CCHLRRW19} proved that it is possible to construct CIHFs for all efficiently searchable relations and proved that such a construction is sufficient to realize the Fiat-Shamir heuristic. Fortunately, the authors also introduced the notion of trapdoor $\Sigma$-protocols and provided a direct compiler of \textit{single-theorem} NIZK from any trapdoor $\Sigma$-protocol. Unfortunately, their NIZK construction requires the circular security of FHE, which is a somewhat less standard assumption. Peikert and Shiehian \cite{PS19} later constructed a CIHF from plain LWE only, effectively giving a NIZK compiler from standard assumptions in the CRS model. However, even requiring only plain LWE, the compiler only provides a single-theorem NIZK. While it is possible to additionally use the compiler of \cite{FLS99} to achieve a multi-theorem NIZK, the combined compiler of \cite{CCHLRRW19} and FLS \cite{FLS99} only works if the trapdoor $\Sigma$-protocols already are protocols for proving an NP-complete language~$\LL$. In this case, we have to take the \textit{inefficient} Karp reduction. \textit{If $\LL$ is not an NP-complete language, then in the FLS compiler, one might have to come up with a concrete trapdoor $\Sigma$-protocol} for \textit{OR proofs} that proves the validity of a statement in $\LL$ OR correct computation of PRG (with the witness as input). Unfortunately, we are not aware of any such solution so far, let alone their efficiency.  Recently, several constructions \cite{LNPT20,CX23} provided a \textit{direct multi-theorem NIZK compiler from any trapdoor $\Sigma$-protocols}. Among these, the compiler of \cite{CX23} is not post-quantum secure, while the compiler of \cite{LNPT20} is post-quantum secure but does not provide the adaptive soundness property for all languages. Fortunately, the compiler gives an adaptively sound NIZK for all \textit{trapdoor languages} (languages having a trapdoor for efficient membership check), which is sufficient in our application.

\begin{table}[h]
\vspace*{0.15cm}
\centering
	\caption{\footnotesize{Comparison with selected works. ``PQS'' refers to ``Post-Quantum Secure'' and  ``Std. Model'' refers to ``Standard Model'' (meaning that the construction does not rely on idealized models such as ROM). We denote $\mathbb{G}$ to be a cyclic group with order $q$ and assume that each element in $\mathbb{G}$ has $\log q$ bits. The notation $\text{op}_{\mathbb{G}}$ refers to the number of exponentiations in $\mathbb{G}$ and $\text{op}_{\mathbb{Z}_q}$ refers to the number of arithmetic operations in $\mathbb{Z}_q$. The cost $O(n \log^2 n)$ refers to computing Lagrange coefficients due to \cite{TCZAPGD20}. The cost of computing the shares is also $O(n \log^2 n)$ as well using the technique in \cite[Theorem 4.3]{KT73}. In the scheme of \cite{GHL22}, the authors omitted the cost of Lagrange interpolation; hence, the actual number of multiplications is $O(v^2+vn+n \log^2 n)$. The scheme uses Bulletproofs \cite{BBBPWM18} so the modulus $q$ is $2^{\Omega(\lambda) }$ and $v=\Omega(\lambda)^{1+\epsilon}$ for $\epsilon>0$. The use of Bulletproofs for NIZK instantiation also makes the scheme not post-quantum secure. For ours, we use the $\Omega$ notation because we are estimating the cost of the NIZKs via the cost of trapdoor $\Sigma$-protocols, meaning that the actual NIZK cost will be greater, as the compiler of \cite{LNPT20} has some complicated components that makes it hard to give the exact cost. Our PVSS has a factor $\lambda$ in both communication and computation cost due to using trapdoor $\Sigma$-protocols with binary challenges, thus we employ parallel repetition $\lambda$ times to achieve negligible soundness error. From the table, we see that our PVSS is the only construction that is post-quantum secure in the standard model, which is the advantage of our protocol over previous works}. }
	\label{table:comparison-table-of-existing-work}
    \vspace{0.2cm}
	\scalebox{0.74}{\begin{tabular}{lllclcc}
	\hline
	\hline
	\textbf{Work} & \textbf{Communication} & \textbf{Computation} & \textbf{Secret} & \textbf{Modulus} & \textbf{PQS} & \textbf{Std. Model.}  \\
	\hline
        \cite{Stadler96} & $O(n \cdot \log q)$ & $O(n^2\cdot \text{op}_{\mathbb{G}}+n \log^2 n\cdot  \text{op}_{\mathbb{Z}_q})$  & $\mathbb{G}$ & $q=O(2^{\lambda})$ & \xmark & \xmark\\
    \cite{Sch99} & $O(n \cdot \log q)$ & $O(n^2\cdot \text{op}_{\mathbb{G}}+n \log^2 n\cdot  \text{op}_{\mathbb{Z}_q})$  & $\mathbb{G}$ & $q=O(2^{\lambda})$ & \xmark & \xmark\\
   \cite{AJ05} & $O(n \cdot \log N)$ & $O(n^2\cdot \text{op}_{\mathbb{Z}_{N^2}}+n \log^2 n\cdot  \text{op}_{\mathbb{Z}_N})$  & $\mathbb{Z}_N$ & $N=O(2^{\lambda^3})$ & \xmark & \checkmark \\
    \cite{HV08} & $O(n \cdot \log q)$ & $O(n^2\cdot \text{op}_{\mathbb{G}}+n \log^2 n\cdot  \text{op}_{\mathbb{Z}_q})$ & $\mathbb{G}$ & $q=O(2^\lambda)$ & \xmark & \checkmark\\
   \cite{JVS14} & $O(n \cdot \log N)$ & $O(n^2\cdot \text{op}_{\mathbb{Z}_{N^2}}+n \log^2 n\cdot  \text{op}_{\mathbb{Z}_N})$  & $\mathbb{Z}_N$ & $N=O(2^{\lambda^3})$ & \xmark & \xmark \\
  \cite{CD17} & $O(n \cdot \log q)$ & $O(n\cdot \text{op}_{\mathbb{G}}+n \log^2 n\cdot  \text{op}_{\mathbb{Z}_q})$ & $\mathbb{G}$ & $q=O(2^\lambda)$ & \xmark & \checkmark \\
  \cite{CD20} & $O((n+\ell)\cdot  \log q)$ & $O(n^2\cdot \text{op}_{\mathbb{G}}+\ell\cdot n \log^2 n\cdot  \text{op}_{\mathbb{Z}_q}))$  & $\mathbb{G}^\ell$ & $q=O(2^\lambda)$ & \xmark & \xmark \\
  \cite{GHL22} & $O(n\cdot (u+ v) \cdot  \log q)$ & $O(v^2+vn+ n \log^2 n)\cdot \text{op}_{\mathbb{Z}_q}$ & $\mathbb{Z}_q$ & $q=O(2^\lambda)$ & \xmark & \xmark \\

\cite{CDGK22} & $O(n  \cdot \log q)$ & $O(n\cdot \text{op}_{\mathbb{G}}+n \log^2 n\cdot  \text{op}_{\mathbb{Z}_q})$ & $\mathbb{G}$ & $q=O(2^\lambda)$ & \xmark &\xmark \\
  
   \cite{CD24} & $O(n \log q)$ & $O(n\cdot \text{op}_{\mathbb{G}}+n \log^2 n\cdot  \text{op}_{\mathbb{Z}_q})$ & $\mathbb{Z}_q$ & $q=O(2^\lambda)$ & \xmark & \xmark\\ 
  \hline 
  Ours & $\Omega(n \lambda (u+ v)\cdot  \log q)$ & $\Omega(\lambda(n^2+ nuv)) \cdot \text{op}_{\mathbb{Z}_q} )$ & $\mathbb{Z}_q$ & $q=\tilde{O}(\lambda^{12}\cdot n)$ & \checkmark & \checkmark \\
	
	\hline
	\hline
	\end{tabular}}
    \vspace{0.2cm}
\end{table}

 \subsection{Structure of the Paper}
The rest of the paper is organized as follows. Section \ref{section-preliminaries} presents the preliminaries. Section \ref{section-generic-pvss} describes a generic PVSS based on an IND-CPA secure PKE and NIZKs for suitable gap languages. Section \ref{section-nizk} describes the supporting trapdoor $\Sigma$-protocols for the required gap languages of the PVSS we described in Section \ref{section-generic-pvss}. Finally, Section \ref{section-lattice-based-pvss} describes the lattice-based PVSS instantiation by plugging the PKE of  \cite{ACPS09} and the trapdoor $\Sigma$-protocols in Section \ref{section-nizk} into the generic PVSS in  Section \ref{section-generic-pvss}. We also provide the choice of parameters and complexity analysis of the lattice-based instantiation there. In Appendix \ref{appendix-comparison-with-generic-sols}, we compare our solution with the generic technique using Karp reduction. In Appendix \ref{appendix-pvss-security}, we present the security proof of the generic PVSS in Section \ref{section-generic-pvss}. In Appendix \ref{appendix-nizk-security}, we present the security proof of the supporting trapdoor $\Sigma$-protocols in Section \ref{section-nizk}. 
\section{Preliminaries}\label{section-preliminaries}

For $q \geq 2$, denote $\mathbb{Z}_q$ to be the ring of integers mod $q$.  Our work considers a prime $p$ and $q=p^2$. Unless specified otherwise, when performing modulo $q$, the result is an element in $\mathbb{Z}_q$. However, in some cases, we require the result to be an integer in $[-(q-1)/2,(q-1)/2]$ instead. We will notify the reader when this is the case. We use $x \leftarrow \mathcal{D}$ to denote that $x$ is sampled from a (not necessarily uniform) distribution $ \mathcal{D}$. We also use $x \xleftarrow{\$}
 \mathcal{S}$ to denote that $x$ is uniformly sampled from a \textit{set} $\mathcal{S}$. For a vector $\bx=(x_1,x_2,\dots,x_v) \in \mathbb{Z}^v$, we denote $||\bx||=\sqrt{\sum_{i=1}^v x_i^2}$ to be its norm. For column vectors $\bx,\by$, we denote $(\bx~||~\by)$ (instead of $(\bx^\top~||~\by^\top)^\top$) to denote the concatenated vector of $\bx$ and $\by$. For two matrices $\bA \in \mathbb{Z}^{v \times u},\bB \in \mathbb{Z}^{v \times w}$ having the same number of rows, we denote $\bC=\begin{bmatrix} \bA&|& \bB \end{bmatrix} \in \mathbb{Z}^{v \times (u+w)}$ to denote the concatenated matrix of $\bA,\bB$. We use $\negl(\lambda)$ to denote a negligible function in~$\lambda$. We use $[n]$ to denote $\{1,2,\dots,n\}$. For a language $\mathcal{L}$, we denote its corresponding binary relation by $\mathcal{R}$. We denote $(x,w) \in \mathcal{R}$ to say that an instance $x$ with witness~$w$ is in $\mathcal{R}$. Also, for a polynomial $p(X)$, we use deg$(P)$ to denote its degree.

 We denote $D_{\sigma}$ to be the \textit{continuous} Gaussian probability distribution with standard deviation $\sigma$. We also denote
$D_{\mathbb{Z}^n,\sigma,\be}$ to be the \textit{discrete} Gaussian probability distribution, assigning probability equal to $e^{-\pi \cdot ||\bx-\be||^2/\sigma^2}/(\sum_{\by \in \mathbb{Z}^v} e^{-\pi \cdot ||\by-\be||^2/\sigma^2})$ for each $\bx \in \mathbb{Z}^n$. When $\be=\bzero$, we simply write $D_{\mathbb{Z}^n,\sigma}$. 

We say that two distribution ensembles $\mathcal{X}_{\lambda},\mathcal{Y}_{\lambda}$ (depend on $\lambda$) are indistinguishable if for all PPT adversary $\Adversary,$ it holds that $\left|\condprob{\Adversary(x)=1} { x\leftarrow \mathcal{X}_{\lambda}}-\condprob{\Adversary(y)=1}{y \leftarrow \mathcal{Y}_{\lambda}}\right| \leq \negl(\lambda)$. It can be seen that if $\mathcal{X}_{\lambda},\mathcal{Y}_{\lambda}$ are indistinguishable and $\mathcal{Z}_{\lambda},\mathcal{Y}_{\lambda}$ are indistinguishable, then $\mathcal{X}_{\lambda},\mathcal{Z}_{\lambda}$ are also indistinguishable as well.

For a protocol with $n$ participants, we consider a \textit{static} adversarial model who can corrupt up to $t$ participants like many previous PVSS schemes \cite{Stadler96,Sch99,AJ05,J11,JVS14,CD17,CD20,DKI022,CDGK22,CDSV23,CD24} where $t$ is some positive integer less than $n/2$. The adversary is probabilistic polynomial time (PPT) and has access to a quantum computer. We also assume that all participants are given access to a broadcast channel so that once a message has been sent, it will be seen by anyone and cannot be deleted or modified.

Below, we state several lemmas for bounding the norm of a vector $\bx \leftarrow D_{\mathbb{Z}^v,\sigma}$.
\begin{lemma}[\cite{MR04}, Theorem 4.4 and \cite{LST18}, Lemma 2.1, Adapted]\label{lemma-bound}
For any $\sigma=\omega(v)$, it holds that
$$\condprob{||\bx||>\sigma\cdot \sqrt{v}}{\bx \leftarrow D_{\mathbb{Z}^v,\sigma}}< 2^{-\Omega(v)}.$$
\end{lemma}

\begin{lemma}\label{lemma-nextbound}
Let $S=\{c_1,c_2,\dots,c_k\} \subseteq [v]$ and $\bx_S=(x_{c_1}~||~x_{c_2}~||~\dots~||~x_{c_k})$. For any~$\sigma=\omega(k)$, then it holds that
$$\condprob{||\bx_S||>\sigma\cdot \sqrt{k}}{\bx \leftarrow D_{\mathbb{Z}^v,\sigma}}<2^{-\Omega(k)}.$$
\end{lemma}
\begin{proof}
Let $\rho_\sigma(\bx)=e^{-||\bx||^2/\sigma^2}$ and $\rho_\sigma(\mathbb{Z}^v)=\sum_{\bx \in \mathbb{Z}^v} \rho_\sigma(\bx)$. Then if $\bx=(x_1~||~x_2~||~\dots~||~x_v)$ it holds that $\rho_\sigma(x_1)\cdot \rho_\sigma(x_2) \dots \cdot \rho_\sigma(x_v) $. Thus $\rho_\sigma(\mathbb{Z}^v)=\sum_{x_1,x_2,\dots, x_v \in \mathbb{Z}} \rho_\sigma(x_1)\cdot \rho_\sigma(x_2) \dots \rho_\sigma(x_v)=\rho_\sigma(\mathbb{Z})^v$. Note that, for any $\br$, we have 
\begin{equation*}
\begin{aligned}
\Pr[\bx_S=\br]&=\sum_{\bx \in \mathbb{Z}^v,\bx_S =\br}
\rho_\sigma(\bx) /\rho_\sigma(\mathbb{Z}^v)=\sum_{\bx \in \mathbb{Z}^v,\bx_S=\br} \rho_\sigma(\bx_S)\cdot  \rho_\sigma(\bx_{[v] \setminus S})/\rho_\sigma(\mathbb{Z}^v)\\
&=\sum_{\bx' \in \mathbb{Z}^{v-k}} \rho_\sigma(\br)\cdot  \rho_\sigma(\bx')/\rho_\sigma(\mathbb{Z}^v)=\rho_\sigma(\br)\cdot \rho(\mathbb{Z}^{v-k})/\rho(\mathbb{Z}^{v})=\rho_\sigma(\br)/\rho(\mathbb{Z}^{k}).
\end{aligned}
\end{equation*}
The last equation is the probability that a vector $\br$ is returned when $\bx_S$ is sampled in $ D_{\mathbb{Z}^k,\sigma}$. Hence, when $\bx \leftarrow D_{\mathbb{Z}^v,\sigma}$, then $\bx_S $ is distributed according to $D_{\mathbb{Z}^k,\sigma}$. Thus, by applying Lemma \ref{lemma-bound} when $\bx_S$ is sampled in $ D_{\mathbb{Z}^k,\sigma}$, we get what we need.
\end{proof}


We also make use of the following important lemma, which will be used to prove the correctness and zero-knowledge property of our NIZKs in the later sections. 
\begin{lemma}[\cite{L12}, Theorem 4.6]\label{lemma-rejection-sampling}
Let $V$ be a subset of $\mathbb{Z}^u$ with norm less than $B$, and~$\sigma \in \mathbb{R}$ such that $\sigma=\omega(B\cdot \sqrt{\log u})$, and $h: V\rightarrow \mathbb{R}$ be a distribution. Then there exists a constant $M$ such that the distribution of the following algorithm $\Adversary:$
\begin{xenumerate}
\item Sample $\bv \leftarrow h,$ and $\bz \leftarrow D_{\sigma,\bv}^u,$
\item Output $(\bz,\bv)$ with probability $\min\left(\frac{D_{\sigma}^u(\bz)}{M\cdot D_{\sigma,\bv}^u(\bz)},1\right)$.
\end{xenumerate}
is within statistical distance  $\frac{2^{-\omega(\log u)}}{M}$ of  the distribution of the following algorithm $\Simulator$:
\begin{xenumerate}
\item Sample $\bv \leftarrow h,$ and $\bz \leftarrow D_{\sigma}^u,$
\item Output $(\bz,\bv)$ with probability $1/M$.
\end{xenumerate} The probability that $\Adversary$ outputs something is at least $\frac{1-2^{-\omega(\log u)}}{M}$. If $\sigma=\alpha B$ for any positive $\alpha$ then $M=e^{1/\alpha+12/\alpha^2},$ the output of $\Adversary$ is within statistical distance $\frac{2^{-100}}{M}$ of the output of $\Simulator$, and the probability that $\Adversary$ outputs something is at least $\frac{1-2^{-100}}{M}$.
\end{lemma}

\subsection{Lattice Assumption and Trapdoor}

We recall some preliminaries from lattice-based cryptography. We now recall the learning with error assumption \cite{Reg05,Re10}.
 It is as follows.
\begin{definition}[The LWE Assumption]\label{definition-lwe}
Let $m,v,q$ be integers. Let $\bs \in \mathbb{Z}_q^v$ be chosen from a distribution. Then for any PPT adversary $\Adversary,$ it holds that 

$$\left| \condprob{b=b'}{\bA\uniform \mathbb{Z}_q^{v\times u},\be \leftarrow D_{\mathbb{Z}^u,\alpha q}, b \uniform \{0,1\}, \\
\text{If $b=0$, $\bb= \bs^\top\bA+\be^\top$, else $\bb\uniform \mathbb{Z}_q^u$,}\\
b' \leftarrow \Adversary(\bA,\bb)\\
}-\dfrac{1}{2} \right| \leq \negl(\lambda).$$
\end{definition}
 It is shown by \cite{Reg05}, for any $q$ such that $\alpha \cdot q >\sqrt{v}$ and when $\bs \uniform \mathbb{Z}_q^v$ or $\bs \leftarrow D_{\mathbb{Z}^v,\alpha q}$, then breaking LWE is as hard as quantumly solving $\mathsf{GapSVP}_{v/\alpha}$. As pointed out by \cite{BLMR13,P16,Lyu24},  the LWE (and $\mathsf{GapSVP}_{v/\alpha}$) problem seems  to be intractable as long as~$v/\alpha=\tilde{O}(2^{v^\epsilon})$ for fixed $0<\epsilon<1$, even for quantum computers.

We describe the lattice trapdoor in \cite{MP12}, which generates a matrix $\bA$ statistically close to uniform and a trapdoor matrix $\bT$ used to invert the LWE function. 
\begin{lemma}[\cite{MP12}, Theorem 5.1]\label{theorem-invert-lwe} There exists algorithms $(\TrapGen,\Invert)$ where $\Invert$ is \textbf{deterministic}, such that, for any $u,u',v$ satisfying $v>1, u>u'= v \lceil \log_2 q \rceil$ and $u> u'+v \log q+\omega(\log v),$ perform the following:
\begin{itemize}
\item $\TrapGen(1^\lambda,v,u):$ On input the security parameter $\lambda$, outputs a matrix $\bA \in \mathbb{Z}_q^{v\times u}$ and a trapdoor $\bT \in \mathbb{Z}_q^{u \times u'}$ such that the distribution of $\bA$ is statistically close to $U(\mathbb{Z}_q^{v \times u})$ with distance at most $2^{-v}$.
\item $\Invert(\bA,\bT,\bb):$ Let $g: \mathbb{Z}_q^v \times \mathbb{Z}^u \rightarrow \mathbb{Z}_q^u$ as $ g(\bs,\be)=\bs^\top\cdot \bA+\be^\top \pmod{q}$. For any vector $\bb$ satisfying $\bb=g(\bs,\be)$ for some $\bs \in \mathbb{Z}_q^v,\be \in \mathbb{Z}^u$ s.t. $||\be||=O(q/\sqrt{v\log q})$, correctly \textbf{inverts} $(\bs,\be)$. In particular, the value $(\bs,\be)$ is unique, and $g(\bs,\be)$ is injective when $||\be|| = O(q/\sqrt{v \log q})$.
\end{itemize}
\end{lemma}

\subsection{Public Key Encryption and the ACPS Encryption Scheme}\label{subsection-pke}
This section recalls the formal definition of public key encryption (PKE) and its security properties. We also recall the ACPS PKE of \cite{ACPS09}.  Our syntax differs slightly from an ordinary PKE, as we have an additional $\PKE.\Key\Verify$ algorithm to check the validity of key pairs and a set $\mathcal{W}^\Dec$ to denote the set of additional witnesses required when performing decryption. For our PVSS application, we will require such a modification. The rest of the syntax is the same as an ordinary PKE.
\begin{definition}[PKE]
Let $\mathcal{M}$ be the message space, and $\mathcal{C}$ be the ciphertext space. Let $\RAND^\Key$ and $\RAND^\Enc$ be the randomness space for key generation and encryption, respectively. Let $\mathcal{W}^\Dec$ denote the additional witness space of the decryption algorithm. A public key encryption is a tuple $\PKE=(\PKE.\Setup,\PKE.\Key\Gen,\PKE.\Key\Verify,\PKE.\Enc,\PKE.\Dec)$, specified as follows.
\begin{itemize}
\item $\PKE.\Setup(1^\lambda) \rightarrow \pp:$ On input a security parameter $\lambda$, this PPT algorithm outputs a public parameter $\pp$.
\item $\PKE.\Key\Gen(\pp) \rightarrow (\pk,\sk):$ On input a public parameter $\pp$, this PPT algorithm outputs a public-secret key pair $(\pk,\sk)$. Sometimes we denote $(\pk,\sk)=\PKE.\Key\Gen(\pp,r)$ to denote that $(\pk,\sk)$ is computed deterministically using some randomness $r \in \RAND^\Key$.
\item $\PKE.\Key\Verify(\pp,\pk,\sk):$ On input public parameter $\pp$, a public-secret key pair $(\pk,\sk)$, this algorithm checks the validity of $(\pk,\sk)$.
\item $\PKE.\Enc(\pp,\pk,m) \rightarrow E:$ This algorithm is executed by the encryptor. On input a public parameter $\pp$, a public key $\pk$, and a message $m \in \mathcal{M}$, this PPT algorithm outputs a ciphertext $E \in \mathcal{C}$. Sometimes we denote $C=\PKE.\Enc(\pp,\pk,m,r)$ to denote that $C$ is deterministically computed from $\PKE.\Enc$ with inputs $\pk,m,r$ where $r$ is the randomness $\in \RAND^\Enc$.
\item $\PKE.\Dec(\pp,\pk,\sk,E) \rightarrow (m,w):$ This algorithm is executed by the decryptor. On input a public parameter $\pp$, a public-secret key pair $(\pk,\sk)$ and ciphertext $E \in \mathcal{C}$, this algorithm outputs the original message $m$ and some additional witness $w \in \mathcal{W}^\Dec$ (possibly $w=\perp$). 
\end{itemize}
\end{definition}
We now present the required security of a PKE below. We require two correctness properties, which together imply the correctness of an ordinary PKE.
\begin{definition}[Key Correctness]\label{key-correctness}  We say that $\PKE$ achieves \textbf{key correctness} if for all $m$ and $\pp \leftarrow \PKE.\Setup(1^\lambda)$ and $(\pk,\sk) \leftarrow \PKE.\Key\Gen(\pp)$, it holds that $\PKE.\Verify(\pp,\pk,\sk)=1$ with overwhelming probability.
\end{definition}

\begin{definition}[Encryption Correctness]\label{definition-encryption-correctness}  We say that $\PKE$ achieves \textbf{encryption correctness} if for all $m\in \mathcal{M}$, $\pp \leftarrow \PKE.\Setup(1^\lambda)$ and all pairs $(\pk,\sk)$ satisfying $\PKE.\Key\Verify(\pp,\pk,\sk)=1$, if $E \leftarrow \PKE.\Enc(\pp,\pk,m)$, then it holds that $\PKE.\Dec(\pp,\pk,\sk,E)=(m,w)$ for some $w \in \mathcal{W}^\Dec$  (consequently, if $(\pk,\sk) \leftarrow \PKE.\Key\Gen(\pp)$, and $E \leftarrow \PKE.\Enc(\pp,\pk,m)$, then $\PKE.\Dec(\pp,\pk,\sk,E)=(m,w)$ with overwhelming probability).
\end{definition}
\begin{definition}[Multi-key IND-CPA Security \cite{CDGK22}]  We say that $\PKE$ achieves \textbf{multi-key IND-CPA security} if for all PPT algorithm $\Adversary$,  it holds that $\Adv^{\mathsf{IND-CPA}}(\Adversary)=|\Pr[\GAME^{\mathsf{IND-CPA}}_0(\Adversary,t)=1]-\Pr[\GAME^{\mathsf{IND-CPA}}_1(\Adversary,t)=1]| \leq \negl(\lambda)$, where  $\GAME^{\mathsf{IND-CPA}}_b(\Adversary,t)$ in Figure \ref{figure-ind-cpa}.
\begin{framedfigure}[ The game $\GAME^{\mathsf{Gen-IND-CPA}}_b(\Adversary,t)$. \label{figure-ind-cpa}]
\vspace{0.05cm}
$\pp \leftarrow \PKE.\Setup(1^\lambda).$ For all $1\leq i\leq t, (\pk_i,\sk_i) \leftarrow \PKE.\Key\Gen(\pp)$,\\
$((m^0_i)_{i=1}^t,(m^1_i)_{i=1}^t)\leftarrow \Adversary(\pp,(\pk_i)_{i=1}^t)$. For all $1\leq i \leq t$, $E_i^\star \leftarrow \PKE.\Enc(\pp,\pk_i,m_i^b)$,\\
$b' \leftarrow \Adversary(\pp,(\pk_i)_{i=1}^t,(m^0_i)_{i=1}^t,(m^1_i)_{i=1}^t,(E^\star_i)_{i=1}^t)$,\\
Return $b'$.
\vspace{0.1cm}
\end{framedfigure}

\end{definition}
The case $t=1$ is just the ordinary IND-CPA security game. It is remarked by~\cite{CDGK22} that, a PKE achieves multi-key IND-CPA iff it achieves IND-CPA.

We now describe the encryption of \cite{ACPS09}. We choose the scheme of \cite{ACPS09} over \cite{GHL22} for better efficiency. Indeed, to encrypt a scalar in $\mathbb{Z}_q$, the scheme of \cite{GHL22} has to encode the scalar into a vector in $\mathbb{Z}_q^\ell$ and then encrypt the whole vector. On the other hand, the scheme of \cite{ACPS09} can encrypt scalars directly, and i) the ciphertext of \cite{ACPS09} has smaller size $O(v\cdot \log q)$ compared to $O(v+\ell)\cdot \log q$ of \cite{GHL22} for some tradeoff parameter~$\ell \geq 2$, ii) the encryption of \cite{ACPS09} requires only $O(u\cdot v)$ multiplications while the scheme of~\cite{GHL22} requires $O(u\cdot v+u\cdot \ell)$ multiplications,  and iii) the decryption process of \cite{ACPS09} only needs $O(v)$ multiplications, while \cite{GHL22} requires $O(v\cdot \ell)$ multiplications \footnote{This will be true if the value $u$ has the same asymptotic complexity $\Theta(v \log q)$ and $q=\text{poly}(v)$ in both schemes. The scheme of \cite{GHL22} considers $u=v$, while the scheme of~\cite{ACPS09} requires $u=\Omega(v\log q)$ for security. However, our trapdoor $\Sigma$-protocols use the algorithm $\TrapGen$ of Lemma \ref{theorem-invert-lwe} to generate a matrix $\bA$ for achieving CRS indistinguishability. Recall that the matrix $\bA$ requires $u=\Omega(v\log q)$ in the lemma. Hence, even if we apply the scheme of \cite{GHL22} into our PVSS, we still need $u=\Theta(v \log q)$ anyway, and in this case, we use the scheme of \cite{ACPS09} for better efficiency.}. Thus, the scheme of \cite{ACPS09} has better overall efficiency.

We refer the reader to \cite{ACPS09} for its formal security proof. Consider parameters $u,v,\alpha,\beta,q,r,B^{\Key\star}_\bs,B^{\Key\star}_\be$ s.t. $B^{\Key\star}_\be>\sqrt{u} \alpha q$ and $B^{\Key\star}_\bs> \sqrt{v} \alpha q$. The values~$B^{\Key\star}_\bs,B^{\Key\star}_\be$ will be specified later in Subsection \ref{subsection-nizk-for-key-generation}. Also, consider randomness sets $\RAND^\Key=\{\be \in \mathbb{Z}^u~|~||\be||<\sqrt{u}\cdot \alpha q\}$, $\RAND^\Enc=\{\br \in \mathbb{Z}^u~|~||\br||<\sqrt{u}\cdot r\}$. The set $\mathcal{W}^\Dec$ can be taken as  $\mathcal{W}^\Dec=\{f \in \mathbb{Z}~|~|f|< \sqrt{u} \cdot r \cdot B^{\Key\star}_\be+\sqrt{v}\cdot \beta q\}$. Now, the encryption scheme of \cite{ACPS09} is described as follows.
\begin{itemize}
\item $\PKE.\Setup(1^\lambda):$ Generate $\bA \uniform \mathbb{Z}_q^{v \times u}$, Return $\pp=(\bA,u,v,\alpha,\beta,B^{\Key\star}_\bs,B^{\Key\star}_\be)$.

\item $\PKE.\Key\Gen(\bA):$ Sample $\bs \leftarrow D_{\mathbb{Z}^v,\alpha q},~\be \leftarrow D_{\mathbb{Z}^u,\alpha q}$. Repeat until $||\bs||<\sqrt{v} \cdot \alpha q$, $||\be||<\sqrt{u} \cdot \alpha q$. Compute $\bb=\bs^\top\cdot\bA+\be^\top \pmod{q}$. 
Return $(\pk,\sk)=(\bb,\bs)$.
\item $\PKE.\Key\Verify(\bA,\bb,\bs):$ Check if $||\bb-\bs^\top\cdot \bA \pmod{q}|| \leq B^{\Key\star}_\be$ and $||\bs|| \leq B^{\Key\star}_\bs$. Return $1$ iff it holds.
\item $\PKE.\Enc(\bA,\bb,m):$ To encrypt a message $m \in \mathbb{Z}_p$, sample $\br \leftarrow D_{\mathbb{Z}^u,r}, e \leftarrow D_{\mathbb{Z},\beta q}$ and compute $\bc_1=\bA\cdot \br \pmod{q},~\bc_2=\bb\cdot\br+e+p\cdot m \pmod{q}$, where $\beta=\sqrt{u}\cdot \log u \cdot (\alpha+\frac{1}{2\cdot q})$. Return $(\bc_1,\bc_2)$.
\item $\PKE.\Dec(\bA,\bb,\bs, (\bc_1,\bc_2)):$ Compute $f=\bc_2-\bs^\top\cdot \bc_1 \pmod{p}$ and cast $f$ as an integer in $[-(p-1)/2,(p-1)/2]$. Finally compute the message $m=(\bc_2-\bs^\top\cdot \bc_1-f)/p \pmod{p}$. Return $(m,f)$, where $f$ is the additional witness.
\end{itemize}
 Due to Lemma \ref{lemma-bound}, it holds that $||\bs||<\sqrt{v} \cdot \alpha q$ and $||\be||<\sqrt{u} \cdot \alpha q$ with overwhelming probability. Thus, in $\PKE.\Gen$, we only have to generate $\bs,\be$ one time with overwhelming probability as well, although we have to write the ``repeat'' part to ensure that their norm is guaranteed to be bounded. In addition, the error $f$ when decrypting does not exceed $|f|\leq ||\be^\top\cdot \br+e||\leq ||\be||\cdot||\br||+||e||$. Thus, as long as $B^{\Key\star}_\bs>\sqrt{v} \cdot \alpha q$, $B^{\Key\star}_\be>\sqrt{u} \cdot \alpha q$ and $B^{\Key\star}_\be\cdot \sqrt{u}\cdot r+\beta\cdot q<p/2$, then  the correct encryption property holds and $f \in \mathcal{W}^\Dec$. This bounds $B^{\Key\star}_\bs,B^{\Key\star}_\be$ right now will be left open, as we will specify it later in Subsection \ref{subsection-nizk-for-key-generation}. The security of the encryption scheme is proved via the LWE assumption in Definition \ref{definition-lwe}. 

\subsection{Trapdoor \texorpdfstring{$\Sigma$}\texorpdfstring{-}Protocols}
In this section, we recall the formal definition of trapdoor $\Sigma$-protocols and their security properties. We consider gap languages $\LL=(\LL_{zk},\LL_{sound})$. Clearly, to achieve the correctness and soundness property, we actually need $\LL_{zk} \subseteq \LL_{sound}$.
\begin{definition}[$\Sigma$-Protocols]
Let $\LL=(\LL_{zk},\LL_{sound})$ be a gap language, let $\RR=(\RR_{zk},\RR_{sound})$ be its corresponding relation. A \textbf{$\Sigma$-protocol} for $\LL$, denoted by $\Sigma$, is an interactive proof between a prover and a verifier with the following syntax:
\begin{itemize}
\item $\Sigma.\Gen_{\pp}(1^\lambda) \rightarrow \pp:$ On input the security parameter $\lambda$, this algorithm outputs public parameter $\pp$. 
\item $\Sigma.\Gen_{\LL}(\pp,\LL) \rightarrow \crs:$ This is a setup algorithm by a third party. On input the public parameter $\pp$ and the description of the language $\LL$, it returns a common reference string $\crs$.
\item $\Sigma.\Prove \langle \PP(\pp,\crs,x,w), \VV(\pp,\crs,x) \rangle:$ This is an interactive protocol, where both parties have a public parameter $\pp$ from $\Gen_{\pp}$, a common reference string~$\crs$ from $\Gen_{\LL}$ and statement $x$, in addition, the prover holds a corresponding witness $w$ of $x$. At the end of the interaction, the verifier outputs a bit $b$.
\end{itemize}
In addition, a $\Sigma$-protocol satisfies the following properties:
\begin{itemize}
\item $3$-Move Form: The protocol $\Sigma.\Prove$ has the following form: The prover outputs a first message $\msg_1\leftarrow \PP(\crs,x,w)$, the verifier then responds with a challenge $c$ from the challenge space, and finally, the prover outputs a second message $\msg_2\leftarrow \PP(\crs,x,w,\msg_1,c)$. Finally, the verifier, outputs a bit $b \leftarrow \VV(\crs,x,\msg_1,c,\msg_2)$.
\item Correctness: For any $\pp \leftarrow \Sigma.\Gen_\pp(1^\lambda)$, $\crs \leftarrow \Sigma.\Gen_\LL(\pp,\LL)$ and for any~$(x,w) \in \RR_{zk}$, if prover provides a valid response $(\msg_1,\msg_2)$, then at the end of $\Sigma.\Prove$, verifier returns $1$ with overwhelming probability. 
\item Special Soundness: For any $\pp \leftarrow \Sigma.\Gen_\pp(1^\lambda)$, $\crs \leftarrow \Sigma.\Gen_\LL(\pp,\LL)$ and for~$x \not \in \LL_{sound}$, and for any first message $\msg_1$, there exists at most one challenge $c=f(\crs,x,\msg_1)$ such that there might exists some second message $\msg_2$ that causes $\VV$ to return $1$. The function $f$ is the bad challenge function. In other words, if $x \not \in \LL_{sound}$ and the challenge is equal to $c$, then the verifier might or might not accept, but if the challenge is not equal to $c$, then the verifier never accepts.
\item Special Zero-knowledge: There exists a simulator $\Simulator,$ on input the common reference string $\crs$, an instance $x \in \LL_{zk}$ and challenge $c$, outputs a simulated transcript $(\msg_1,c,\msg_2)$ that is computationally indistinguishable from a real transcript of $\Sigma.\Prove$.
\end{itemize}
\end{definition}

\begin{definition}[Trapdoor $\Sigma$-Protocols, from \cite{CCHLRRW19,LNPT20}]
Let $\LL=(\LL_{zk},\LL_{sound})$ be a language with corresponding relation $\RR=(\RR_{zk},\RR_{sound})$. A \textbf{trapdoor $\Sigma$-protocol} for $\LL$, denoted by $\Trap\Sigma$, with a bad challenge function $f$ is a $\Sigma$-protocol with two additional algorithms $\Trap\Sigma.\TrapGen$ and $\Trap\Sigma.\BadChallenge$ as follows
\begin{itemize}
\item $\Trap\Sigma.\TrapGen(\pp,\LL,\tr_\LL) \rightarrow (\crs,\tr):$ On input a security parameter $\lambda$, and a trapdoor $\tr_\LL$ depending on $\LL$, this PPT algorithm returns a common reference string $\crs$ and a trapdoor $\tr$.
\item $\Trap\Sigma.\BadChallenge(\crs,x,\msg,\tr) \rightarrow c:$ On input a common reference string $\crs,$ an instance $x$, the first message $\msg$ of the prover and a trapdoor $\tr$, this deterministic algorithm returns a bit $c$.
\end{itemize}
In addition, the two algorithms must satisfy the following properties:
\begin{itemize}
\item CRS Indistinguishability: For any PPT adversary $\Adversary$, any $\pp \leftarrow \Sigma.\Gen_{\pp}(1^\lambda)$ and trapdoor $\tr_\LL$ for the language $\LL$, it holds that:
\begin{equation*}
\begin{aligned}
&|\condprob{\Adversary(\crs)=1}{\crs \leftarrow \Trap\Sigma.\Gen_\LL(\pp,\LL)}-\\
&~~~\condprob{\Adversary(\crs)=1}{(\crs,\tr) \leftarrow \Trap\Sigma.\TrapGen(\pp,\LL,\tr_\LL)} | \leq \negl(\lambda)
\end{aligned}
\end{equation*}
\item Correctness: For any $x \not \in \LL_{sound},$ $\pp \leftarrow \Sigma.\Gen_{\pp}(1^\lambda)$ and  $(\crs,\tr) \leftarrow \Trap\Sigma.\TrapGen(\pp,\LL,\tr_{\LL})$ it holds that $\Trap\Sigma.\BadChallenge(\crs,x,\msg,\tr)=f(\crs,x,\msg)$. Equivalently, $\BadChallenge$ returns a challenge $c$ such that if the verifier's challenge is not equal to $c$, then no second prover message could make the verifier accept.
\end{itemize}
\end{definition}

\subsection{Non-Interactive Zero-Knowledge Arguments}
In this section, we present the formal definition of non-interactive zero-knowledge arguments (NIZK) and their security properties. The language for the NIZK might additionally consist of a trapdoor $\tau$ for checking whether an element is in $\LL_{sound}$. While in a general syntax, this is not needed, however, in several instantiations such as \cite{LNPT20}, the language requires such a trapdoor to achieve adaptive soundness. Thus, we use $[\tau]$ to denote that $\tau$ might be optionally included in the protocol.

\begin{definition}[NIZK, Adapted from \cite{LNPT20}]
Let $\LL=(\LL_{zk},\LL_{sound})$ be a gap language with corresponding relation $\RR=(\RR_{zk},\RR_{sound})$. A NIZK argument for $\LL$ is a tuple $\NIZK=(\NIZK.\Gen_{\pp},\NIZK.\Gen_{\LL},$ $\NIZK.\Prove,$ $\NIZK.\Verify)$, specified as follows.
\begin{itemize}
\item $\NIZK.\Gen_{\pp}(1^\lambda) \rightarrow \crs:$ On input a security parameter $\lambda$, this PPT algorithm outputs a public parameter $\pp$.
\item $\NIZK.\Gen_{\LL}(\pp,\LL,[\tau]) \rightarrow \crs:$ On input a public parameter $\pp$, the description of $\LL$ which might contain a trapdoor $\tau$ for checking membership in $\LL_{sound}$, this PPT algorithm outputs the language-dependent part $\crs_{\LL}$ of the common reference string $\crs=(\pp,\crs_{\LL})$.
\item $\NIZK.\Prove(\crs,x,w) \rightarrow \pi:$ This is an algorithm executed by the prover. On input a common reference string $\crs$, a statement $x$, and a witness $w$, this algorithm outputs a proof $\pi$.
\item $\NIZK.\Verify(\crs,x,\pi) \rightarrow 0/1:$ This is an algorithm executed by the verifier. On input a common reference string $\crs$, a statement $x$, and a proof $\pi$, this algorithm outputs a bit $b \in \{0,1\}$ which certifies the validity of $(x,\pi)$.
\end{itemize}
\end{definition}

For simplicity, we denote an algorithm $\NIZK.\Setup(1^\lambda,\LL)$ to execute $\NIZK.\Gen_{\pp}$, then $\NIZK.\Gen_{\LL}$ to generate a common reference string $\crs$. The algorithm might also optionally take the trapdoor $\tau$ for $\LL$ as an input, but for simplicity we omit it here. We now present the required security properties of an NIZK below.

\begin{definition}[Correctness]\label{definition-correctness}
We say that $\NIZK$ achieves \textbf{correctness} if for all $\crs \leftarrow \NIZK.\Setup(1^\lambda,\LL)$ and $(x,w) \in \RR_{zk}$, if $\pi \leftarrow \NIZK.\Prove(\crs,x,w)$, then it holds that $\NIZK.\Verify(\crs,x,\pi)=1$ with probability $1-\negl(\lambda)$.
\end{definition}

\begin{definition}[Adaptive Soundness]\label{definition-adaptive-soundness} We say that $\NIZK$ achieves \textbf{adaptive soundness} if for all PPT adversary $\Adversary$, it holds that
$$\condprob{x \not \in \LL_{sound}~\land\\
\NIZK.\Verify(\crs,x,\pi)=1
}{\crs \leftarrow \NIZK.\Setup(1^\lambda,\LL),\\
(x,\pi)\leftarrow \Adversary(\crs)\\} \leq \negl(\lambda).$$
\end{definition}

\begin{definition}[Adaptive Multi-theorem Zero-Knowledge] We say that $\NIZK$ achieves \textbf{adaptive multi-theorem zero-knowledge} if there exists a PPT simulator $\Simulator_{\NIZK}=(\Simulator_{\crs},\Simulator_{\pi})$ such that for all PPT adversaries $\Adversary$, it holds that
\begin{equation*}
\begin{aligned}
&\Adv^{\mathsf{ZK}}(\Adversary)=|\condprob{b=1}{\crs\leftarrow \NIZK.\Setup(1^\lambda,\LL),~
b \leftarrow \Adversary^{\mathsf{P}(\crs,.,.)}(\crs)}\\
&\quad\quad\quad\quad\quad-\condprob{b=1}{(\crs,\rho) \leftarrow \Simulator_{\crs},~
b \leftarrow \Adversary^{\mathcal{O}(\crs,\rho,.,.)}(\crs)}|\leq \negl(\lambda),
\end{aligned}
\end{equation*}
where the oracle $\mathsf{P}(\crs,.,.)$ on input $(\crs,x,w),$ outputs $\perp$ if $(x,w) \not \in \RR_{zk},$ otherwise outputs $\pi \leftarrow \NIZK.\Prove(\crs,x,w)$. The oracle $\mathcal{O}(\crs,\rho,.,.)$, on input $(\crs,\rho,x,w),$ outputs $\perp$ if $(x,w) \not \in \RR,$ otherwise outputs $\pi \leftarrow \Simulator_{\pi}(\crs,\rho,x)$.
\end{definition}
The work of \cite{CCHLRRW19} constructed a NIZK in the CRS model for any trapdoor-$\Sigma$ protocol from a PKE and a CIHF. Note that CIHF can be constructed from the LWE assumption, according to \cite{PS19}. However, the construction is only single-theorem ZK and requires the transformation of \cite{FLS99} to achieve multi-theorem ZK. The transformation requires proving the Graph Hamiltonicity problem, which is highly inefficient (see Appendix \ref{appendix-comparison-with-generic-sols} for the case of LWE).   Recently, \cite{LNPT20} provided a compiler that transforms any trapdoor $\Sigma$-protocol into an adaptive soundness and multi-theorem zero-knowledge NIZK (at least, adaptive soundness holds for trapdoor languages) \textit{without} needing to use Karp reduction.

\subsection{Shamir Secret Sharing Scheme}\label{section-special-ss}
We recall the Shamir secret sharing scheme \cite{Sha79} and define the language $\LL^{\SSS}_{n,t}$ consisting of all vectors $\bs=(s_1~||~s_2~||~\dots~||~s_n)$ that are valid shares of a Shamir secret sharing scheme with threshold $t$ (recall that we have mentioned them in Subsection \ref{section-technical-overview}). First, the syntax of the Shamir secret sharing scheme is as follows.

\begin{itemize}
\item $\SSS.\Share(s,n,t):$ Chooses a polynomial $p(X) \in \mathbb{Z}_p[X]$ of degree $t$. Compute the share $s_i=p(i) \pmod{p}$. Return $(s_i)_{i=1}^n$.
\item $\SSS.\Combine(S,(s_i)_{i \in S}):$ Compute  $s=\sum_{i \in S} \lambda_{i,S}\cdot  s_i \pmod{p}$, where $\lambda_{i,S}=\prod_{j \in S,j \neq i}=j/(j-i) \pmod{p}$ are the Lagrange coefficients.
\end{itemize}

Shamir's secret sharing scheme satisfies the following properties (see \cite{AL17}):
\begin{itemize}
\item $(t+1)$-Correctness: For $s \in \mathbb{Z}_p$  and $(s_i)_{i=1} \leftarrow \SSS.\Share(s,n,t)$, then for any set $S$ with $|S| \geq t+1$, it holds that $\SSS.\Combine(S,(s_i)_{i \in S})=s$.

\item $t$-Privacy: For any $s,s'\in \mathbb{Z}_p$ and any set $S$ with $|S| \leq t$, the distributions $\{(s_i)_{i \in S}~|~(s_i)_{i=1}^n \leftarrow \SSS.\Share(s,n,t)\}$ and  $\{(s'_i)_{i \in S}~|~(s'_i)_{i=1}^n \leftarrow \SSS.\Share(s',n,t)\}$ are identical.
\end{itemize}

 Next, we define the language of valid shares in the Shamir secret sharing scheme 
 $$\LL^\SSS_{n,t}=\{\bs 
 \in \mathbb{Z}_p^n~|~\exists~s,r~:~\bs=\SSS.\Share(s,n,t,r)\}.$$
 We can easily see that, the set $\LL^\SSS_{n,t}$ is equal to the code
 $$C=\{(p(1)~||~p(2)~||~\dots~||~p(n)) \in \mathbb{Z}_p^n~|~p(X) \in \mathbb{Z}_p[X],~\text{deg}(p) \leq t\}$$
 Consider the dual code $C^\perp=\{\bv \in \mathbb{Z}_p^n~|~\bs^\top\cdot\bv=0 \pmod{p}~\forall~\bs \in C\}$. By \cite[Lemma 1]{CD17}, we easily see that $\bs$ is a vector of valid shares for the Shamir secret sharing scheme iff $\bs^\top\cdot \bv=0 \pmod{p}$ for all $\bv \in C^\perp$. Since $C^\perp$ is  a linear subspace of $\mathbb{Z}_p^n$, let $\bH^{t}_n \in \mathbb{Z}_p^{n \times (n-t-1) }$ be the generator matrix of $C^\perp$, then $\bs \in \LL^\SSS_{n,t}$ iff $\bs^\top\cdot \bH^{t}_n =\bzero \pmod{p}$. It is well-known (see \cite[Subsection 2.1]{CD17}) that $$C^\perp=\{(v_1 q(1)~||~v_2 q(2)~||~\dots~||~v_n q(n)) \in \mathbb{Z}_p^n~|~q(X)\in \mathbb{Z}_p[X],~\text{deg}(q) \leq n-t-2\}.$$ Its generator matrix is equal to $\bH^{t}_n=\begin{bmatrix} \bh_1~|~\bh_2~|~\dots~|~\bh_{n-t-1} \end{bmatrix}$ where the vector $\bh_i=(v_1\cdot 1^{i-1}~||~v_2\cdot 2^{i-1}~||~\dots~||~v_n\cdot n^{i-1})$ and $v_i=\prod_{i=1,j \neq i}^n 1/(j-i) \pmod{p}$. This matrix will be required later for designing the trapdoor $\Sigma$-protocols.

\subsection{Public Verifiable Secret Sharing Scheme} \label{subsection-pvss-definition}
In this section, we define the model for public verifiable secret sharing, based on~\cite{CD24} and previous works \cite{Stadler96,FO98,Sch99,HV08,CD17,DKI022,CDSV23,BL23} (for syntax, we will mostly follow the work of \cite[Subsection 2.1]{CD24} as it gives the most formal description so far). As in \cite{CD24}, existing non-interactive PVSSs often follow a common syntax as follows.
\begin{definition}[PVSS, Adapted from \cite{CD24}]
Let $\mathcal{S}, \mathcal{S'}$ be the space of secrets and shares, respectively.  A $(n,t)$-public verifiable secret sharing  scheme with $t \leq n/2$ is a tuple of algorithms $\PVSS=$ $(\PVSS.\Setup,$ $ \PVSS.\Key\Gen,$ $ \PVSS.\Key\Verify,$ $\PVSS.\Share,$ $\PVSS.\Share\Verify,$ $\PVSS.\Dec,$ $\PVSS.\Dec\Verify,\PVSS.\Combine)$, specified as below. 
    \begin{itemize}
    \item $\PVSS.\Setup(1^\lambda) \rightarrow \pp:$ This algorithm is run by a trusted third party. On input security parameter $\lambda$, it returns a public parameter $\pp$.
    \item $\PVSS.\Key\Gen(\pp) \rightarrow ((\pk,\sk),\pi):$ This algorithm is run by each participant. It returns a public-secret key pair $(\pk,\sk)$ and a proof $\pi$ of valid key generation.
    \item $\PVSS.\Key\Verify(\pp,\pk,\pi) \rightarrow 0/1:$ This algorithm is run by a public verifier, it outputs a bit $0$ or $1$ certifying the validity of the public key $\pk$.
    \item $\PVSS.\Share(\pp,(\pk_i)_{i=1}^{n'},s,n',t) \rightarrow (E=(E_i)_{i=1}^{n'},\pi):$ This algorithm is executed by the dealer to share the secret $s \in \mathcal{S}$ for $n' \leq n$ participants. It outputs the ``encrypted shares'' $E=(E_i)_{i=1}^{n'}$ and a proof $\pi$ for correct sharing.
    \item $\PVSS.\Share\Verify(\pp,(\pk_i)_{i=1}^{n'},n',t,E,\pi) \rightarrow 0/1:$ This algorithm is run by a verifier,  it outputs a bit $0$ or $1$ certifying the validity of the sharing process.
    \item $\PVSS.\Dec(\pp,\pk_i,\sk_i,E_i) \rightarrow (s_i,\pi):$ This algorithm is executed by participant $P_i$ owning public-secret key $(\pk_i,\sk_i)$. It outputs a decrypted share $s_i \in \mathcal{S'}$ from $E_i$ and optionally a proof $\pi$ of correct decryption.
    \item $\PVSS.\Dec\Verify(\pp,(\pk_i,E_i,s_i),\pi):$ This algorithm is run by a public verifier; it outputs a bit $0$ or $1$ certifying the validity of the decryption process.
    \item $\PVSS.\Combine(\pp,S,(s_i)_{i \in S}) \rightarrow s/\perp:$ This algorithm is executed by a public verifier. For a set $S$ and a tuple of shares $(s_i)_{i \in S}$, It outputs the original share $s$ or $\perp$ if the secret cannot be reconstructed.
   
    \end{itemize}
\end{definition}
The definition splits the PVSS into $3$ phases. The first is the key generation phase, where each participant generates their keys using $\PVSS.\Key\Gen$ and others verify them using $\PVSS.\Key\Verify$. The second is the sharing phase is where the dealer uses $\PVSS.\Share$ to share a secret for $n'$ participants passed the key verification (for simplicity, we index participants who passed key verification by $\{1,\dots,n'\}$), then others use $\PVSS.\Share\Verify$ to verify the validity of the transcript. The final phase is the reconstruction phase, where participants use $\PVSS.\Dec$ to decrypt the share, then others verify the decrypted shares using $\PVSS.\Dec\Verify$, and finally $\PVSS.\Combine$ is executed on correctly decrypted shares.  

We define the security properties of a PVSS. The security definition will be slightly \textit{different} from \cite{CD24}, especially the correctness and verifiability property, since we additionally need to capture our instantiation for lattice and gap languages.  First, we define valid share language for the PVSS as follows. 

\begin{definition}[Valid Share Language]\label{definition-well-formedness}
We say that $\LL^\Share_{t} \subseteq \bigcup_{i=t+1}^n (\mathcal{S'})^i$ is a \textbf{valid share language} if:  For any $t+1\leq n' \leq n$ and $(s_1~||~s_2~||~\dots~||~s_{n'}) \in \LL^{\Share}_{t}$, there is some $s \in \mathcal{S}$ s.t. for any $S \subseteq [n']$ with $|S| \geq t+1$,  $\PVSS.\Combine(\pp,S,(s_i)_{i \in S})=s$.
\end{definition}

We need it to satisfy the \textit{correctness}, \textit{verifiability} and \textit{IND2-privacy} property. For correctness, we need that if the honest dealer shares a secret $s$, then all honest participants will agree on $s$ in the reconstruction phase, even when there are $t$ malicious participants. For verifiability, we need that, even when a dealer and up to $t$ dishonest participants, if the sharing transcript is accepted, then all honest participants must agree on some secret $s'$, and when reconstructing the secret, it must output $s'$. For IND2-privacy, we need that for any secrets $s^0,s^1$ chosen by the adversary, the sharing transcript of $s^0,s^1$ are indistinguishable.

\begin{definition}[Correctness]
We say that $\PVSS$ achieves \textbf{correctness} if for any PPT adversary $\Adversary$ and $s \in \mathcal{S}$, the following game $\GAME^{\PVSS-\Correctness}(\Adversary,s)$ in Figure~\ref{figure-game-correctness} outputs $1$ with probability $1-\negl(\lambda)$. 

\begin{framedfigure}[Game $\GAME^{\PVSS-\Correctness}(\Adversary,s)$ \label{figure-game-correctness}.] 
\vspace{0.1cm}
\footnotesize $\pp\leftarrow\PVSS.\Setup(1^\lambda)$, $\CC\leftarrow\Adversary(\pp)$. If $|\CC|>t$ return $0$.\\
\footnotesize $((\pk_i,\sk_i),\pi^0_i) \leftarrow \PVSS.\Key\Gen(\pp)~\forall~i\not \in \CC$, $(\pk_i,\pi_i^0)_{i \in \CC}\leftarrow\Adversary(\pp,\CC)$,\\
\footnotesize \textcolor{gray}{Correctness of key generation}\\
\footnotesize If $\exists ~i \not \in \CC$ s.t. $\PVSS.\Key\Verify(\pp,\pk_i,\pi^0_i)=0$, return $0$.\\
\footnotesize Let $G=\{i \in [n]~|~\PVSS.\Key\Verify(\pp,\pk_i,\pi^0_i)=1\}$. Assume that $G=[n']$ and $[n] \setminus \CC \subseteq G$. If not, we re-enumerate the participants $P_i$ with $i \in G$ with an element in $[n']$.\\
\footnotesize $(E=(E_i)_{i=1}^{n'},\pi^1) \leftarrow \PVSS.\Share(\pp,(\pk_i)_{i=1}^{n'},s,n',t)$.\\ 
\footnotesize \textcolor{gray}{Correctness of sharing. This is to capture the case when adversaries could choose $(\pk_i,\sk_i)$ (not necessary from $\PVSS.\Key\Gen$) that passes $\PVSS.\Key\Verify$ but causes $\PVSS.\Share\Verify$ to return $0$, even when an honest dealer computes $\PVSS.\Share$. It might happen in gap language setting.}\\
\footnotesize If $\PVSS.\Share\Verify(\pp,(\pk_i)_{i=1}^{n'},n',t,E,\pi^1)=0$, return $0$.\\

\footnotesize $(s_i,\pi^2_i) \leftarrow \PVSS.\Dec(\pp,\pk_i,\sk_i,E_i)~\forall~i \not \in \CC$,\\ 
\footnotesize $(s_i,\pi^2_i)_{i  \in [n']\cap \CC}\leftarrow \Adversary(\pp,(\pk_i,\pi^0_i)_{i \in [n']},E,\pi^1,(s_i,\pi^2_i)_{i \not \in \CC})$.\\
\footnotesize \textcolor{gray}{Correctness of share decryption}\\
\footnotesize If $\exists~i \not \in \CC$ s.t. $\PVSS.\Dec\Verify(\pp,\pk_i,E_i,s_i,\pi^2_i)=0$, return $0$.\\
\footnotesize \textcolor{gray}{Correctness of share reconstruction. We need that any $t+1$ participants who passed $\PVSS.\Dec\Verify$ must agree on $s$.}\\
\footnotesize Let $S=\{i \in [n']~|~\PVSS.\Dec\Verify(\pp,E_i,s_i,\pi^2_i)=1\}$. If $|S| <t+1$, return $0$. If there exists some $S' \subseteq S$, $|S'| \geq t+1$  such that $\PVSS.\Combine(\pp,S',(s_i)_{i \in S'}) \neq s$, return $0$.\\
\footnotesize Return 1.

\vspace{0.1cm}
\end{framedfigure}
\end{definition}

\begin{definition}[Verifiability] 
 We say $\PVSS$ achieves \textbf{$(\LL^\Key,\LL^\Share_{t})$-verifiability} if i) Each instance in $\LL^\Key$ has an unique corresponding witness, ii) $\LL^\Share_{t}$ is a valid share language for $\PVSS$ (Definition \ref{definition-well-formedness}) and iii) if for any PPT adversary $\Adversary$, the game $\GAME^{\PVSS-\Verify}(\Adversary)$ in Figure \ref{figure-game-verifiability} outputs $1$ with negligible probability.

\begin{framedfigure}[Game $\GAME^{\PVSS-\Verify}(\Adversary)$ \label{figure-game-verifiability}.] 
\vspace{0.1cm}
\footnotesize $\pp\leftarrow\PVSS.\Setup(1^\lambda)$. Parse $\pp=(\pp',\pp^\star)$.  $\CC\leftarrow\Adversary(\pp)$. If $|\CC|>t$ return $0$.\\
\footnotesize $((\pk_i,\sk_i),\pi^0_i) \leftarrow \PVSS.\Key\Gen(\pp)~\forall~i\not \in G\cap \CC$, $(\pk_i,\pi_i^0)_{i \in \CC}\leftarrow\Adversary(\pp,\CC)$,\\
\footnotesize Let $G=\{i \in [n]~|~\PVSS.\Key\Verify(\pp,\pk_i,\pi^0_i)=1\}$. Assume that $G=[n']$ and $[n] \setminus \CC \subseteq G$. If not, we re-enumerate the participants $P_i$ with $i \in G$ with an element in $[n']$.\\

\footnotesize $(E=(E_i)_{i=1}^{n'},\pi^1)\leftarrow \Adversary(\pp,(\pk_i,\pi^0_i)_{i \in G})$.\\

\footnotesize $(s_i,\pi^2_i) \leftarrow \PVSS.\Dec(\pp,\pk_i,\sk_i,E_i)~\forall~i \not \in \CC$,\\ 
\footnotesize $(s'_i,\pi^2_i)_{i  \in G\cap \CC}\leftarrow \Adversary(\pp,(\pk_i,\pi^0_i)_{i \in [n]},E,\pi^1,(s_i,\pi^2_i)_{i \not \in \CC})$.\\
\footnotesize \textcolor{gray}{The execution of $\Adversary$ is done. The checks below are not in polynomial time. But we can perform all these checks as long as we want because we do not interact with $\Adversary$ anymore.}\\
\footnotesize \textcolor{gray}{Verifiability of key generation}\\ \footnotesize If $(\pp',\pk_i) \not \in \LL^\Key$ for some $i \in G \cap \CC$, return $1$.\\
\footnotesize \textcolor{gray}{Verifiability of sharing}\\
\footnotesize At this point, consider unique $(\sk_i)_{i \in G\cap \CC}$ s.t. $((\pp',\pk_i),\sk_i) \in \RR^\Key~\forall i \in G\cap \CC$.\\
\footnotesize Let $(s_i,.)\leftarrow \PVSS.\Dec(\pp,\pk_i,\sk_i,E_i)~\forall~i \in G \cap \CC$. \\ \footnotesize If $(s_1~||~s_2~||~\dots~||~s_{n'}) \not \in \LL^\Share_{t}$ and $\PVSS.\Share\Verify(\pp,(\pk_i)_{i=1}^{n'},n',t,E,\pi^1)=1$, return $1$.\\
\footnotesize \textcolor{gray}{Verifiability of decryption}\\
\footnotesize If $s'_i \neq s_i$ and $\PVSS.\Dec\Verify(\pp,\pk_i,E_i,s'_i,\pi^2_i)=1$ for some $i \in G\cap \CC$, return $1$.\\
\footnotesize If $\PVSS.\Dec\Verify(\pp,\pk_i,E_i,s_i,\pi^2_i)=0$ for some $i \in G\setminus \CC$, return $1$.\\
\footnotesize Return $0$.
\vspace{0.1cm}
\end{framedfigure}
\end{definition}

\begin{definition}[IND2-Privacy]\label{definition-pvss-privacy}
We say that $\PVSS$ achieves \textbf{IND2-privacy} if for any PPT adversary $\Adversary$, it holds that $\Adv^{\PVSS-\mathsf{IND}}(\Adversary)=|\Pr[\GAME^{\PVSS-\mathsf{IND}}_0(\Adversary)=1-\Pr[\GAME^{\PVSS-\mathsf{IND}}_1(\Adversary)=1]]| \leq \negl(\lambda)$, where  $\GAME^{\PVSS-\mathsf{IND}}_b(\Adversary)$ is in Figure \ref{figure-game-privacy}.

\begin{framedfigure}[Game $\GAME^{\PVSS-\mathsf{IND}}_b(\Adversary)$ with supporting interactive oracle $\OO_{\PVSS,\Adversary}(.)$. \label{figure-game-privacy}] 
\vspace{0.1cm}
\footnotesize \underline{$\GAME^{\PVSS-\mathsf{IND}}_b(\Adversary)$:}\\
\footnotesize $\pp\leftarrow\PVSS.\Setup(1^\lambda)$, $\CC\leftarrow\Adversary(\pp)$. If $|\CC|>t$ return $0$.\\
\footnotesize $((\pk_i,\sk_i),\pi^0_i) \leftarrow \PVSS.\Key\Gen(\pp)~\forall~i\not \in \CC$, $(\pk_i,\pi_i^0)_{i \in \CC}\leftarrow\Adversary(\pp,\CC)$,\\
\footnotesize Let $G=\{i\in [n]~|~\PVSS.\Key\Verify(\pp,\pk_i,\pi^0_i)=1\}$. Assume that $G=[n']$ and $[n] \setminus \CC \subseteq G$. If not, we re-enumerate the participants $P_i$ with $i \in G$ with an element in $[n']$.\\
\footnotesize $(s^0,s^1)\leftarrow \Adversary^{\OO_{\PVSS,\Adversary}(.)}(\pp,(\pk_i,\pi^0_i)_{i \in [n']})$.\\
\footnotesize \textcolor{gray}{Challenge phase}\\
\footnotesize $(E^b,\pi^b) \leftarrow \PVSS.\Share(\pp,s^b,n',t)$.\\
\footnotesize $b' \leftarrow \Adversary((\pp,(\pk_i,\pi^0_i)_{i \in [n']},s^0,s^1,E^b,\pi^b)$.\\
\footnotesize Return $b'$.\\
\vspace{0.5cm}\\

\footnotesize \underline{Interactive oracle $\OO_{\PVSS,\Adversary}(s):$}\\
\footnotesize $(E=(E_i)_{i=1}^{n'},\pi^1)\leftarrow \PVSS.\Share(\pp,(\pk_i)_{i=1}^{n'},s,n',t)$.\\
\footnotesize $(s_i,\pi^2_i) \leftarrow \PVSS.\Dec(\pp,\pk_i,\sk_i,E_i)~\forall~i \not \in \CC,(s_i,\pi^2_i)_{i  \in \CC}\leftarrow \Adversary(\pp,(\pk_i,\pi^0_i)_{i \in [n']},E,\pi^1)$.\\
\footnotesize Let $S_2=\{i \in G~|~\PVSS.\Dec\Verify(\pp,E_i,s_i,\pi^2_i)=1\}$.\\
\footnotesize Return $\PVSS.\Combine(\pp,(s_i)_{i \in S})$.
\vspace{0.2cm}
\end{framedfigure}

\end{definition}

Here in the verifiability definition, instead of forcing $E=(E_i)_{i=1}^{n'}$ to be the encryption of the shares, we only require that: For any valid key pairs $(\pk_i,\sk_i)_{i=1}^{n'} \in\RR^\Key$, if $(E,\pi)$ is accepted by verifier, then the \textit{honestly} decrypted message $s_i$ of $E_i$ using secret key $\sk_i$ must be valid shares of some secret $s$. Assuming the existence of $\sk_i$ at this point is reasonable, because after the key generation phase, then for each $P_i$ passed verification, such a secret key $\sk_i$ must exist. In the definition, currently we have to restrict $\LL^\Key$ s.t. each $\pk_i$ has a \textit{unique} corresponding witness $\sk_i$ so the ``honestly decrypted'' share $(s_i.) \leftarrow \PVSS.\Dec(\pp,\pk_i,E_i,\sk_i)$ is \textit{uniquely determined} for each $i \in G\cap \CC$. In this way, given $(\pp',\pk_i) \in \LL^\Key$ for all $i \in [n']$ and an accepted transcript $(E,\pi)$, then a verifier is convinced that there is some \textit{unique} $(s_1~||~s_2~||~\dots~||~s_{n'}) \in \LL^\Share_{t}$ that could be decrypted from $(E)_{i=1}^n$, and honest participants would agree on a unique secret $s$. This is due to Definition \ref{definition-well-formedness} of $\LL^\Share_t$ and all shares $s_i$ are uniquely determined. Finally, when $s'_i$ is revealed, $\PVSS.\Dec\Verify$ returns $1$, meaning that $s'_i$ is the honestly decrypted share $s_i$ and the same secret $s$ is reconstructed. So our verifiability definition fully captures that a verifier is convinced that the whole PVSS is executed correctly. In the privacy property, we define $\OO_{\PVSS,\Adversary}(.)$ to capture that the adversary still cannot distinguish between $s^0$ and $s^1$, \textit{even if it has already executed the PVSS several times} with previous (adversarially chosen) secrets. Previous definitions \cite{AJ05,HV08,CD17,DKI022,CD24} only consider privacy for a single PVSS execution. Also, there is a weaker property named IND1-privacy \cite{CD17} where the secrets $s^0,s^1$ are uniformly chosen by the challenger, but we will use IND2-privacy like \cite{CD24} since it implies IND1-privacy \cite{HV08}.

\section{A Generic PVSS Construction from Public Key Encryption and NIZK}\label{section-generic-pvss}

In this section, we describe a generic, non-interactive PVSS construction from any i) IND-CPA secure public key encryption scheme (PKE) where each public key has a unique corresponding secret key and ii) an NIZK for gap languages. We then formally prove the security of the generic PVSS. 

To our knowledge, we are the first to give a \textit{formal construction} of a generic PVSS from an IND-CPA secure PKE and NIZK for a given \textit{gap language} ($\LL_{zk}\subseteq \LL_{sound}$), based on the idea of the GMW approach of \cite{GMW86}. The work of \cite{GMW86} only provided an informal PVSS sketch (without concrete formal proof for the PVSS), and it considers exact languages ($\LL_{zk}=\LL_{sound}$). Other works, such as \cite{Stadler96,FO98,Sch99,YY00,HV08,J11,JVS14,CD17,CD20,GHL22,DKI022} only provide specific constructions from the decisional Diffie-Hellman or factoring problems instead of using a generic IND-CPA secure PKE. Finally,  \cite[Figure $6$]{CDGK22} also constructed a generic PVSS. Their construction and ours have two major differences: First, their encryption scheme requires the encryption, the decryption algorithm, and the public key to be \textit{linear maps} in terms of the message, ciphertext, and secret key, respectively. Ours only requires any IND-CPA encryption scheme with i) message space $\mathbb{Z}_p$, ii) a key verification algorithm $\PKE.\Key\Verify$, and iii) the public key has a unique secret key (\textit{no homomorphism required}). Our key verification algorithm $\PKE.\Key\Verify$ is a generalization of the former: In the former construction, they also need key verification, and it can be done by checking if $\pk=F(\sk)$ for some linear function~$F$. This process can be done by simply defining a $\PKE.\Key\Verify$ algorithm that captures the verification. Another difference is that, we need to define NIZKs for gap languages ($\LL_{zk}\subseteq \LL_{sound}$) to capture lattice instantiations, while their construction requires exact languages ($\LL_{zk}= \LL_{sound}$)\footnote{The reason for this difference is due to the lattice setting, since in lattice setting, we need to provide proofs of smallness that only has a gap soundness property for the following sense: If the prover has a witness vector $\be$ s.t. $||\be|| \leq t_p$, then the verifier accepts, but if the verifier accepts, it only implies that there is a witness $||\be|| \leq c\cdot t_p$ for $c>1$. Thus, one can imagine $\LL_{zk}$ is the set of vectors $\be$ s.t. $||\be|| \leq t_p$, while $\LL_{sound}$ is is the set of vectors $\be$ s.t. $||\be|| \leq c\cdot t_p$.}. Due to this, we need to define suitable gap languages so that it would be possible to construct practical NIZKs for them.  We will provide concrete lattice-based instantiations of the NIZKs in Section \ref{section-nizk}.

\subsection{Construction}\label{construction-generic-pvss}
 We now describe the generic construction from the idea of the gap languages $\LL^\Key,\LL^\Enc,\LL^\Dec$ we informally defined in Subsection~\ref{section-technical-overview}. Now, suppose that there exist the following primitives: 
\begin{itemize}
\item A $(n,t)$-Shamir secret scheme $\SSS=(\SSS.\Share,$ $\SSS.\Combine)$ (Subsection \ref{section-special-ss}) with parity check matrix $\bH^{t}_n$ and corresponding language of valid shares $\LL^\SSS_{n,t}=\{\bs \in \mathbb{Z}_p^n~|~\bs^\top\cdot \bH^{t}_n=\bzero \pmod{p}\}$. Finally, define $\LL^\SSS_{t}=\bigcup_{i=t+1}^n \LL^\SSS_{i,t}$.
\item A public key encryption scheme $\PKE=(\PKE.\Setup,\PKE.\Key\Gen,\PKE.\Key\Verify,$ $\PKE.\Enc,\PKE.\Dec)$ with message space $\mathcal{M}=\mathbb{Z}_p$ and witness space $\mathcal{W}^\Dec$. In addition, it must hold that: Given a public key $\pk$, there exists \textbf{at most one} valid secret key $\sk$ such that $\PKE.\Verify(\pp,\pk,\sk)=1$. 
\item Two sets $\mathcal{W}^\Dec_{zk}$ and $\mathcal{W}^\Dec_{sound}$ s.t. $\mathcal{W}^\Dec\subseteq \mathcal{W}^\Dec_{zk}$ and $\mathcal{W}^\Dec_{zk} \subseteq\mathcal{W}^\Dec_{sound}$. For now, we view them as generic sets s.t. there exist  NIZKs for $\LL^\Key,\LL^\Enc,\LL^\Dec$ below.
\item An adaptive soundness and adaptive multi-theorem NIZK $\NIZK_0=(\NIZK_0,\Setup,\NIZK_0.\Prove,\NIZK_0.\Verify)$ for $\LL^{\Key}$ where
$$\LL^{\Key}_{zk}=\{(\pp,\pk)~|~\exists~ \sk~\land~r \in \RAND^\Key~\text{s.t.}~ (\pk,\sk)=\PKE.\Setup(\pp,r)\},$$
$$\LL^{\Key}_{sound}=\{(\pp,\pk)~| ~\exists~ \sk~\text{s.t.}~\PKE.\Key\Verify(\pp,\pk,\sk)=1\}.$$
\item An adaptive soundness and adaptive multi-theorem NIZK $\NIZK_1=(\NIZK_1,\Setup,\NIZK_1.\Prove,\NIZK_1.\Verify)$ for $\LL^{\Enc}$ where
\begin{equation*}
\begin{aligned}&\LL^{\Enc}_{zk}=\{(\pp,n,t,(\pk_i,E_i)_{i=1}^n)~|~(\pp,\pk_i) \in \LL^\Key_{sound}~\forall~ 1\leq i\leq n~\land~\exists~ s_i\in \mathbb{Z}_p,\\
&\quad\quad\quad r_i \in \RAND^\Enc~\text{s.t.}~E_i=\PKE.\Enc(\pp,\pk_i,s_i,r_i)~\land~\bs \in \LL^\SSS_{n,t}\},
\end{aligned}
\end{equation*}
\begin{equation*}
\begin{aligned}
&\LL^{\Enc}_{sound}=\{(\pp,n,t,(\pk_i,E_i)_{i=1}^n)~|~(\pp,\pk_i) \in \LL^\Key_{sound}~\forall~ 1\leq i\leq n~\land~\exists~ s_i \in \mathbb{Z}_p,\\
&~\quad\quad w_i\in \mathcal{W}^\Dec_{zk}~\text{s.t.}~(s_i,w_i)=\PKE.\Dec(\pp,\pk_i,\sk_i,E_i)~ \land~\bs \in \LL^\SSS_{n,t}\},
\end{aligned}
\end{equation*}
where $\sk_i$ is the unique witness of $(\pp,\pk_i)$ in $\RR^\Key_{sound}$. It is easy to see that $\LL^\Enc_{zk} \subseteq \LL^\Enc_{sound}$ due to the encryption correctness property of the PKE (Def.~\ref{definition-encryption-correctness}). Here in $\LL^\Enc_{zk}$, the dealer does not need to prove that $(\pp,\pk_i) \in \LL^\Key_{sound}$ because it has been proved by participants (it holds for honest participants by the key correctness property of the PKE, see Def. \ref{key-correctness}). Instead, $\NIZK_1.\Prove$ is designed so that dealer only needs the witnesses $(s_i,r_i)_{i=1}^n$ to prove that $E_i=\PKE.\Enc(\pp,\pk_i,s_i,r_i)~\land~ \bs \in \LL^\SSS_{n,t}$. 
\item  An adaptive soundness and adaptive multi-theorem NIZK $\NIZK_2=(\NIZK_2,\Setup,\NIZK_2.\Prove,\NIZK_2.\Verify)$ for $\LL^{\Dec}$ where
\begin{equation*}
\begin{aligned}&\LL^{\Dec}_{zk}=\{(\pp,\pk,E,s)~|~\exists~(\sk,r)~\text{ s.t. }\PKE.\Dec(\pp,\pk,\sk,E)=(s,w)\\
&\quad\quad\quad\quad\quad\quad\quad\quad\quad\quad\quad\land~ w \in \mathcal{W}^\Dec_{zk}~ \land ~(\pk,\sk)=\PKE.\Setup(\pp,r) \},
\end{aligned}
\end{equation*}
\begin{equation*}
\begin{aligned}
&\LL^{\Dec}_{sound}=\{(\pp,\pk,E,s)~|~\exists~\sk~\text{ s.t. }~\PKE.\Dec(\pp,\pk,\sk,E)=(s,w)\\
&\quad\quad\quad\quad\quad\quad\quad\quad\quad\quad\quad \land~w \in \mathcal{W}^\Dec_{sound}~ \land ~\PKE.\Key\Verify(\pp,\pk,\sk)=1\},
\end{aligned}
\end{equation*}
\end{itemize}
The generic PVSS construction is as follows.

\begin{itemize}
\item $\PVSS.\Setup(1^\lambda):$ A trusted third party generates public parameters $\pp' \leftarrow \PKE.\Setup(1^\lambda)$ and construct the language $\LL^\Key,\LL^\Enc,\LL^\Dec$ based on $\pp'$. Then from their description, generate $\crs_i$ from $\NIZK_i.\Setup$ for all $i \in \{0,1,2\}$. Return $\pp=(\pp',(\crs_i)_{i=0}^2)$.
\item $\PVSS.\Key\Gen(\pp):$ Each participant $P_i$ provides a public-secret key pair $(\pk_i,\sk_i) \leftarrow \PKE.\Key\Gen(\pp',r_i)$ for some randomness $r_i$ and provides a proof $\pi_i \leftarrow \NIZK_0\Prove(\crs_0,(\pp',\pk_i),(\sk_i,r_i))$. Finally, return $((\pk_i,\sk_i),\pi_i)$.
\item $\PVSS.\Key\Verify(\pp,\pk,\pi):$ Return $\NIZK_0.\Verify(\crs_0,(\pp',\pk),\pi)$.
\item $\PVSS.\Share(\pp,(\pk_i)_{i=1}^{n'},s,n',t):$ Let $s$ be the secret and $n'$ be the number of participants who passed key verification. The dealer computes $(s_i)_{i=1}^{n'} \leftarrow \SSS.\Share(s,n',t)$ and $E_i=\PKE.\Enc(\pp',\pk_i,s_i,r_i)$ for all $1 \leq i \leq n'$, then provides a proof $\pi \leftarrow \NIZK_1.\Prove(\crs_1,(\pp',n',t,(\pk_i,E_i)_{i=1}^{n'}),(s_i,r_i)_{i=1}^{n'})$. Returns $((E_i)_{i=1}^{n'},\pi)$.
\item $\PVSS.\Share\Verify(\pp,(\pk_i)_{i=1}^{n'},n',t,E,\pi):$ Simply return $\NIZK_1.\Verify(\crs_1,(\pp',n',$ $t,$ $(\pk_i,E_i)_{i=1}^{n'}),$   $\pi)$.
\item $\PVSS.\Dec(\crs_1,\pk_i,E_i,\sk_i):$ Compute $(s_i,w_i)=\PKE.\Dec(\pp,\pk_i,\sk_i,E_i)$ and receive additional witness $w_i$ for decryption. Then provides a proof $\pi_i \leftarrow \NIZK_2.\Prove(\crs_2,(\pp',\pk_i,E_i,s_i),(\sk_i,r_i,w_i))$. Return $(s_i,\pi_i)$.
\item $\PVSS.\Dec\Verify(\pp,\pk_i,E_i,s_i,\pi_i):$ Return $\NIZK_2.\Verify(\crs_2,(\pp',\pk_i,E_i,s_i),\pi_i)$.
\item $\PVSS.\Combine(\pp,S,(s_i)_{i \in S}):$ If $|S| \leq t$ return $\perp$. Otherwise return $s=\SSS.\Combine(S,(s_i)_{i \in S})$.
\end{itemize}

\subsection{Security Proof}\label{subsection-security-proof}

\begin{theorem}\label{theorem-correctness}
The construction in Subsection~\ref{construction-generic-pvss} satisfies the correctness property.
\end{theorem}
\begin{proof}[Proof Sketch]
\textcolor{blue}{At a high level, this is implied by the correctness of the NIZKs and SSS. Indeed, a verifier always accepts an honestly generated proof from the NIZK, and the correctness of SSS ensures that any $t+1$ valid shares will recover the secret $s$}. The full proof will be presented in Appendix \ref{appendix-pvss-correctness}.
\end{proof}

\begin{theorem}\label{theorem-verifiability}
Each instance in the language $\LL^\Key_{sound}$ has a unique witness, and the language $\LL^{\SSS}_{t}$ is a valid share language for the PVSS (Definition \ref{definition-well-formedness}). Finally, the construction in Subsection~\ref{construction-generic-pvss} satisfies the $(\LL^\Key_{sound},\LL^\SSS_{t})$-verifiability property.
\end{theorem}
\begin{proof}[Proof Sketch]
\textcolor{blue}{At a high level, the verifiability of key generation and sharing is easily implied by the soundness of the NIZKs. The verifiability of decryption requires slightly more analysis: Suppose $(s'_i,\pi_i)$ is provided by $P_i,$ and $\sk_i,s_i$ are the valid secret key and actual share $P_i$ (in this phase, we can assume that $P_i$ has passed the key generation phase and the dealer has passed the sharing phase. In particular, we have $\PKE.\Verify(\pp',\pk_i,\sk_i)=1$ and $(s_i,.)=\PKE.\Dec(\pp',\pk_i,\sk_i,E_i)$).  There are two possible cases: 
\begin{itemize}
\item  There does not exist $\sk'_i$ such that $(s'_i,w_i)=\PKE.\Dec(\pp',\pk_i,\sk'_i,E_i)$ and $\PKE.\Verify(\pp',\pk_i,\sk'_i)=1$ for some $w_i \in \mathcal{W}^\Dec_{sound}$. We can later reach a contradiction due to the soundness property of the NIZK. 
\item There exists $\sk'_i$ such that $(s'_i,w_i)=\PKE.\Dec(\pp',\pk_i,\sk'_i,E_i)$ and $\PKE.\Verify(\pp',\pk_i,\sk'_i)=1$ for some $w_i \in \mathcal{W}^\Dec_{sound}$, but $s'_i \neq s_i$.  We can later reach a contradiction due to the fact that each public key has a unique corresponding secret key, which forces $s'_i=s_i$.
\end{itemize}
Finally, the correctness of the NIZK for decryption also implies that the dealer cannot choose $E$ that has passed $\PVSS.\Share\Verify$ ($E$ might not be from $\PVSS.\Share$, as we are considering gap languages) but causes $\PVSS.\Dec\Verify$ to return $0$ for honest participants.} The full proof will be presented in Appendix \ref{appendix-pvss-verifiability}.
\end{proof}

\begin{theorem}\label{theorem-privacy}
The construction in Subsection~\ref{construction-generic-pvss} satisfies the IND2-privacy property.
\end{theorem}
\begin{proof}[Proof Sketch]
\textcolor{blue}{This is implied by the zero-knowledge property of the NIZK, the IND-CPA property of the PKE, and the privacy of SSS}. The formal proof \textcolor{blue}{requires considering a sequence of hybrid games and} will be presented in Appendix \ref{appendix-pvss-privacy}.
\end{proof}

\section{Supporting Trapdoor \texorpdfstring{$\Sigma$}\texorpdfstring{-}Protocols} \label{section-nizk}
In the previous section, we needed three NIZK protocols for correct key generation, sharing, and decryption. In this section, we describe the trapdoor $\Sigma$-protocols for correct key generation, sharing, and decryption when instantiated with the lattice-based encryption of \cite{ACPS09}. In this way, we could achieve the respective NIZKs for these relations in the common reference string model using the compiler of \cite{LNPT20}. Among the three protocols, the trapdoor $\Sigma$-protocol for LWE key generation has been described in \cite[Appendix G]{LNPT20}, but we will describe it anyway for completeness and adapt its technique to design the remaining two. The remaining two trapdoor $\Sigma$-protocols for sharing and decryption are our own proposal and design in this paper. A final small note is that, in this section, we will denote the shares by $m_i$ instead of $s_i$ in NIZKs for sharing and decryption to represent them as plaintexts.
\subsection{Trapdoor \texorpdfstring{$\Sigma$}\texorpdfstring{-}Protocol for Correct Public Key Generation}\label{subsection-nizk-for-key-generation}
In this section, we describe the NIZK for correct LWE key generation. Consider four bounds $ B^\Key_{\bs}, B^\Key_{\be},  B^{\Key\star}_{\bs}$ and $ B^{\Key\star}_{\be}$. Recall that in the scheme of \cite{ACPS09}, if $(\bb,\bs) \leftarrow \PKE.\Setup(\bA)$ then $\bb=\bs^\top\cdot \bA+\be^\top \pmod{q}$ for some $||\bs|| \leq\sqrt{v} \cdot \alpha\cdot q$ and $||\be|| \leq \sqrt{u} \cdot \alpha\cdot q$  with overwhelming probability. Thus we can use these  values as the bound $B^\Key_{\bs}$ and $B^\Key_{\be}$ respectively, and let $\LL^\Key_{zk}$ to be the set of all $(\bA,\bb)$ s.t. $\bb=\bs^\top\cdot\bA+\be \pmod{q}$ and $||\bs|| \leq B^\Key_{\bs}$, $||\be|| \leq B^\Key_{\be}$. The relaxed bounds $B^{\Key\star}_{\bs}$ and $B^{\Key\star}_{\be}$ will be determined later, and they are the same bounds denoted in Section \ref{subsection-pke}. The languages $\LL^{\Key}=(\LL^\Key_{zk},\LL^\Key_{sound})$ for public key generation is as below. 
\begin{equation*}
\begin{aligned}
\LL^{\Key}_{zk}=\{(\bA,\bb) \in \mathbb{Z}_q^u~|~\exists~\bs \in \mathbb{Z}^v,~\be \in \mathbb{Z}^u~&:~
\bb=\bs^\top\cdot \bA+\be^\top \pmod{q}\\
&\land~||\bs|| \leq B^\Key_{\bs},~ ||\be|| \leq B^\Key_{\be}\},
\end{aligned}
\end{equation*}
\begin{equation*}
\begin{aligned}
\LL^{\Key}_{sound}=\{(\bA,\bb) \in \mathbb{Z}_q^u~|~\exists~\bs \in \mathbb{Z}^v,~\be \in \mathbb{Z}^u&~:~
\bb=\bs^\top\cdot \bA+\be^\top \pmod{q}\\
&\land~||\bs|| \leq B^{\Key\star}_{\bs},~ ||\be|| \leq B^{\Key\star}_{\be}\},
\end{aligned}
\end{equation*}
\noindent  Consider parameters $(B^\Key_\bs,B^{\Key\star}_{\bs},B^\Key_\be,B^{\Key\star}_{\be},\sigma^\Key)$ such that $\sigma^\Key> \sqrt{(B^\Key_{\be})^2+(B^\Key_{\bs})^2}\cdot \sqrt{\log (u+v)}$ and $2\sqrt{(u+v)}\cdot \sigma^\Key<B^{\Key\star}_{\bs}=B^{\Key\star}_{\be}=O(p/(\sqrt{v \log p})$. The trapdoor $\Sigma$-protocol is below:
\begin{itemize}
\item $\Trap\Sigma.\Gen_{\pp}(1^\lambda):$ Choose moduli $p,q=p^2$, dimensions $u,v$. Return $(p,q,u,v)$.
\item $\Trap\Sigma.\Gen_{\LL}(\pp,\LL^\Key):$ This algorithm takes input the description of the language $\LL^\Key$, specified by a matrix $\bA \uniformly \mathbb{Z}_q^{v \times u}$ and the bounds~$B^\Key_\bs,B^{\Key\star}_{\bs},B^\Key_\be,$ $B^{\Key\star}_{\be}$. Return $\crs=(\bA,p,q,u,v,$ $B^\Key_\bs,B^{\Key\star}_{\bs},B^\Key_\be,$ $B^{\Key\star}_{\be},$ $\sigma^\Key)$.
\item $\Trap\Sigma.\TrapGen(\pp,\LL^\Key,\bT):$   This algorithm takes input the description of $\LL^\Key$, specified by a matrix $\bA$, a trapdoor $\bT$ such that $(\bA,\bT) \leftarrow \TrapGen(1^\lambda,v,u)$ (see Lemma \ref{theorem-invert-lwe}), and the bounds $B^\Key_\bs,B^{\Key\star}_{\bs},B^\Key_\be,$ $B^{\Key\star}_{\be}$. Return $(\crs,\tr)=((\bA,p,q,u,v,B^\Key_\bs,B^{\Key\star}_{\bs},B^\Key_\be,B^{\Key\star}_{\be},\sigma^\Key),\bT)$. 

\item $\Trap\Sigma.\Prove \langle \PP(\crs,\bx,\bw), \VV(\crs,\bx) \rangle:$ For a statement $\bx=(\bA,\bb)$ and witness $(\bs,\be)$ of $\LL^\Key_{zk}$, the two parties interact as follows.

\begin{enumerate}
\item $\PP$ samples the vectors $(\br~||~\bff) \leftarrow D_{\mathbb{Z}^{u+v},\sigma}$, and provides the values $\bd=\br^\top\cdot \bA+\bff^\top \pmod{q}$ to $\VV$. 
\item $\VV$ samples a challenge $c \uniform \{0,1\}$ and sends $c$ to $\PP$.
\item $\PP$ sends the value $\bz=\br+c\cdot \bs,~\bt=\bff+c\cdot \be $ to $\VV$ with probability $1-K$, where $K=\text{min}\left(\frac{D_{\mathbb{Z}^{u+v},\sigma^\Key}(\bff)}{M\cdot D_{\mathbb{Z}^{u+v},\sigma^\Key,c\cdot (\bs~||~\be) }(\bff)},1\right)$ where $M=e^{1/\log (u+v)+12/\log^2 (u+v)}$.
\item $\VV$ checks whether $\bz^\top\cdot \bA+\bt^\top=\bd+c\cdot \bb \pmod{q}$ and finally checks if $||(\bz~||~\bt)|| \leq \sqrt{u+v}\cdot \sigma^\Key$.
\end{enumerate}
\item $\Trap\Sigma.\BadChallenge(\crs,\bb,\bd,\bT):$ While $0 \leq c \leq 1$, do as follows.
\begin{enumerate}
\item Compute $(\bz_c,\bt_c)=\Invert(\bA,\bT,\bd+c\cdot \bb)$.
\item If $||(\bz_c~||~\bt_c)|| > \sqrt{u+v}\cdot \sigma^\Key$, return $1-c$. 
\end{enumerate} Otherwise, return $\perp$. 
\end{itemize}
In the definition, $\BadChallenge$ might output a bit in the challenge space even if there is no bad challenge. However, this is not a problem, because to soundly instantiate Fiat-Shamir, when $\bx \not \in \LL_{sound}$ the $\BadChallenge$ only needs to output the bad challenge if it exists so that CIHF could avoid it. When there is no bad challenge, then whichever challenge from CIHF would ensure that no valid prover response exists, so any bit outputted by $\BadChallenge$, in this case, would be fine.

\begin{theorem}\label{theorem-nizk-for-key-generation}
Consider $(B^\Key_\bs,B^{\Key\star}_\bs,B^\Key_\be,B^{\Key\star}_\be,\sigma^\Key)$ specified above. Then the construction is a trapdoor $\Sigma$-protocol for $\LL^{\Key}=(\LL^{\Key}_{zk},\LL^{\Key}_{sound})$.
\end{theorem}
\begin{proof}[Proof Sketch]
\textcolor{blue}{It is straightforward to see that the construction satisfies correctness. The zero-knowledge property is done by considering a simulator that samples $(\bz~||~\bt)$ from $D_{\mathsf{Z}^{u+v},\sigma^\Key}$and computes $\bd=\bz^\top\cdot\bA_\bt-c\cdot \bb \pmod{q}$ and we will prove that the simulated transcript is statistically close to the real transcript via Lemma \ref{lemma-rejection-sampling}. CRS indistinguishability follows from Lemma \ref{theorem-invert-lwe}. Special soundness is easily seen by proving that there cannot be two valid responses $\bz_c,\bt_c$ for all $c\in \{0,1\}$, which leads to contradiction. Finally, to prove correctness of the $\BadChallenge$ function, we first prove that it must return a bit in $\{0,1\}$ and this implies that there is some $c$ such that $(\bz_c,\bt_c)=\Invert(\bA,\bT,\bd+c\cdot \bb)$, and $||(\bz_c~||~\bt_c)|| > \sqrt{u+v}\cdot \sigma^\Key$. We prove that, when $c=0$, then there is no valid response when the challenge is not equal to~$1$ and similarly for $c=1$, thus $\BadChallenge$ has returned the correct output}. The full security proof is presented in Appendix \ref{appendix-proof-of-key-generation}. 
\end{proof}

\noindent \textbf{Parallel Repetition.} The above scheme has a soundness error $1/2$. To achieve negligible soundness error, we should run the scheme in parallel as follows: First, compute the first messages $(\msg_{1i})_{i=1}^\lambda=(\br_i~||~\bff_i)_{i=1}^\lambda$, then run the CIHF to obtain challenges $(c_i)_{i=1}^\lambda$. Finally, compute the response  $(\msg_{2i})_{i=1}^\lambda=(\bz_i~||~\bt_i)_{i=1}^\lambda$ and returns it with probability at least $(1-2^{-100})/M$ by Lemma \ref{lemma-rejection-sampling}.  This modification requires $\sigma^\Key>\sqrt{\log(\lambda(u+v))} \sqrt{\lambda((B^\Key_{\be})^2+(B^\Key_{\bs})^2)}$ and $M=e^{1/\log(\lambda( (u+v))+12/\log^2 (\lambda(u+v))}=1-1/\text{poly}(u)$. Hence, when using the NIZK compiler from the parallel repetition, we only need to recompute the first messages $O(1)$ times to receive valid responses. Note that the bounds $B^{\Key\star}_\be,B^{\Key\star}_\bs$ are still taken to be $\sqrt{u+v}\cdot \sigma^\Key$, which is unchanged. This is because from Lemma \ref{lemma-rejection-sampling}, the second message is statistically close to the vector $(\msg'_{21}~||~\msg'_{22}~||~\dots~||~\msg'_{2\lambda}) \leftarrow D_{\mathbb{Z}^{\lambda(u+v)},\sigma^\Key}$ and is returned with probability $1/M$ Next, we can use Lemma \ref{lemma-nextbound} to prove that, the probability that all of $\msg'_{2i}=(\bz_i~||~\bt_i) \in \mathbb{Z}^{u+v}$ has norm at most $\sqrt{u+v}\cdot \sigma^\Key$ is at least $1-\lambda\cdot 2^{-(u+v)}$ by union bound.

\subsection{Trapdoor \texorpdfstring{$\Sigma$}\texorpdfstring{-}Protocol for Correct Sharing}\label{subsection-nizk-for-sharing}
We now design an NIZK for correct sharing using the encryption of \cite{ACPS09}, \textit{assuming that the public keys of participants are correctly generated}. At a high level, given $(\bA,\bb_i) \in \LL^\Key_{sound}$ for all $i \leq n$, the encryptor encrypts $n$ scalar $\bmm=(m_i)_{i=1}^n$ (again, recall that we denote $m_i$ instead of $s_i$ to represent them as plaintexts) into $(\bc_{1i},\bc_{2i})_{i=1}^n$ and needs to prove that they are valid encryptions $(m_i)_{i=1}^n$. This happens if and only if $\bc_{1i}=\bA\cdot \br_i \pmod{q}$ and
$\bc_{2i}=\bb\cdot \br_i+e_i+p\cdot m_i \pmod{q}~\forall~1\leq i\leq n$ for some small vectors $\br$ and values $e_i$. In addition, these scalars must satisfy $\bmm \in \LL^\SSS_{n,t}$. Denote the bound  of $||\br||$ as $B^\Enc_\br$ and $|e|$ as $B^\Enc_e$ (for concreteness, we have $B^\Enc_\br=\sqrt{u}\cdot r$ and $B^\Enc_e=\sqrt{v} \cdot \beta q$ where $r,\beta$ is given in Subsection \ref{subsection-pke}. But for now, we only need to know that $B^\Enc_\br$ and $B^\Enc_e$ are appropriately determined). The above conditions are all the necessary and sufficient conditions for $\LL^\Enc_{zk}$. 

Now, to define $\LL^\Enc_{sound}$,
 we easily see that for some sufficient large $B^{\Enc\star}_f<p/2$,  $(m_i,f_i)=\PKE.\Dec(\bA,\bb_i,\bs_i,(\bc_{1i},\bc_{2i}))$ for some $f_i \in \mathcal{W}^\Dec_{zk}$ iff it holds that $\bc_{2i}-\bs_i^\top\cdot \bc_{1i}=p\cdot m_i+f_i \pmod{q}$ and $|f_i|<B^{\Enc\star}_f$. The condition $\bmm \in \LL^\SSS_{n,t}$ remains the same in $\LL^\Enc_{sound}$. Finally, we need to define  $\mathcal{W}^\Dec_{zk}$. We define it as  $\mathcal{W}^\Dec_{zk}=\{f \in \mathbb{Z}~|~|f|<B^{\Enc\star}_f\}$ for some $B^{\Enc\star}_f>B^\Enc_r\cdot B^{\Key\star}_\be+B^\Enc_e=\sqrt{u}r\cdot B^{\Key\star}\be+\sqrt{v} \beta q$ and  $B^{\Enc\star}_f<p/2$, to be specified later. With this choice, we have $\mathcal{W}^\Dec \subseteq \mathcal{W}^\Dec_{zk}$ where $\mathcal{W}^\Dec$ is in Definition \ref{subsection-pke}. With the observations above, the languages $\LL^{\Enc}=(\LL^\Enc_{zk},\LL^\Enc_{sound})$ for correct sharing are described below. 
\begin{equation*}
\begin{aligned}
\LL^{\Enc}_{zk}=&\{(\bA,n,t,(\bb_i,\bc_{1i},\bc_{2i})_{i=1}^n) ~|~ \exists~ \bs_i \in \mathbb{Z}^v,\be_i \in \mathbb{Z}^u, (\br_i~||~e_i)_{i=1}^n\in (\mathbb{Z}^{u+1}),\bmm \in \mathbb{Z}_p^n ~:~\\
&\bb_i=\bs_i^\top\cdot \bA+\be_i^\top \pmod{q},~||\bs_i||<B^{\Key\star}_{\bs},~ ||\be_i||<B^{\Key\star}_{\be}~\forall~1\leq i\leq n \\
&\land~\bc_{1i}=\bA\cdot \br_i \pmod{q},~\bc_{2i}=\bb\cdot \br_i+e_i+p\cdot m_i \pmod{q}~\forall~1\leq i\leq n\\
&\land~|e_i| \leq B^\Enc_{e},~||\br_i|| \leq B^\Enc_{\br}~\forall~1\leq i\leq n~\land~\bmm \in \LL^\SSS_{n,t}\},
\end{aligned}
\end{equation*}
\begin{equation*}
\begin{aligned}
\LL^{\Enc}_{sound}=&\{(\bA,n,t,(\bb_i,\bc_{1i},\bc_{2i})_{i=1}^n)~|~\exists~ \bs_i \in \mathbb{Z}^v,~\be_i \in \mathbb{Z}^u, (f_i)_{i=1}^n \in \mathbb{Z},~\bmm \in \mathbb{Z}_p^n ~:~ \\ &\bb_i=\bs_i^\top\cdot \bA+\be_i^\top \pmod{q},~||\bs_i||<B^{\Key\star}_{\bs},~ ||\be_i||<B^{\Key\star}_{\be}~\forall~1\leq i\leq n \\
&\land~ \bc_{2i}-\bs^\top\cdot \bc_{1i}=p\cdot m_i+f_i \pmod{q}~\forall~1 \leq i \leq n\\
& \land ~ |f_i| \leq B^{\Enc\star}_f
~\forall~1 \leq i \leq n 
\land~\bmm \in \LL^\SSS_{n,t}\}.
\end{aligned}
\end{equation*}

\noindent For the choice of $B^{\Enc\star}_f$ above so that  $\mathcal{W}^\Dec \subseteq \mathcal{W}^\Dec_{zk}$, we see that $\LL^\Enc_{zk} \subseteq \LL^\Enc_{sound}$ and the relation  $\LL^\Enc$ is consistent with the one in Section \ref{section-generic-pvss} when instantiated with the scheme of \cite{ACPS09}. Now, we assume that  $(\bA,\bb_i) \in \LL^\Key_{sound}$ already for all $i \leq n$, which is reasonable. Indeed,
for participants passed $\PVSS.\Key\Verify$ of Subsection~\ref{construction-generic-pvss}, we are convinced that there exist $\bs_i,\be_i$ such that $\be_i$ has small norm and $\bb_i=\bs_i^\top\cdot \bA+\be_i^\top \pmod{q}$, even for dishonest participants\footnote{This fits our definition of $\LL^\Enc_{sound}$ in Section \ref{section-generic-pvss} and the security definition in Figure \ref{figure-game-verifiability}.}. So the existence of $\bs_i,\be_i$ is given ``for free'', and the dealer only needs to prove that his instance satisfies \textit{the latter conditions} in $\LL^\Enc_{zk}$. Finally, for checking $\bmm \in \LL^\SSS_{n,t}$, then recall in Subsection \ref{section-special-ss} we can simply use a parity check matrix $\bH^{t}_n$ and check whether $\bmm^\top\cdot  \bH^{t}_n=0 \pmod{p}$. For the bad challenge function, we use the trapdoor $\bT$ to extract $\bs_i,\be_i$ given that $(\bA,\bb_i) \in \LL^\Enc_{sound}$. Consider the parameters $(B^\Enc_e,B^{\Enc\star}_f,B^\Enc_\br,\sigma^\Enc)$ such that $\sigma^\Enc>\sqrt{n\cdot(B^\Enc_e)^2+(B^\Enc_\br)^2}\cdot \sqrt{\log(nu+n)}$ and $2\sigma^\Enc \cdot \sqrt{u+1}\cdot (B^{\Key\star}_\be+1)<B^{\Enc\star}_f < p/2$ where $B^{\Key\star}_\be$ is the bound of $||\be||$. The trapdoor $\Sigma$-protocol is as follows.

\begin{itemize}
\item $\Trap\Sigma.\Gen_{\pp}(1^\lambda):$ Choose moduli $p,q=p^2$, dimensions $u,v$. Return $(p,q,u,v)$.
\item $\Trap\Sigma.\Gen(\pp,\LL^\Enc):$  This algorithm takes input the description of the language $\LL^\Enc$, specified by a matrix $\bA \uniformly \mathbb{Z}_q^{v \times u}$ and the bounds $B^\Enc_e,B^{\Enc\star}_f,$ $B^\Enc_\br$. Return $\crs=(\bA,p,q,u,v,B^\Enc_e,B^{\Enc\star}_f,$ $B^\Enc_\br,\sigma^\Enc)$.
\item $\Trap\Sigma.\TrapGen(\pp,\LL^\Enc,\bT):$ This algorithm takes input the description of the language $\LL^\Enc$, specified by a matrix $\bA$, a trapdoor $\bT$ such that $(\bA,\bT) \leftarrow \TrapGen(1^\lambda,v,u)$ (see Lemma \ref{theorem-invert-lwe}) and the bounds $B^\Enc_e,B^{\Enc\star}_f,$ $B^\Enc_\br$. Return $(\crs,\tr)=((\bA,p,q,u,v,B^\Enc_e,B^{\Enc\star}_f,B^\Enc_\br,\sigma^\Enc),\bT)$. 
\item $\Trap\Sigma.\Prove \langle \PP(\crs,\bx,\bw), \VV(\crs,\bx) \rangle:$ Given $(\bA,\bb_i) \in \LL^\Key_{sound}$ for all $i\leq n$. For $\bx=(\bA,n,t,(\bb_i,\bc_{1i},\bc_{2i})_{i=1}^n)$ and witness $(\br_i,e_i,m_i)_{i=1}^n$ of $\LL^\Enc_{zk}$ , the two parties construct the matrix $\bH^{t}_n$ and interact as follows.
\begin{enumerate}
\item $\PP$ samples a vector $(\bv~||~\bk) \leftarrow D_{\mathbb{Z}^{nu+n},\sigma^\Enc}$ and a random vector $~\bu=(u_1~||~\dots~||~u_n) \in \mathbb{Z}_p^n$ conditioning to $\bu^\top\cdot\bH^{t}_n=\bzero \pmod{p}$. $\PP$ then parses $\bv=( \bv_1~||~\bv_2~||~\dots~||~\bv_n)$, $\bk=(k_1~||~k_2~||~\dots~||~k_n)$ where $\bv_i \in \mathbb{Z}_q^u$. Finally, $\PP$ provides $\ba_{1i}=\bA\cdot \bv_i \pmod{q},~ \ba_{2i}=\bb_i\cdot \bv_i+k_i+p\cdot u_i \pmod{q}$ to $\VV$ for all $1 \leq i \leq n$.
\item $\VV$ samples a challenge $c \uniform \{0,1\}$ and sends $c$ to $\PP$.
\item $\PP$ provides the value $\bz_i=\bv_i+c \cdot \br_i,$ $h_i=k_i+c\cdot e_i,~$ $t_i=u_i+c\cdot m_i \pmod{p}$ to $\VV$ for all $1 \leq i \leq n$ with probability $1-K$, where $K=\text{min}\left(\frac{D_{\mathbb{Z}^{nu+n},\sigma^\Enc}(\bv~||~\bk)}{M\cdot D_{\mathbb{Z}^{nu+n},\sigma^\Enc,c\cdot (\br~||~\be') }(\bv~||~\bk)},1\right)$, $\be'=(e_1~||~e_2~||~\dots~||~e_n)$ and $M=e^{1/\log (nu+n)+12/\log^2 (nu+n)}$ ($e$ here is the base of the natural logarithm).
\item $\VV$ checks if $t_i \in \mathbb{Z}_p$, $\bA\cdot \bz_i=\ba_{1i}+c\cdot \bc_{1i} \pmod{q},~$ $\bb_i\cdot \bz_i+h_i+p\cdot t_i=\ba_{2i}+c\cdot \bc_{2i} \pmod{q}$. It also checks if $||(\bz_i~||~h_i)|| \leq \sigma^\Enc\cdot \sqrt{u+1}$ for all $1 \leq i \leq n$. Finally, check if $\bt^\top\cdot\bH^{t}_n=\bzero \pmod{p}$. Accept iff all checks pass.
\end{enumerate}
\item $\Trap\Sigma.\BadChallenge(\crs,n,t,(\bb_i,\bc_{1i},\bc_{2i})_{i=1}^n,(\ba_{1i},\ba_{2i})_{i=1}^n,\bT):$ Do as follows.
\begin{enumerate}
\item For each $1 \leq i\leq n$, compute $(\bs_i,\be_i)=\Invert(\bA,\bT,\bb_i)$ and $\bH^{t}_n$. 
\item While $0 \leq c \leq 1$, do the following:
\begin{itemize}
\item While $1\leq i \leq n$:
\begin{itemize}
\item Compute $f_{ci}=\ba_{2i}+c\cdot\bc_{2i}-\bs^\top\cdot (\ba_{1i}+c\cdot \bc_{1i}) \pmod{p}$ and cast $f_{ci}$ as integer in $[-(p-1)/2,(p-1)/2]$. $\bt_{ci}=(\ba_{2i}+c\cdot\bc_{2i}-\bs^\top\cdot (\ba_{1i}+c\cdot \bc_{1i})-f_{ci})/p \pmod{p}$. 
\item If $||f_{ci}||> B^{\Enc\star}_f/2$  return $1-c$.
\end{itemize}
\item  If  $\bt_c^\top\cdot\bH^{t}_n \neq \bzero \pmod{p}$, then also return $1-c$.
\end{itemize}
\item Otherwise, return $\perp$.
\end{enumerate}
\end{itemize}
Sampling the vector $\bu$ can be done via choosing a random polynomial $p(X) \in \mathbb{Z}_p[X]$ of degree $t$, then computing $p(i)=u_i$ for all $1 \leq i \leq n$. Below we state the theorem, assuming that $(\bA,\bb_i) \in \LL^\Key_{sound}$ already for all $i$, which is reasonable for the PVSS setting when participants have proved the validity of $\bb_i$.

\begin{theorem}\label{theorem-nizk-for-sharingg}
Consider $(B^\Enc_e,B^{\Enc\star}_f,B^\Enc_\br,\sigma^\Enc)$ specified above, and suppose that $(\bA,\bb_i) \in \LL^{\Key}_{sound}$ for all $1 \leq i \leq n$ and consider $\bs_i,\be_i$ such that $((\bA,\bb_i),(\bs_i,\be_i)) \in \RR^{\Key}_{sound}$. Then the construction is a trapdoor $\Sigma$-protocol for $\LL^{\Enc}=(\LL^{\Enc}_{zk},\LL^{\Enc}_{sound})$.
\end{theorem}

\begin{proof}[Proof Sketch] \textcolor{blue}{The general proof idea for proving zero-knowledge, special soundness, CRS indistinguishability and $\BadChallenge$ correctness is identical to Theorem \ref{theorem-nizk-for-key-generation}, thus we omit it here.} The full security proof is presented in Appendix \ref{appendix-proof-of-sharing}.
\end{proof}

\noindent \textbf{Parallel Repetition.} The scheme has soundness error $1/2$. To achieve negligible soundness error, we use the same parallel repetition strategy mentioned in the trapdoor $\Sigma$-protocol for key generation. This modification requires $\sigma^\Enc>\sqrt{\lambda(n\cdot (B^\Enc_e)^2+(B^\Enc_\br)^2)}$ and $M=1/e^{1/\log(\lambda(nu+n))+12/\log^2(\lambda(nu+n))}= 1-1/\text{poly}(u)$. Similarly, when using the NIZK compiler from the parallel repetition, we only need to recompute the first messages $O(1)$ times. The bound $B^{\Enc\star}_f$ is unchanged.

\subsection{Trapdoor \texorpdfstring{$\Sigma$}\texorpdfstring{-}Protocol for Correct Decryption}\label{subsection-nizk-for-decryption}

Finally, we design a trapdoor $\Sigma$-protocol for the correct decryption. Recall that in the decryption process of \cite{ACPS09}, for $B_f^\Dec<p/2$, we have $(m,f)=\PKE.\Dec(\bA,\bb,\bs,(\bc_1,\bc_2))$ with $|f|<B_f^\Dec$ if and only if $\bc_2-\bs^\top\cdot \bc_1=p\cdot m+f \pmod{q}$ with $|f|< B_f^\Dec$. We show how to define $\LL^\Dec$ that is consistent with the one in Section \ref{section-generic-pvss}.

First, defining $\LL^\Dec_{zk}$ is natural: We need three conditions: $\bc_2-\bs^\top\cdot \bc_1=p\cdot m+f \pmod{q}$, $\bb=\bs^\top\cdot\bA+\be^\top \pmod{q}$, and $|f|,||\bs||,||\be||$ are small. Thus $(\bs,\be,f)$ are the witnesses to prove that $(\bA,\bb,(\bc_1,\bc_2),m) \in \LL^\Dec_{zk}$ and we can define the set $\mathcal{W}^\Dec_{zk}$ in Section \ref{section-generic-pvss} to be the set of all $f$ s.t. $|f|<B^\Dec_f$ \footnote{In previous subsection, we defined $\mathcal{W}^\Dec_{zk}$ the set of all $f$ s.t. $|f|<B^{\Enc\star}_f$. For now, we try to view the two NIZKs as independent protocols. Thus, we need two separate bounds $B^{\Enc\star}_f$ and $B^\Dec_f$. When we integrate them into a PVSS later, we set $B^{\Enc\star}_f=B^\Dec_f$.}. For $\LL^\Dec_{sound}$, we define  $\mathcal{W}^\Dec_{sound}$ to be the set of all $f$ s.t. $|f|<B^{\Dec\star}_f$ for some $B^{\Dec\star}_f>B^{\Dec}_f$ determined later. Thus the language $\LL^{\Dec}=(\LL^\Dec_{zk},\LL^\Dec_{sound})$ for correct decryption is described below.
\begin{equation*}
\begin{aligned}
\LL^\Dec_{zk}=&\{(\bA,\bb,(\bc_1,\bc_2),m)~|~\exists~\bs \in \mathbb{Z}^v,~\be \in \mathbb{Z}^u,~f \in \mathbb{Z}~:~\\
&\land~\bb=\bs^\top\cdot \bA+\be^\top \pmod{q},~\bc_2-p\cdot m=\bs^\top\cdot \bc_1+f \pmod{q}\\
&\land~|f| \leq B^\Dec_f,~||\bs|| \leq B^\Dec_{\bs},~||\be|| \leq B^\Dec_{\be}\},
\end{aligned}
\end{equation*}
\begin{equation*}
\begin{aligned}
\LL^\Dec_{sound}=&\{(\bA,\bb,(\bc_1,\bc_2),m)~|~\exists~\bs \in \mathbb{Z}^v,~\be \in \mathbb{Z}^u,~f \in \mathbb{Z}~:~\\
&\land~\bb=\bs^\top\cdot \bA+\be^\top \pmod{q},~\bc_2-p\cdot m=\bs^\top\cdot \bc_1+f \pmod{q}\\
&\land~|f| \leq B^{\Dec\star}_f,~||\bs|| \leq B^{\Dec\star}_{\bs},~||\be|| \leq B^{\Dec\star}_{\be}\},
\end{aligned}
\end{equation*}

\noindent We see that $\LL^{\Dec}_{zk} \subseteq \LL^{\Dec}_{sound}$. Consider  $(B^\Dec_\bs,B^\Dec_\be,B^\Dec_f,B^{\Dec\star}_\bs,B^{\Dec\star}_\be,B^{\Dec\star}_f,\sigma^\Dec)$ s.t. $\sigma^\Dec>\sqrt{(B^\Dec_f)^2+(B^\Dec_\bs)^2+(B^\Dec_{\be})^2}\cdot \sqrt{\log (u+v+1)}$ and $~2\sqrt{u+v+1}\cdot \sigma^\Dec=B^{\Dec\star}_{\be}=B^{\Dec\star}_{\bs}=B^{\Dec\star}_f=O(p/\sqrt{v \log p})$. The trapdoor $\Sigma$-protocol is as follows.
\begin{itemize}
\item $\Trap\Sigma.\Gen_{\pp}(1^\lambda):$ Choose moduli $p,q=p^2$, dimensions $u,v$. Return $(p,q,u,v)$.
\item $\Trap\Sigma.\Gen(\pp,\LL^\Dec):$  This algorithm takes input the description of the language $\LL^\Dec$, specified by a matrix $\bA \uniformly \mathbb{Z}_q^{v \times u}$ and the bounds $B^\Dec_\bs,$ $B^\Dec_\be,$ $B^\Dec_f,$ $B^{\Dec\star}_\bs,$ $B^{\Dec\star}_\be,B^{\Dec\star}_f$. Return $\crs=(\bA,p,q,u,v,B^\Dec_\bs,$ $B^\Dec_\be,$ $B^\Dec_f,$ $B^{\Dec\star}_\bs,$ $B^{\Dec\star}_\be,B^{\Dec\star}_f,\sigma^\Dec)$.
\item $\Trap\Sigma.\TrapGen(\pp,\LL^\Dec,\bT):$ This algorithm takes input the description of the language $\LL^\Dec$, specified by a matrix $\bA$, a trapdoor $\bT$ such that $(\bA,\bT) \leftarrow \TrapGen(1^\lambda,v,u)$ (see Lemma \ref{theorem-invert-lwe}) and the bounds $B^\Dec_\bs,$ $B^\Dec_\be,$ $B^\Dec_f,$ $B^{\Dec\star}_\bs,$ $B^{\Dec\star}_\be,B^{\Dec\star}_f$. Return $(\crs,\tr)=((\bA,p,q,u,v,B^\Dec_\bs,B^\Dec_\be,B^\Dec_f,B^{\Dec\star}_\bs,B^{\Dec\star}_\be,B^{\Dec\star}_f,\sigma^\Dec),\bT)$. 
\item $\Trap\Sigma.\Prove \langle \PP(\crs,\bx,\bw),\VV(\crs,\bx) \rangle$: For a statement $\bx=(\bA,\bb,(\bc_1,\bc_2),m)$ and witness $(\bs,\be,f)$ of $\LL^\Dec_{zk}$, the two parties interact as follows.
\begin{enumerate}
\item $\PP$ samples the vectors $(\br~ ||~\bk~||~k) \leftarrow D_{\mathbb{Z}^{v+u+1},\sigma^\Dec}$, and provides the values $\bd=\br^\top\cdot \bA+\bk^\top$ and $h=\br^\top\cdot \bc_1+k$ to $\VV$. 
\item $\VV$ samples a challenge $c \uniform \{0,1\}$ and sends $c$ to $\PP$.
\item $\PP$ sends the value $\bz=\br+c\cdot \bs \pmod{q},~\bt=\bk+c\cdot \be,~t=k+c\cdot f$ to $\VV$ with probability $1-K$, where $K=\text{min}\left(\frac{D_{\mathbb{Z}^{v+u+1},\sigma^\Dec}(\bk~||~k)}{M\cdot D_{\mathbb{Z}^{v+u+1},\sigma^\Dec,c\cdot (\bs~||~\be~||~f) }(\bk~||~k)},1\right)$ where $M=e^{1/\log (v+u+1)+12/\log^2 (v+u+1)}$.
\item $\VV$ checks whether $\bz^\top\cdot \bA+\bt^\top=\bd+c\cdot \bb \pmod{q}$ and $\bz^\top\cdot \bc_1+t=h+c\cdot (\bc_2-p\cdot m) \pmod{q}$ and finally checks if $||(\bz~||\bt~||~t)|| \leq \sqrt{v+u+1}\cdot \sigma^\Dec$.
\end{enumerate}
\item $\Trap\Sigma.\BadChallenge(\crs,(\bb,
\bc_1,\bc_2,m),(\bd,h),\bT):$ While $0\leq c\leq 1$, do:
\begin{enumerate}
\item Compute $(\bz_c,\bt_c)=\Invert(\bA,\bT,\bd+c\cdot \bb)$ and $t_c=h+c \cdot (\bc_2-p\cdot m)-\bz_c^\top \cdot \bc_1 \pmod{p}$. Then cast $t_c$ as an integer in $[-(p-1)/2,(p-1)/2]$.

\item If $||\bz_c||>B^{\Dec^\star}_\bz/2$ or $||\bt_c||>B^{\Dec^\star}_\be/2$ or $||t_c||>B^{\Dec^\star}_f/2$, then return $1-c$.

\end{enumerate}
Otherwise, return $\perp$. 
\end{itemize}

\begin{theorem}\label{theorem-nizk-for-decryption}
Consider $(B^\Dec_\bs,B^\Dec_\be,B^\Dec_f,B^{\Dec\star}_\bs,B^{\Dec\star}_\be,B^{\Dec\star}_f,\sigma^\Dec)$ specified above. Then the construction is a trapdoor $\Sigma$-protocol for  $\LL^{\Dec}=(\LL^{\Dec}_{zk},\LL^{\Dec}_{sound})$.
\end{theorem}
\begin{proof}[Proof Sketch]
\textcolor{blue}{Again, the general proof idea is identical to the one in Theorem \ref{theorem-nizk-for-key-generation}, thus we omit it here.} The security proof is presented in Appendix \ref{appendix-proof-of-decryption}.
\end{proof}

\noindent \textbf{Parallel Repetition.} Similar to previous sections, to achieve negligible soundness error, we use the parallel repetition strategy. This modification requires $\sigma^\Dec>$ $\sqrt{\log(\lambda(u+v+1))} \sqrt{\lambda((B^\Dec_{\be})^2+(B^\Dec_{f})^2+(B^\Dec_{\bs})^2)}$ and $M=e^{1/\log(\lambda( (u+v+1))+12/\log^2 (\lambda(u+v+1))}$ $=1-1/\text{poly}(u)$. When using the NIZK, we only need to recompute the first messages in $O(1)$ times. The bound $B^{\Dec\star}_\bs,B^{\Dec\star}_\be,$ $B^{\Dec\star}_f$ are unchanged.

\subsection{Final Step: From Trapdoor \texorpdfstring{$\Sigma$}\texorpdfstring{-}Protocols to NIZKs}\label{subsection-compiler}
Finally, we need the protocols above to be non-interactive for the PVSS. Fortunately, there exist compilers such as~\cite{LNPT20,CX23} that transform any $\Sigma$-protocols into multi-theorem NIZKs. Among them, the compiler of \cite{CX23} relies on the decisional Diffie-Hellman assumption, thus it is not post-quantum secure. The compiler of \cite{LNPT20} is post-quantum secure as its components can be instantiated from plain LWE. The compiler, however, is a \textit{statistical} zero-knowledge NIZK, and such statistical NIZKs cannot provide adaptive soundness under falsifiable assumptions (see \cite{AF07,Pass13}). However, as pointed out by \cite[Appendix A]{LNPT20}, it is possible to bypass this result by considering languages having trapdoors that allow efficiently checking whether an element is in $\LL_{sound}$. The compiler of \cite{LNPT20} provides the \textit{simulation soundness} property for all such languages via the following theorem.  \begin{theorem}[\cite{LNPT20}, Theorems 3.4 and 4.5]\label{nizk-compiler}
For a language $\LL=(\LL_{zk},\LL_{sound})$ with a trapdoor that allows efficiently checking element in $\LL_{sound}$, assume the existence of a trapdoor $\Sigma$-protocol for $\LL$.
Then, there exists a compiler that transforms the trapdoor $\Sigma$-protocol into a simulation soundness NIZK in the CRS model for $\LL$. 
\end{theorem}

 In Appendix \ref{appendix-adaptive-soundness}, we will prove that simulation soundness also implies adaptive soundness for such languages. Also, fortunately, \textit{all our languages $\LL^\Key,\LL^\Enc,\LL^\Dec$ have the trapdoor} $\bT$ for recognizing whether a statement is valid or not with probability $1$ as follows.
 \begin{itemize}
 \item For $\LL^{\Key}_{sound}$, using $\bT$, we can extract $\bs,\be$ from $\bb$ s.t. $\bb=\bs^\top\cdot \bA+\be \pmod{q}$ and thus $(\bA,\bb) \in \LL^\Key_{sound}$ iff $||\bs|| \leq B^{\Key\star}_{\bs}, ||\be|| \leq B^{\Key\star}_{\be}$.
 \item For $\LL^{\Enc}_{sound}$, using $\bT$, we can extract $\bs_i$ and compute $(m_i,f_i)_{i=1}^n$ and thus, $(\bA,(\bb_i,\bc_{1i},\bc_{2i})_{i=1}^n) \in \LL^\Enc_{sound}$ iff $|f_i| \leq B^{\Enc\star}_f$ and $\bmm^\top\cdot \bH^{t}_n=\bzero \pmod{p}$.
 \item For $\LL^\Dec_{sound},$ using $\bT$, similarly we can extract $\bs,\be,f$ and thus $(\bA,\bb,\bc_1,\bc_2,m) \in \LL^\Dec_{sound}$ iff $|f|<B^{\Dec\star}_f, ||\bs|| \leq B^{\Dec\star}_{\bs}, ||\be|| \leq B^{\Dec\star}_{\be}$.
 \end{itemize}
 
 Thus, they are trapdoor languages. Consequently, there exists an NIZK satisfying the adaptive soundness and adaptive multi-theorem zero-knowledge property for the languages $\LL^\Key,\LL^\Enc,\LL^\Dec$ in Subsections  \ref{subsection-nizk-for-key-generation}, \ref{subsection-nizk-for-sharing}, \ref{subsection-nizk-for-decryption}, respectively.

\section{A Lattice-based PVSS Instantiation}\label{section-lattice-based-pvss}
This section will describe our lattice-based PVSS by combining the encryption scheme of \cite{ACPS09} and the concrete NIZKs in Section \ref{section-nizk}, and then plugging them into the generic construction in Section \ref{section-generic-pvss}. Note that the security of the PVSS is already implied by the generic construction in Section \ref{section-generic-pvss}. We then discuss the choice of parameters and analyze the complexities. While our PVSS works with any special secret sharing scheme, such as \cite{DI14,SARP22,ANP23}, we instantiate our scheme with the Shamir secret sharing scheme as it is the most commonly used scheme. For simplicity, whenever NIZK is needed, we use the same notation $(\NIZK_i)_{i=0}^2$ in Section \ref{section-generic-pvss} and refer the reader to their corresponding description in Section \ref{section-nizk}. 

\subsection{The Key Generation Phase}
We describe $\PVSS.\Setup,\PVSS.\Key\Gen,$ $\PVSS.\Key\Verify$. They are as follows.
\begin{itemize}
\item $\PVSS.\Setup(1^\lambda):$ Sample $\bA \uniform \mathbb{Z}_q^{v \times u}$ and the necessary public parameters~$(\crs_i)_{i=0}^2$ for the three NIZKs in the previous sections. Return $\pp=(\bA,p,q,u,v,\alpha,\beta,r,(\crs_i)_{i=0}^2)$, where $\alpha,\beta,r$ are parameters of the PKE. 
\item $\PVSS.\Key\Gen(\pp)$:  Given $\bA$, each $P_i$ proceeds as follows.
\begin{enumerate}
\item Sample two vectors $\bs_i \leftarrow D_{\mathbb{Z}^v,\alpha q},~\be_i \leftarrow D_{\mathbb{Z}^u,\alpha q}$. Repeat until $||\bs_i|| \leq \sqrt{v} \alpha q$ and $||\be_i|| \leq \sqrt{u} \alpha q$. Compute $\bb=\bs_i^\top\cdot \bA+\be^\top_i \pmod{q}$.
\item Compute the proof $\pi_i \leftarrow \NIZK_0.\Prove(\crs_0,(\bA,\bb_i),(\bs_i,\be_i))$ (from the trapdoor $\Sigma$-protocol in Subsection \ref{subsection-nizk-for-key-generation} and the compiler in Theorem \ref{nizk-compiler}). 
\item Return $((\bb_i,\bs_i),\pi_i)$.
\end{enumerate}

\item $\PVSS.\Key\Verify(\pp,\bb,\pi):$ The verifier execute $b=\NIZK_0.\Verify(\crs_0,(\bA,\bb),\pi)$ in Subsection \ref{subsection-nizk-for-key-generation} and any participant $P_i$ with $b_i=0$ is disqualified.
\end{itemize}
Note that, the value $(\bs,\be)$ is unique as long as $||\be|| \leq p/O(\sqrt{v \log p})$ by Theorem \ref{theorem-invert-lwe}. Hence,  by setting $B^{\Key\star}_{\be}=O(p \sqrt{v \log p})$, the PKE satisfies the requirement that there is \textit{at most one} $\bs$ such that $\PKE.\Key\Verify(\bA,\bb,\bs)=1$ in  Section \ref{section-generic-pvss}. Also, the NIZKs in Section \ref{section-nizk} require that $\bA$ must be generated with the trapdoor $\bT$ to achieve adaptive soundness so that we could use the security proof in Section \ref{section-generic-pvss} (we need the NIZKs to achieve adaptive soundness for the languages). Hence, we are supposed to modify the $\PKE.\Setup$ a bit by generating $\bA$ from $\TrapGen$ (we see that this modification preserves all the properties of the PKE we need in Section~\ref{section-generic-pvss} due to the distribution of $\bA$ is statistically close to a matrix $\bA$ generated by the PKE of \cite{ACPS09}). However, in the scheme above, we used the real setup algorithm by generating $\bA \uniformly \mathbb{Z}_q^{ v\times u}$ as we do not want the third party to keep any redundant trapdoor.  Fortunately, as proved in Subsection \ref{appendix-adaptive-soundness}, since the matrix $\bA$ in both cases is statistically close, all the properties of the PVSS are preserved when we generate~$\bA \uniformly \mathbb{Z}_q^{v \times u}$.

\subsection{The Sharing Phase}
This section describes $\PVSS.\Share,\PVSS.\Share\Verify$. They are as follows.
\begin{itemize}
\item $\PVSS.\Share(\pp,(\bb_i)_{i=1}^n,s,n',t):$  The dealer indexes the participants who passed the key verification process by $\{1,2\dots,n'\}$ and proceeds as follows.
\begin{enumerate}
\item Choose a random polynomial $p(X) \in \mathbb{Z}_p[X]$ of degree $t$ such that $p(0)=s$ and sets the $i$-th share as  $s_i=p(i) \pmod{p} \in \mathbb{Z}_p$ for all $1\leq i \leq n'$.
\item Sample $(\br_i)_{i=1}^n$ from $D_{\mathbb{Z^u},r}$, and $(e_i)_{i=1}^{n'}$ from $D_{\mathbb{Z},\beta q}$ (where $\beta$ is in Subsection \ref{subsection-pke}) and computes $(\bc_{1i},\bc_{2i})=(\bA\cdot \br_i \pmod{q},\bb_i\cdot \br_i+e_i+p\cdot s_i \pmod{q})$.  Repeat until $||\br_i||\leq \sqrt{u} r$ and $|e_i|\leq \sqrt{v} \beta q$ for all $1\leq i \leq n'$. 

\item Compute $\pi\leftarrow \NIZK_1.\Prove(\crs_1,(\bA,n',t,(\bb_i,\bc_{1i},\bc_{2i}))_{i=1}^{n'},(s_i,\br_i,$ $e_i)_{i=1}^{n'})$ (from the trapdoor $\Sigma$-protocol in Subsection \ref{subsection-nizk-for-sharing} and the NIZK compiler in Theorem \ref{nizk-compiler}).
\item Finally, return $((\bc_{1i},\bc_{2i})_{i=1}^{n'},\pi)$.
\end{enumerate}

\item $\PVSS.\Share\Verify(\pp,(\bb_i)_{i=1}^{n'},n',t,(\bc_{1i},\bc_{2i})_{i=1}^{n'},\pi):$ Compute $\NIZK_1.\Verify(\crs_1,$ $(\bA,n',t,$ $(\bb_i,\bc_{1i},\bc_{2i}))_{i=1}^{n'},\pi)$ as instructed in Subsection~\ref{subsection-nizk-for-sharing} and disqualify the dealer if the result is $0$.  
\end{itemize}
 Note that, in Subsection \ref{subsection-nizk-for-sharing},  even for dishonest dealers who have passed verification, it holds that $(\bc_{1i},\bc_{2i})_{i=1}^n$ are still valid encryption of a share vector $\bmm$ in the sense that: For some $B^{\Enc\star}_f<p/2$, if $(m_i,f_i)=\PKE.\Dec(\bA,\bb_i,\bs_i,(\bc_{1i},\bc_{2i}))$, then $|f_i|<B^{\Enc\star}_f$, and $\bmm \in \LL^\SSS$. By choosing $\mathcal{W}^{\Dec}_{zk}$ to be the set of all $f$ s.t. $|f|<B^{\Enc\star}_f<p/2$, then $\mathcal{W}^\Dec\subseteq \mathcal{W}^\Dec_{zk}$, which is what we need according to Subsection \ref{construction-generic-pvss}.

\subsection{The Reconstruction Phase}
This section describes $\PVSS.\Dec,\PVSS.\Combine$. They are as follows.
\begin{itemize}
\item $\PVSS.\Dec(\pp,\bb_i,(\bc_{1i},\bc_{2i}),\bs_i):$ Participant $P_i$ proceeds as follows.
\begin{enumerate}
\item Compute $f_i=\bc_{2i}-\bs^\top\cdot \bc_{1i} \pmod{q}$, then computes $f_i'=f_i \pmod{p}$ and $s_i=(f_i-f_i')/p \pmod{p}$ (we have $\bc_{2i}-\bs_i^\top\cdot \bc_{1i}=p\cdot s_i+f'_i \pmod{q}$). 
\item Compute $\pi_i\leftarrow \NIZK_2.\Prove(\crs_2,(\bA,\bb_i,(\bc_{1i},\bc_{2i}),s_i),(\bs_i,\be_i,f'_i))$ (from the trapdoor $\Sigma$-protocol in Subsection \ref{subsection-nizk-for-decryption} and the compiler in Theorem~\ref{nizk-compiler}). 
\item Return $(s_i,\pi_i)$.
\end{enumerate}
\item $\PVSS.\Dec\Verify(\pp,(\bb_i,(\bc_{1i},\bc_{2i}),s_i),\pi_i):$ Compute $\NIZK_2.\Verify(\crs_2,$ $(\bA,\bb_i,(\bc_{1i},$ $\bc_{2i}),s_i),\pi_i)$ as instructed in Subsection \ref{subsection-nizk-for-decryption}. 
\item $\PVSS.\Combine(\pp,S,(s_i)_{i \in S}):$ If $|S| \leq t$ return $\perp$.  Otherwise return $s=$ $\sum_{i \in S'} \lambda_{i,S'} s_i \pmod{p}$ where $\lambda_{i,S'}=\prod_{j \in S',j \neq i} j/(j-i) \pmod{p}$.
\end{itemize}

By setting $B^{\Enc\star}_f=B^\Dec_f$ and choosing $\mathcal{W}^\Dec_{sound}$ to be the set of all $f \in \mathbb{Z}$ such that $|f|<B^{\Dec\star}_f$ as in Subsection \ref{subsection-nizk-for-decryption} (in particular, $B^{\Dec}_f<B^{\Dec\star}_f$), then $\mathcal{W}^\Dec_{zk} \subseteq \mathcal{W}^\Dec_{sound}$, which satisfies what we need in Section \ref{section-generic-pvss}. Finally, we conclude that the PVSS is a secure PVSS in the CRS model according to   Section~\ref{section-generic-pvss} and Theorems \ref{theorem-nizk-for-key-generation},  \ref{theorem-nizk-for-sharingg}, \ref{theorem-nizk-for-decryption}, \ref{nizk-compiler}.

\subsection{Parameters Setting}\label{appendix-parameter-setting}
We now summarize the parameters when combining the encryption scheme and previous NIZKs into the unified PVSS. For the encryption scheme in \cite{ACPS09}, we have the parameters $(u,v,p,q,r,\alpha,B^\Key_\bs,B^\Key_\be,B^\Enc_{\be},B^\Enc_\br)$,  the encryption scheme is secure if we choose  $v=\Omega(\lambda)$, $\alpha q=\Omega(\sqrt{v})$, $B^\Key_{\bs}=\sqrt{v}\cdot \alpha q$, $B^\Key_{\be}=\sqrt{u}\cdot \alpha q$, and $B^\Enc_\br=\sqrt{u}\cdot r$ with $r=\omega(\log u)$, and $B^\Enc_e=r  u \alpha q $. In addition, we also choose $u=O(v \log q)$ for Lemma \ref{theorem-invert-lwe} to work. For the NIZK, recall that we have additional parameters $(B^{\Key\star}_\bs,B^{\Key\star}_\be,\sigma^\Key,B^{\Enc\star}_f,B^\Enc_\br,\sigma^\Enc,\textcolor{blue}{B^\Dec_\bs},B^\Dec_\be,B^\Dec_f,\textcolor{blue}{B^{\Dec\star}_\bs},B^{\Dec\star}_\be,B^{\Dec\star}_f,\sigma^\Dec)$ for key generation, sharing and decryption. 

For key generation, we use parallel repetition $\lambda$ times of trapdoor $\Sigma$-protocol and need $\sigma_\Key>\sqrt{\lambda((B^{\Key}_{\be})^2+(B^{\Key}_{\bs})^2)}\cdot \sqrt{\log (\lambda(u+v))}$ and $2 \sqrt{u+v}\cdot \sigma^\Key<\textcolor{blue}{B^{\Key\star}_{\bs}}=B^{\Key\star}_{\be}<O(p/\sqrt{v \log p})$ for the NIZK to work, as specified in Subsection \ref{subsection-nizk-for-key-generation}. \textcolor{blue}{By setting} $u=v \log p$, $\alpha q=O(\sqrt{v})$ \textcolor{blue}{and $v=O(\lambda)$ as above}, we see that \textcolor{blue}{$B^\Key_\bs=O(v), B^\Key_\be=O(v \sqrt{\log p})$}, thus the value $2 \sqrt{u+v}\cdot \sigma^\Key$ (and thus \textcolor{blue}{the minimal value of} $B^{\Key\star}_{\be}$) is at most $O(v^{2}\cdot (\log v)^{0.5}\cdot (\log p))$. Consequently, the minimal required modulus for the key generation process is $p= \tilde{O}(v^{2.5})$. 

Next, for correct sharing, we use parallel repetition $\lambda$ times of trapdoor $\Sigma$-protocol. In the worst case (some participants and dealer are dishonest), the parameters $(B^\Enc_e,B^{\Enc\star}_f,B^\Enc_\br,\sigma^\Enc)$ must satisfy $\sigma^\Enc>\sqrt{\lambda(n\cdot (B^\Enc_e)^2+(B^\Enc_\br)^2)}\cdot \sqrt{\log(\lambda(nu+n))}$ and $2\sigma^\Enc \cdot \sqrt{u+1}\cdot (B^{\Key\star}_\be+1)<B^{\Enc\star}_f\leq p/2$ (this also satisfies the required bound for the encryption correctness property in Subsection \ref{subsection-pke}) and $B^{\Key\star}_\be=O(p/\sqrt{v \log p})$. These bounds are specified in Subsection \ref{subsection-nizk-for-sharing}. \textcolor{blue}{By computing these bounds in terms of $v,p,n,r$ with $u=v\log p$ and $v=O(\lambda)$, we achieve the minimal bounds} $B^\Enc_e=O(\textcolor{blue}{r}v^{1.5} \cdot \log p)$, $B^\Enc_\br=O(\textcolor{blue}{r}\sqrt{v} \cdot (\log p)^{0.5})$, $B^{\Key\star}_\be=O(v^{2} \cdot (\log v)^{0.5} (\log p))$. \textcolor{blue}{By plugging these values to compute $\sigma^\Enc$,  we see that $\sigma^\Enc=O(r\cdot v^2\cdot n^{0.5}\cdot \log p (\log v)^{0.5} \sqrt{\log n})$. As a result, we achieve the minimal bound $B^{\Enc\star}_f=O(r\cdot v^{4.5}\cdot  n^{0.5}\cdot  (\log p)^{2.5} \log v \sqrt{\log n})$}.  \textcolor{blue}{Finally, since $r=\omega(\log u)$}, the minimal required modulus for the key sharing process is~$p= \tilde{O}(v^{4.5}\cdot n^{0.5})$.

Finally, for correct decryption, we use parallel repetition $\lambda$ times and require that $\sigma^\Dec>\sqrt{\lambda((B^\Dec_\bs)^2+(B^\Dec_f)^2+(B^\Dec_{\be})^2)}\cdot \sqrt{\log (\lambda(u+v+1))}$ and $\sqrt{u+v+1}\cdot\sigma^\Dec=\textcolor{blue}{B^{\Dec\star}_\bs}=B^{\Dec\star}_{\be}=B^{\Dec\star}_f=O(p/\sqrt{v \log p})$ according to Subsection \ref{subsection-nizk-for-decryption}, where $B_f^\Dec=B^{\Enc\star}_f$, $B_{\bs}^\Dec=B^{\Key\star}_\bs$ and $B_{\be}^\Dec=B^{\Key\star}_\be$.  This is because for an honest participant, when $(\bA,(\bb_i,\bc_{1i},\bc_{2i})_{i=1}^n) \in \LL^{\Enc}_{sound}$, then upon decryption, it receives a witness $f_i$ of norm at most $B^{\Enc\star}_f$, thus $B_f^\Dec=B^{\Enc\star}_f$. Also, the bound of $\be_i$ is simply $B^{\Key}_\be$ for honest participants, hence $B_{\be}^\Dec=B^{\Key}_\be$. For a dishonest participant, if it passes verification, then the bound of $\be_i$ and $f_i$ is equal to $B^{\Dec\star}_{\be}$ and $B^{\Dec\star}_f$ respectively, which implies correct decryption as long as they are smaller than $p/2$.  \textcolor{blue}{By computing these bounds in terms of $v,p,n,r$, we achieve the minimal bounds} $B^{\Dec}_\bs=B^{\Key\star}_{\bs}=B_{\be}^\Dec=B^{\Key\star}_\be=O(v^{2} \cdot (\log v)^{0.5} \log p)$, and  $B^{\Dec}_f=B^{\Enc\star}_{f}=O(\textcolor{blue}{r}\cdot v^{4.5} \cdot n^{0.5}\cdot  \log v (\log p)^{2.5} \sqrt{\log n})$. \textcolor{blue}{By plugging these values to compute $\sigma^\Dec$  we see that $\sigma^{\Dec}=O(r\cdot v^{5} \cdot n^{0.5}\cdot (\log v)^{1.5} (\log p)^{2.5} \sqrt{\log n})$. As a result, we achieve the minimal bound $B^{\Dec\star}_\bs=B^{\Dec\star}_\be=B^{\Dec\star}_f=O(r\cdot v^{5.5} \cdot n^{0.5}\cdot (\log v)^{1.5} (\log p)^{3} \sqrt{\log n})$}. Thus, the minimal required modulus for the decryption process is $p= \tilde{O}(v^{\textcolor{blue}{6}}\cdot n^{0.5})$.  

Also, note that if $B_\be^{\Key\star}<O(p/\sqrt{v \log p})$ then the LWE function is injective, according to Theorem \ref{theorem-invert-lwe}, thus for any $\bb$, there exists a unique $\bs$ such that $\PKE.\Key\Verify(\bA,\bb,\bs)=1$, according to our requirement. According to the analysis above, we can choose $p=\tilde{O}(v^{6}\cdot n^{0.5})$, hence  $q=\tilde{O}(v^{\textcolor{blue}{12}}\cdot n)$ to make the PVSS work (but one can choose $q$ to be sub-exponential in $v$ if large secret is required). \textcolor{blue}{As a result, we have specified all the parameters required in our PVSS.}

\subsection{Complexity Analysis}\label{section-complexity-analysis}
We analyze the communication and computation complexity of i) the trapdoor $\Sigma$-protocols and ii) the other operations of the PVSS, including computing the shares, encrypting, decrypting, and computing recovery coefficients. We refer to the former as the trapdoor $\Sigma$-protocol cost and the latter as other costs. We separate these costs because, while we can compute the exact cost of other operations in terms of $n,u,v,\log q$, we can only estimate the cost of the NIZK to be at least the cost of trapdoor $\Sigma$-protocols. This is because the actual NIZK cost depends on the compiler, and the compiler of \cite{LNPT20} requires many complicated components and additional parameters that make it hard to give a concrete cost. However, currently we can still estimate the NIZK compiler to be at least the cost of the trapdoor $\Sigma$-protocols, and the more efficient the trapdoor $\Sigma$-protocols (and the size of the instances), the more efficient the NIZKs.  The complexities are summarized in Table~\ref{table:summary}. We compare our work with the technique using Karp reduction in Appendix \ref{appendix-comparison-with-generic-sols}. For computation complexity, we need the following lemmas.
\begin{lemma}[\cite{TCZAPGD20}, Section $3$]\label{lemma-fast-computation}
Let $p$ be a prime number and let $a_1,a_2,\dots,a_n$ be distinct elements over $\mathbb{Z}_p$. Computing $\lambda_{i}=\prod_{i \neq n} a_j/(a_i-a_i) \pmod{p}$ can be done withing $O(n \log^2 n)$ arithmetic operations.
\end{lemma}

\begin{lemma}[\cite{KT73}, Theorem $4.3$]\label{lemma-fast-computation-part-two}
For a polynomial $p(X)$ of degree at most $n$, computing $(p(i))_{i=1}^n$ can be done within $O(n \log^2 n)$ arithmetic operations.
\end{lemma}

\noindent \textbf{Communication Complexity.} The detailed cost is as follows.

\begin{itemize}
\item In the key generation phase, each participant submits a public key of size $O( u \log q)$ (other costs) and a proof of size $O(\lambda\cdot u \log q)$ (trapdoor $\Sigma$-protocol cost). The total cost will be multiplied by $n$, since there are $n$ participants.
\item In the sharing phase, the dealer has to submit $n$ encryptions, each of size $O((v+1) \log q)$ (other costs). It also needs to submit a proof of size $O(\lambda (u+v) n \log q)$ (trapdoor $\Sigma$-protocol cost). The total cost will be multiplied by $n$.
\item In the reconstruction phase, each participant needs to submit the decrypted share of size $O(\log q)$ (other costs) and a proof of size $O(\lambda\cdot v \log q)$ (trapdoor $\Sigma$-protocol cost). The total cost will be multiplied by $n$.
\end{itemize}

\noindent \textbf{Computation Complexity.} The detailed cost is as follows.

\begin{itemize}
    \item In $\PVSS.\Key\Gen$, computing the public key requires $O(\lambda uv)$ operations (other costs). The cost of the proof is dominated by computing $\br^\top\bA+\bff^\top$ in $\lambda$ times, which requires $O(\lambda uv)$ arithmetic operations (trapdoor $\Sigma$-protocol cost). In $\PVSS.\Key\Verify$, verifying each proof requires $O(\lambda uv)$ cost (trapdoor $\Sigma$-protocol cost). Since we verify $n$ proofs, the cost is $O(n\lambda uv )$ (trapdoor $\Sigma$-protocol cost).
    
    \item In $\PVSS.\Share$, the cost is dominated by computing the shares (other costs) and NIZK first messages $(\ba_{1i},\ba_{2i})_{i=1}^n$ in $\lambda$ times (trapdoor $\Sigma$-protocol cost). For a polynomial $p(X)$ of degree at most $n$, computing $(p(i))_{i=1}^n$ requires $O(n \log^2 n)$ operations due to Lemma \ref{lemma-fast-computation-part-two}.  Computing the first message requires $O(\lambda(n^2+nuv))$ cost, dominated by the cost of sampling $\bu$ (which also requires $O(n\log^2 n)$ operations due to Lemma \ref{lemma-fast-computation-part-two}) and computing $\ba_{1i}=\bA\cdot \br_i \pmod{q}$.
    \item In $\PVSS.\Share\Verify$, the cost is dominated by~computing $\bA\cdot \bz_i \pmod{q}$ and $\bt^\top\cdot \bH^{t}_n\pmod{p}$ in $\lambda$ times. The former requires  $O(\lambda n\cdot uv)$ operations, while the latter trivially requires  $O(\lambda n^2)$ arithmetic operations. 
    \item Decrypting the shares requires  $O(v)$ arithmetic operations (other costs), and the computation cost of the proof is  $O(\lambda uv )$ (trapdoor $\Sigma$-protocol cost). The verification cost of in $\PVSS.\Dec\Verify$ is also $O(\lambda uv )$ (trapdoor $\Sigma$-protocol cost), however we need to verify $n$ proofs, so it will be multiplied by $n$.
    \item Finally, reconstructing the secret $s=\sum_{i \in S} \lambda_{i,S}\cdot  s_i \pmod{p}$ in $\PVSS.\Combine$ costs $O(n \log^2 n)$ arithmetic operations, dominated by computing $\lambda_{i,S}=\prod_{j \in S, j \neq i} j/(j-i) \pmod{p}$ due to Lemma \ref{lemma-fast-computation} (other cost only). 
\end{itemize}

\begin{table}
\centering
	\caption{\footnotesize{ Summary of complexities. $\Sigma$-Pcls refers to trapdoor $\Sigma$-protocol. We measure the computation complexity by the number of arithmetic operations (addition, multiplication, inversion, comparison) over $\mathbb{Z}_q$ and $\mathbb{Z}_p$ required of a single participant. Others refer to operations (computing shares, encryption, decryption, reconstructing shares) that are not proving and verifying trapdoor $\Sigma$-protocols.}}
	\label{table:summary}
    \vspace{0.2cm}
	\scalebox{0.72}{\begin{tabular}{lllll}
	\hline
	\hline
	& \textbf{Comm. (Others)} & \textbf{Comm. ($\Sigma$-Pcls)} & \textbf{Comp. (Others)}. &  \textbf{Comp. ($\Sigma$-Pcls)}.   \\
	\hline
    Key Generation & $O(n u \log q)$ & $O(\lambda n (u+v) \log q)$ & $O(u v )$ & $O(\lambda u v )$ \\
    Key Verification & N/A & N/A & N/A & $O(\lambda nu v )$ \\
    Sharing & $O( n v \log q)$  & $O(\lambda n (u+v) \log q)$ & $O(n \log^2 n+ n uv  )$ & $O(\lambda (n \log^2 n+ n uv))$ \\
Share Verification & N/A  & N/A & N/A & $O(\lambda (n^2+ n uv) )$ \\
    
    Share Decryption & $O(n \log q)$ & $O(\lambda n u \log q)$ & $O(v  )$ & $O(\lambda  u v )$ \\
Dec. Verification & N/A & N/A & N/A & $O(\lambda n u v )$ \\
    
    Reconstruction & N/A & N/A & $O(n \log^2 n )$ &  N/A \\
    
    Total & $O((nv+u) \log q)$ &  $O(\lambda n (u+v) \log q )$ & $O(n \log^2 n+nuv)$  & $O(\lambda (n^2+nuv))$   \\
	\hline
	\hline
	\end{tabular}}
\end{table}

\subsection{Final Remark on not Using Amortized Encryption}\label{subsection-final-remark}
 In \cite{GHL22}, the authors provided a PKE for multi-receiver in an attempt to reduce the communication complexity of encryption from $O(nv \log q)$ to $O((n+v) \log q)$ \footnote{This does not include the cost of the trapdoor $\Sigma$-protoocols, only the PKE.} and computation complexity from $O(nuv)$ to $O(uv+nu)$ operations. The idea is that, instead of generating $n$ values $\bc_{1i}=\bA\cdot \br_i \pmod{q}$, the authors only need one common value $\bc_1=\bA\cdot \br \pmod{q}$ and thus $\bc_{2i}=\bb_i\cdot \br+e'_i+p\cdot m  \pmod{q}$. However, it requires a modulus  $q=2^{\Omega(\lambda)}$ and the value $v$ also needs to be increased to $O(\lambda^{1+\epsilon})$ (for some $\epsilon>0$). Hence, while the number of scalars and the number of operation are reduced as above, the value $\log q$ is increased from $O(\log (\lambda n))$ to $O(\lambda)$, and the value $v$ is increased from $O(\lambda)$ to $O(\lambda^{1+\epsilon})$, meaning that the total (bit) complexities are not reduced, especially the communication complexity of trapdoor $\Sigma$-protocols and the computation complexity when $u=\Theta(v \log q)$ and the complexity of each arithmetic operation depends on $\log q.$ This will cause inefficiency when we only need to share small secrets \footnote{Indeed, the bit complexity of operations is $O(\log^2 q)$, see \cite[Subsection 2.4.4]{MV18}, and best optimizations \cite{H21} might not be lower than $\tilde{\Omega}(\log q)$. This implies that the time complexity will worsen if $q$ is increased to $2^{\Omega(\lambda)}.$ For example, the bit computation complexity of the PKE alone will be increased from $O(n\lambda^2 \log^2(n\lambda))$ to $O(n \lambda^{4+\epsilon})$ if we assume each operation takes $O(\log^2 q)$ bit complexity.}.

To allow  $q=\text{poly}(\lambda,n)$, the authors rely on strong assumptions and heuristics (For example, their parameters are not sufficient to imply decisional LWE security. However, the authors make a strong assumption that it would imply security. They also rely on inequalities based on heuristics. See Subsections $2.3$, $3.2$ and Appendix A.4 of their work). If we apply their techniques to get an amortized scheme with polynomial modulus, we need to rely on these strong assumptions. Hence, we choose the non-amortized scheme to ensure that the required modulus is still polynomial and only needs to rely on the standard decisional LWE assumption.

\section{Conclusion}
In this work, we proposed the first non-trivial lattice-based PVSS from standard assumptions. The NIZKs in our PVSS are proven in the CRS model and are \textit{not} the result of merely using the generic Karp reduction technique. Hence, our scheme is the \textit{most efficient solution} for i) post-quantum security and ii) achieving security in the standard model. Although the required modulus $q$ is polynomial in $n,v$, it is still quite large (see Subsection \ref{appendix-parameter-setting}). Hence, it would be desirable to think of optimizations for the trapdoor $\Sigma$-protocols to reduce the required size of $q$. We leave the work of optimization for the trapdoor $\Sigma$-protocols and the amortized encryption scheme with polynomial modulus to future work.

\vspace{0.75cm}

\noindent\textbf{CRediT authorship contribution statement.} The authors confirm contribution to the paper as follows: Conceptualization, Pham Nhat Minh, Khoa Nguyen, and Khuong Nguyen-An; methodology, Pham Nhat Minh, Khoa Nguyen, and Khuong Nguyen-An; formal analysis, Pham Nhat Minh, Khoa Nguyen, and Khuong Nguyen-An; writing-original draft preparation, Pham Nhat Minh, Khoa Nguyen, and Khuong Nguyen-An; writing-review and editing, Pham Nhat Minh, Khoa Nguyen, Willy Susilo, and Khuong Nguyen-An; supervision, Khoa Nguyen, Willy Susilo, Khuong Nguyen-An; project administration, Khuong Nguyen-An. All authors reviewed the results and approved the final version of the manuscript.\\

\noindent\textbf{Declaration of Competing Interest.} The authors declare that they have no known competing financial
interests or personal relationships that could have appeared to influence the work reported in this paper.\\

\noindent\textbf{Data availability.} No data was used for the research described in the article.\\

\noindent\textbf{Acknowledgement.} During the preparation of this work, Nhat-Minh Pham and Khuong Nguyen-An are supported by Ho Chi Minh City University of Technology (HCMUT), VNU-HCM under grant number T-KHMT-2024-02. 

\vspace{1cm}
\begin{center}
{\noindent \textbf{\large{Appendices}}}
\end{center}
\appendix
\renewcommand{\thesection}{\Alph{section}}

\section{Comparison with Generic Technique using Karp Reduction} \label{appendix-comparison-with-generic-sols}

 Our solution is better than the trivial trapdoor $\Sigma$-protocol solution of using Karp reduction. Indeed, by using the techniques of \cite{CCHLRRW19,PS19} for the same problem, We need to encode the LWE statement into a CNF-SAT circuit with $\Omega(u\cdot v\cdot \log^2 q)$ clauses and $\Omega(v\cdot  \log q)$ variables, and reducing to Graph Hamiltonicity language would require a graph with size $\Omega(u^3 \cdot \log q)$ \cite[Appendix F.2]{LNPT20}. Finally, representing the graph with an adjacent matrix would require $\Omega(u^6 \cdot \log^2 q)$ bits (Here, it is known that there is a \textcolor{blue}{(indirect) reduction from CNF-SAT with $a$ clauses, $b$ variables to graph Hamiltonicity through other problems}, and the resulting graph is estimated to have at least $\Omega(ab)$ vertices, see \cite[Chapter 10.5]{AHU74} or \cite{Cormen22}). Thus, both the communication and computation complexity for each iteration of the trapdoor $\Sigma$-protocol for LWE using Karp reduction cannot be lower than $\Omega(u^6 \cdot \log^2 q)$, because we need to at least encrypt the same number of bits and publish the encryption. Ours only needs $O((u+v) \log q)$ communication and $O(uv)$ computation complexity. 
 
 Similarly, if the generic Karp reduction technique is used for correct sharing, this would require a circuit of size $\Omega((n^2u^2v+n^3u)^2\log^6 q)=\Omega(n^4u^6 \log^4q+2n^5u^4 \log^5q+n^6u^2 \log^6 q)$ (Here the number of clauses is $O((n uv+n^2)\cdot  \log^2 q)$ and the number of variables is $O(n u\cdot  \log q)$. This is due to $\Omega(n)$ conditions, \textcolor{blue}{each having the form $\bc_1=\bA\cdot \br \pmod{q}$, $\bc_2=\bb\cdot \br+e+p\cdot m \pmod{q}$, and an additional condition  $\mathbf{m}^\top\cdot \bH_n=0 \pmod{p}$} that incurs $\Omega(n^2 \log^2 q)$ bit clauses), which is much worse compared to~$O( n(u+v) \log q)$ communication and $O(n^2+nuv)$ computation complexity of ours. Finally, the complexities of trapdoor $\Sigma$-protocols for decryption using Karp reduction would also be $\Omega(u^6 \cdot \log^2 q)$ as well. Hence, our solution would be better than generic solutions.

\section{Formal Security Proof of the Generic PVSS}\label{appendix-pvss-security}

\subsection{Proof of Theorem \ref{theorem-correctness}}\label{appendix-pvss-correctness}

\begin{proof}[Proof of Theorem \ref{theorem-correctness}]
Indeed, recall in Figure \ref{figure-game-correctness}, it suffices to prove the following.
\begin{itemize}
\item If $((\pk_i,\sk_i),\pi_i) \leftarrow \PVSS.\Key\Gen(\pp)$, then $\PVSS.\Key\Verify(\pp,\pk_i,\pi_i)=1$ with probability $1-\negl(\lambda)$.
\item If $(E,\pi) \leftarrow \PVSS.\Share(\pp,(\pk_i)_{i=1}^{n'},s,n',t)$, then $\PVSS.\Share\Verify(\pp,$ $(\pk_i)_{i=1}^{n'},n',$ $t,$ $E,\pi)=1$  with probability $1-\negl(\lambda)$.
\item If $(s_i,\pi_i) \leftarrow \PVSS.\Dec(\pp,\pk_i,E_i,\sk_i)$, $\PVSS.\Dec\Verify(\pp,\pk_i,E_i,s_i,\pi_i)=1$  with probability $1-\negl(\lambda)$.
\item Finally, if $s_1,s_2,\dots,s_{n'}$ is shared by the dealer, and for a subset $S$ of participants that passed $\PVSS.\Dec\Verify$, then the decrypted share must be $s_i$ for all $i \in S$ with probability $1-\negl(\lambda)$. In addition, $|S| \geq t+1$ and for any $S' \subseteq S$ with $|S'| \geq t+1$, $\PVSS.\Combine(\pp,S',(s_i)_{i \in S'})$ returns $s$.
\end{itemize}

For the first part, if $(\pk_i,\sk_i) \leftarrow \PVSS.\Key\Gen(\pp,r_i)$, then $(\pp',\pk_i) \in \LL^{\Key}_{zk}$ with some witness $(\sk_i,r_i)$. Since $\PVSS.\Key\Gen$ honestly produces the proof $\pi_i$ from $\NIZK_0.\Prove$ and $\PVSS.\Key\Verify$ executes $\NIZK_0.\Verify$, thus $\PVSS.\Key\Verify$ returns $1$ with probability $1-\negl(\lambda)$ for each $P_i$ with $i \not \in \CC$. 

Before the second part, we prove that $(\pp',\pk_i) \in \LL^\Key_{sound}$ for all $i \in G=[n']$ with probability $1-\negl(\lambda)$. It is trivial for $i \in G \setminus \CC$ (honest participants) due to the key correctness property ($\LL^\Key_{zk} \subseteq \LL^\Key_{sound}$, due to Definition \ref{key-correctness}). It suffices to consider $i \in G\cap \CC$. Indeed, for a fixed $i \in \CC$, consider the probability that  $(\pp',\pk_i) \not \in \LL^\Key_{sound}$ and $\PVSS.\Key\Verify$ returns $1$. Due to the adaptive soundness property, it is negligible. Hence the probability that there is some $i \in \CC$ s.t. $(\pp',\pk_i) \not \in \LL^\Key_{sound}$ and $\PVSS.\Key\Verify$ returns $1$ is negligible due to union bound. So, for all $i \in \CC$, it must hold that $\PVSS.\Key\Verify$ returns $0$ or $(\pp',\pk_i) \in \LL^\Key_{sound}$ with  probability~$1-\negl(\lambda)$. Since $\PVSS.\Key\Verify$ returns $1$ for exactly those $i \in G\cap \CC$, we conclude that $(\pp',\pk_i) \in \LL^\Key_{sound}$ for all those $i$ with probability $1-\negl(\lambda)$.

For the second part, if $(\pp',n',t,(\pk_i,E_i)_{i=1}^{n'})) \in \LL^{\Enc}_{zk}$ with some witness $(s_i,r_i)_{i=1}^{n'}$ and $\pi_i\leftarrow \NIZK_1.\Prove(\crs_1,(\pp',n',t,(\pk_i,E_i)_{i=1}^{n'}),(s_i,r_i)_{i=1}^{n'})$, then it holds that $\NIZK_1.\Verify(\crs_1,(\pp',n',t,(\pk_i,E_i))_{i=1}^{n'},\pi_i)=1$ due to the correctness of the NIZK. Since previously we have $(\pp',\pk_i) \in \LL^\Key_{sound}$ for all $i \in [n']$ and $\PVSS.\Share$ produces the encryption $E_i$ of $s_i$, we have $(\pp',n',t,(\pk_i,E_i)_{i=1}^{n'})) \in \LL^{\Enc}_{zk}$ with corresponding witness $(s_i,r_i)_{i=1}^{n'}$ \footnote{Again, recall that while we require $(\pp',\pk_i) \in \LL^\Key_{sound}$, it has been proved by $P_i$ already. We design our NIZK so that the dealer only needs the witnesses $(s_i,r_i)_{i=1}^{n'}$ to prove the remaining conditions. A concrete example is just in Subsection \ref{subsection-nizk-for-sharing}.}.  Finally, the algorithm provides correct proof~$\pi$ using $\NIZK_1.\Prove$ and $\PVSS.\Share\Verify$ simply executes $\NIZK_1.\Verify$, the returned output is $1$ with probability $1-\negl(\lambda)$.

For the third part, since $((\pp',\pk_i),(\sk_i,r_i)) \in \RR^\Key_{zk}$ and $\PVSS.\Dec$ outputs $(s_i,w_i)=\PKE.\Dec(\pp',E_i,\pk_i,\sk_i)$, we have $w_i \in \mathcal{W}^\Dec \subseteq \mathcal{W}^\Dec_{zk}$. So $(\pp',\pk_i,E_i,s_i) \in \LL^{\Dec}_{zk}$ with witness $(\sk_i,r_i,w_i)$. Finally, since $\PVSS.\Dec$ uses $\NIZK_2\Prove$ to produce $\pi_i$, while $\PVSS.\Dec\Verify$ executes $\NIZK_2.\Verify$, hence $\PVSS.\Dec\Verify$ returns $1$ with probability $1-\negl(\lambda)$ for each $i \not \in \CC$.

Finally, suppose $(s_1,s_2,\dots,s_{n'})$ are the shares of $s$ shared by the dealer. For honest participants, they would honestly generate $(\pk_i,\sk_i)\leftarrow \PVSS.\Key\Gen(\pp)$. When $E_i$ is the encryption of $s_i$, it holds that $(s_i,w_i)=\PKE.\Dec(\pp,\pk_i,\sk_i,E_i)$ for some $w_i \in \mathcal{W}^\Dec_{zk}$ (Definition \ref{definition-encryption-correctness}). Thus an honest participant would receive a correct $s_i$ s.t. $(\pp',\pk_i,E_i,s_i) \in \LL^\Dec_{zk}$ and makes $\PVSS.\Dec\Verify$ returns $1$ like analyzed above. Now, for each dishonest participant who publishes  $s'_i$ and has passed verification, then there exists some $\sk'_i$ s.t. $\PVSS.\Key\Verify(\pp',\pk_i,\sk'_i)=1$ and $(s'_i,w'_i)=\PKE.\Dec(\pp',\pk_i,\sk'_i,E_i)$ for some $w_i \in \mathcal{W}^\Dec_{sound}$ due to the soundness of the NIZK with probability $1-\negl(\lambda)$. Due to the uniqueness of $\sk'_i$, we have $\sk'_i=\sk_i$ and thus  $s'_i=s_i$. Therefore, participants who passed verification must have published the exact shares $s_i$'s. Note that the set $S$ consists of at least $n-t\geq t+1$ participants, and $\PVSS.\Combine$ returns the original secret $s$ for any set $S' \subseteq S$ of size $t+1$ due to the correctness property of $\SSS$. 

In conclusion, the probability that  $\GAME^{\mathsf{PVSS-Correctness}}(\Adversary,s)$ returns $0$ is negligible. Hence, the PVSS satisfies the correctness property, as desired.
\end{proof}

\subsection{Proof of Theorem \ref{theorem-verifiability}}\label{appendix-pvss-verifiability}

\begin{proof}[Proof of Theorem \ref{theorem-verifiability}]
First, $\LL^\SSS_{t}$ is a valid share language  (Definition \ref{definition-well-formedness}) due to the correctness $\SSS$: It is the union of all $(\LL^\SSS_{i,t})_{i=t+1}^n$ where $\LL^\SSS_{n,t}$ is the set of all $(s_i)_{i=1}^n$ s.t. $(s_i)_{i=1}^n\leftarrow \SSS.\Share(s,n,t)$. Thus for any set $S$ with size at least $t+1$, $\SSS.\Combine(S,(s_i)_{i \in S})=s$, so each $\LL^\SSS_{i,t}$ is a valid share set and also so is $\LL^\SSS_{t}$. Also, the fact that $\LL^\Key_{sound}$ has a unique witness is trivial due to our requirement of the PKE. We now focus on the last property. We see that, in Figure \ref{figure-game-verifiability}, if the game returns $1$, then at least one of the following events must happen:
\begin{itemize}
\item The adversary $\Adversary$ could produce some $(\pk_i,\pi_i)$ s.t. $(\pp',\pk_i) \not \in \LL^{\Key}_{sound}$ and $\PVSS.\Key\Verify$ returns $1$ for some $i \in \CC$.
\item It holds that $(\pp',\pk_i) \in \LL^\Key_{sound}$ for all $i \in G$ with corresponding secret key $\sk_i$ as witness, however $\Adversary$ could produce some $(E,\pi)$ s.t: If   $(s_i,w_i)=\PVSS.\Dec(\pp,\pk_i,\sk_i,E_i)$ for all $i \in S_1=[n']$, then $(s_1~||~s_2~||~\dots~||~s_{n'}) \not \in \LL^{\SSS}_{t}$, and $\PVSS.\Share\Verify(\pp',(\pk_i)_{i=1}^{n'},n',t,E,\pi)=1$. 
\item The value $s'_i$ revealed by $P_i$ is not equal to the honestly decrypted share $s_i$ but $\PVSS.\Dec\Verify(\pp,\pk_i,E_i,s'_i,\pi^2_i)=1$ for some $i \in\CC$. 
\item It could produce some $(E,\pi)$ s.t. $(\pp,(\pk_i,E_i)_{i=1}^{n'}) \in \LL^\Enc_{sound}$ s.t. $(s'_i,\pi_i)\leftarrow \PVSS.\Dec(\pp,\pk_i,\sk_i,E_i)$ but $\PVSS.\Dec\Verify(\pp,\pk_i,E_i,s'_i,\pi^2_i)=0$.

\end{itemize}

It suffices to prove that the probability that each event above could happen is negligible. Thus, the probability that at least one of them happens is also negligible.

First, consider a fixed $i \in \CC$. Due to the adaptive soundness definition of the NIZK, if $(\pp',\pk_i) \not \in \LL^\Key_{sound}$ then $\NIZK_0.\Verify$ returns $1$ with negligible probability. So by union bound over all $i \in \CC$, the first event happens with negligible probability.

Second, suppose $(E=(E_i)_{i=1}^{n'},\pi)$ is provided by the dealer and $s_i=\PVSS.\Dec(\pp,\pk_i,\sk_i,E_i)$ for some $\sk_i$ such that $((\pp',\pk_i),\sk_i) \in \RR^\Key_{sound}$. Due to the definition of $\LL^{\Enc}_{sound}$ and the soundness property of $\NIZK_1$, if $(\pp',(\pk_i,E_i)_{i=1}^{n'}) \not \in \LL^\Enc_{sound}$, the probability that $\NIZK_1.\Verify(\crs_1,(\pp',n,t,(E_i,\pk_i)_{i=1}^{n'}),\pi)=1$ and $(s_1~||~s_2~||~\dots~||~s_{n'}) \not \in \LL^{\SSS}_{n,t}$ (which is equivalent to $(s_1~||~s_2~||~\dots~||~s_{n'}) \not \in \LL^{\SSS}_{t}$ since $\LL^{\SSS}_{t}=\cup_{i=t+1}^n \LL^\SSS_{i,t}$) is negligible. 

Third, suppose $(s'_i,\pi_i)$ is provided by participant $P_i$ for decryption process. If~$s_i' \neq s_i$ and $\PVSS.\Dec\Verify$ returns $1$, there are two possible cases:
\begin{itemize}
\item There does not exist $\sk'_i$ such that $(s'_i,w_i)=\PKE.\Dec(\pp',\pk_i,\sk'_i,E_i)$ and $\PKE.\Verify(\pp',\pk_i,\sk'_i)=1$ for some $w_i \in \mathcal{W}^\Dec_{sound}$. Thus $(\pp',\pk_i,E_i,s'_i) \not \in \LL^\Dec_{sound}$, and the probability that $\NIZK_2.\Verify$ returns $1$ is  negligible.  Hence, in this case, $\PVSS.\Dec\Verify$ also returns $1$ with negligible probability. 
\item There exists $\sk'_i$ such that $(s'_i,w_i)=\PKE.\Dec(\pp',\pk_i,\sk'_i,E_i)$ and $\PKE.\Verify(\pp',\pk_i,\sk'_i)=1$ for some $w_i \in \mathcal{W}^\Dec_{sound}$, but $s'_i \neq s_i$. Recall that there is at most one $\sk'_i$ that makes $\PKE.\Key\Verify$ returns $1$, hence, $\sk'_i=\sk_i$ and thus $s'_i=s_i$. So, we reach a contradiction, and this case cannot happen.
\end{itemize}

Finally, suppose $(\pp',(E_i,\pk_i)_{i=1}^{n'}) \in \LL^\Enc_{sound}$. Then it holds that $(s_i,w_i)=\PKE.\Dec(\pp',\pk_i,\sk_i,E_i)$ for all $i \not \in \CC$ and $w_i \in \mathcal{W}^\Dec_{zk}$. Thus, $(\pp',\pk_i,E_i,s_i) \in \LL^\Dec_{zk}$ with suitable witnesses $(s_i,r_i,w_i)$. Since $\pi^2_i$ is provided by $\NIZK_2.\Prove$, the probability that $\NIZK_2.\Verify$ (and thus $\PVSS.\Dec\Verify$) returns $0$ is negligible.

Hence, all the events above happen with negligible probability, thus, the probability that $\Adversary$ could win is negligible due to union bound, as desired.\end{proof}

\subsection{Proof of Theorem \ref{theorem-privacy}}\label{appendix-pvss-privacy}

\begin{proof}[Proof of  Theorem \ref{theorem-privacy}]
Let $\CC'$ be the set of corrupted participants that have passed $\PVSS.\Key\Verify$. Clearly $|\CC'| \leq t$. Also, recall that we assume $G=\{1,2,\dots,n'\}$ be the set of participants who have passed $\PVSS.\Key\Verify$ (here $n'=n-t+|\CC'| \geq t+1)$. To prove the privacy property, we define the following sequence of games. \vspace{0.07cm}

\noindent $\GAME_{b,1}$: The original game in Definition \ref{definition-pvss-privacy}. Initially, $\Adversary$ submits $(\pk_i,\pi_i)_{i \in \CC}$, while the challenger submits $(\pk_i,\pi_i)_{i \not \in \CC}$. Let $G$ be the set of participants who passed key verification, and assume $G=\{1,2,\dots,n'\}$. In each query, $\Adversary$ chooses $s_i$ and executes $\OO_{\PVSS,\Adversary}(s_i)$ to share and reconstruct $s_i$. The challenge phase is when $\Adversary$ outputs $s^0,s^1$. The challenger  computes $((E^b_i)_{i=1}^n,\pi) \leftarrow \PVSS.\Share(\pp,(\pk_i)_{i=1}^{n'},s^b,t)$ and gives $((E^b_i)_{i=1}^n,\pi)$ to $\Adversary$. Finally $\Adversary$ outputs a bit $b'$.

\noindent $\GAME_{b,2}:$ Same as $\GAME_{b,1}$, except that in $\OO_{\PVSS,\Adversary}(.)$ and the challenge phase, instead of creating the proofs from $\NIZK_1$ in $\PVSS.\Share$, the challenger instead uses the simulator to simulate the proofs using the instance $(\pp',n',t,(\pk_i,E_i))_{i=1}^{n'}$ (note that in this game, the challenger still knows all the witnesses for both $\NIZK_1$).

\noindent $\GAME_{b,3}:$ Same as $\GAME_{b,2}$, except that in $\OO_{\PVSS,\Adversary}(.)$, instead of creating the proofs from $\NIZK_2$ in $\PVSS.\Share$, the challenger instead uses the simulator to simulate the proofs using the instance $(\pp,\pk_i,E_i,s_i)$.

\noindent $\GAME_{b,4}:$ Same as $\GAME_{b,3}$, except that in the key generation process, instead of creating the proof $\pi_i$ for $\pk_i$ from $\NIZK_0$, the challenger instead uses the simulator to simulate the proofs using the instance $(\pp,\pk_i)$.

\noindent  $\GAME_{b,5}:$ Same as $\GAME_{b,4}$, except that in the challenge phase, after computing $(s^b_i)_{i=1}^{n'} \leftarrow \SSS.\Share(s^b)$, the challenger provides the encryption of $0$  for each $i \not \in \CC'$. In other words, it computes $E^b_i \leftarrow \PKE.\Enc(\pp,\pk_i,0)$ for all $i \not \in \CC'$ and $E^b_i \leftarrow \PKE.\Enc(\pp,\pk_i,s^b_i)$ for all $i \in \CC'$. The other steps are the same.

 \noindent  $\GAME_{b,6}:$ Same as $\GAME_{b,5}$, except that in the challenge phase, the challenger samples the shares $(s'_i)_{i \in \CC'}$ of $0$ instead of $s^b$ in $\PVSS.\Share$. The challenger then computes $(E^b_i)_{i \in \CC'}$ which is encryption of $(s'_i)_{i \in \CC'}$, while providing encryptions $(E^b_i)_{i \not \in \CC'}$ of $0$ for each $i \not \in \CC'$ like $\GAME_{b,4}$. The other steps are the same.  
 
 \vspace{0.2cm}

 We need to prove that $\GAME_{0,0}$ and $\GAME_{1,0}$ are indistinguishable. Note that $\GAME_{0,6}$ and $\GAME_{1,6}$ are identical: Initially, the challenger samples the public keys $(\pk_i)_{i \not \in \CC}$ and provides the simulated proof to $\Adversary$. The queries of $\OO_{\PVSS,\Adversary}$ are also identical: In both games, after receiving the secret $s_i$ from $\Adversary$, the challenger produces the shares of $s_i$ and the simulated proofs. In the challenge phase, in both games, the challenger samples $(s_i)_{i \in \CC'}$ s.t. they are valid shares of $0$, and provides the encryptions $(E^b_i)_{i=1}^{n'}$ s.t. $(E^b_i)_{i \not \in \CC'}$ are encryptions of $0$, while $(E^b_i)_{i \in \CC'}$ are encryptions of $(s_i)_{i \in \CC'}$. So the distribution of $(E^b_i)_{i=1}^{n'}$ is identical. Finally, the challenger provides the simulated proof $\pi\leftarrow \Simulator_{\pi}(\crs,(\pp',n',t,(\pk_i,E^b_i)_{i=1}^{n'},\rho)$, so the distribution of $\pi$ in both games are also identical. Therefore, both games are identical. 

Now it suffices to prove that $\GAME_{b,i}$ and $\GAME_{b,i+1}$ are indistinguishable in the sense that $|\Pr[\GAME_{b,i}(\Adversary)=1]-\Pr[\GAME_{b,i+1}(\Adversary)=1]| \leq \negl(\lambda)$ for all $b \in \{0,1\}$ and $1 \leq i \leq 5$. This implies $\GAME_{0,1}$ is indistinguishable from $\GAME_{0,6}$,  and  $\GAME_{1,1}$ is also indistinguishable from $\GAME_{1,6}$, thus we are done.
 
We see that $\GAME_{b,2}$ is indistinguishable from $\GAME_{b,1}$ due to the adaptive multi-theorem zero-knowledge property of the NIZK. Indeed, consider the adversary who distinguishes $\GAME_{b,2}$ from $\GAME_{b,1}$, we construct an adversary $\Adversary'$ breaking the multi-theorem zero-knowledge property as follows.
\begin{itemize}
\item $\Adversary'$ receives $\crs$ from the NIZK challenger, which is the CRS of $\NIZK_1.\Setup$ or a simulated $\crs$ from the simulator.
\item $\Adversary'$ then sends $\crs$ to $\Adversary$ and honestly executes the key generation process and receives $(\pk_i)_{i=1}^{n'}$ and secret keys $(\sk_i)_{i \in G \setminus \CC}$.
\item Whenever $\Adversary$ asks for a share $s_i$, $\Adversary'$ computes $(s_{ij})_{j=1}^n$, $(E_{ij})_{j=1}^n$, then submits the instance and witness to the NIZK challenger and receives $\pi$ which is from $\NIZK_1.\Prove$ or the NIZK simulator. $\Adversary'$ then returns the $(E_i,\pi)$ to $\Adversary$. 
\item $\Adversary'$ performs the remaining steps in $\OO_{\PVSS,\Adversary}(s_i)$. In the challenge phase, $\Adversary'$ performs the same action as above and outputs whatever $\Adversary$ outputs.
\end{itemize}
We see that if the CRS is from $\NIZK_1.\Setup$ and the proofs are from $\NIZK_1.\Prove$, then $\Adversary'$ receives the real proofs and $\Adversary$ receives the transcript of $\GAME_{b,1}$. If the CRS and proofs are from the simulator, then $\Adversary'$ receives the simulated proofs, and $\Adversary$ receives the transcript of $\GAME_{b,2}$. Thus $|\Pr[\GAME_{b,2}(\Adversary)=1]-\Pr[\GAME_{b,1}(\Adversary)=1]| \leq  \Adv^{\mathsf{ZK}}(\Adversary')$. Thus, if $\Adversary$ distinguishes $\GAME_{b,1}$ and $\GAME_{b,2}$ with non-negligible probability $\epsilon$, then $\Adversary'$ can distinguish between the real and simulated proofs with probability $\epsilon$, contradiction. 

For the same reason, we see that  $|\Pr[\GAME_{b,3}(\Adversary)=1]-\Pr[\GAME_{b,2}(\Adversary)=1]| \leq  \Adv^{\mathsf{ZK}}(\Adversary')$ and $|\Pr[\GAME_{b,4}(\Adversary)=1]-\Pr[\GAME_{b,3}(\Adversary)=1]| \leq  \Adv^{\mathsf{ZK}}(\Adversary')$.

Next, $\GAME_{b,5}$ is indistinguishable from $\GAME_{b,4}$ due to the multi-key IND-CPA security of the PKE. Indeed, if $\Adversary$ could distinguish between $\GAME_{b,5}$ and $\GAME_{b,4}$, we construct an adversary $\Adversary'$ breaking the multi-key IND-CPA property as follows.
\begin{itemize}
\item $\Adversary'$ receives $n-t$ public keys $(\pk_i)_{i \not \in \CC'}$ from the multi-key IND-CPA challenger. $\Adversary'$ then generates the simulated CRS for $\NIZK_0,\NIZK_1,\NIZK_2$ and gives them to $\Adversary$. $\Adversary'$ uses the simulator of $\NIZK_0$ to produce the simulated proof for $\pk_i$ for each $i \not \in \CC'$. $\Adversary$ then passes $\pk_i$ and the simulated proofs to $\Adversary$.
\item On each query of $\OO_{\PVSS,\Adversary}(.)$, $\Adversary'$ receives $s_i$ from $\Adversary$, and computes the shares $(s_{ij})_{j=1}^{n'}$, $\Adversary'$. $\Adversary'$ computes the encryption $(E_{ij})_{j=1}^{n'}$ using $(\pk_i)_{i=1}^{n'}$. Then, it sends $\Adversary$ the encryptions and simulated proofs using the simulator of $\NIZK_1$.
\item Also, in each query of $\OO_{\PVSS,\Adversary}(.)$, when needing to reveal $(s_{ij})_{j \not \in \CC'}$, $\Adversary'$ uses the NIZK simulator of $\NIZK_2$ to simulate the proofs.
\item In the final step, $\Adversary'$ receives the $s^0,s^1$ from the $\Adversary$ and computes the shares $(s^b_{ij})_{j=1}^{n'}$ of $s^b$. $\Adversary'$ then outputs i) the \textit{vector} $(s^b_{ij})_{j \not \in \CC'}$ and ii) a \textit{vector} of zeroes to the multi-key IND-CPA challenger and receives $(E_i^\star)_{i \not \in \CC'}$. $\Adversary'$ at the same time computes the encryption $(E_i^\star)_{i \in \CC'}$ of $(s^\star_{ij})_{j \in \CC'}$. $\Adversary'$ then returns $(E_i^\star)_{i \in [n']}$ to $\Adversary$. Finally, $\Adversary'$ performs the remaining steps and outputs whatever $\Adversary$ outputs.
\end{itemize}
 Now, consider $(E^\star_i)_{i \not\in \CC'}$. If they are the encryptions of zeroes, then $\Adversary'$ interacts with $\GAME_0^{\Gen-\mathsf{IND-CPA}}$ and $\Adversary$ receives the transcript of $\GAME_{b,5}$. Otherwise, they are encryptions of $(s^b_i)_{i \not \in \CC},$ thus $\Adversary'$ interacts with $\GAME_1^{\Gen-\mathsf{IND-CPA}}$ and $\Adversary$ receives the transcript of $\GAME_{b,4}$. So we see that $|\Pr[\GAME_{b,4}(\Adversary)=1]-\Pr[\GAME_{b,5}(\Adversary)=1]| \leq  \Adv^{\Gen-\mathsf{IND-CPA}}(\Adversary',n-t)$. Therefore, if $\Adversary$ distinguishes $\GAME_{b,4}$ and $\GAME_{b,5}$ with non-negligible probability $\epsilon$, then $\Adversary'$ can break the multi-key IND-CPA security with probability $\epsilon$, contradiction.

Finally, consider $\GAME_{b,6}$ and $\GAME_{b,5}$: The difference only lies in the final transcript $((E^b_i)_{i=1}^{n'},\pi)$. In both games, $\Adversary$ receives the transcript $((E^b_i)_{i=1}^{n'},\pi)$ where $\pi\leftarrow \Simulator_{\pi}(\crs,(\pp',n',t,(\pk_i,E^b_i)_{i=1}^{n'}),\rho)$. The transcript $(E^b_i)_{i=1}^{n'}$ can be split into $(E^b_i)_{i \not \in \CC'}$ and $(E^b_i)_{i \in \CC'}$ where the difference is in $(E^b_i)_{i \in \CC'}$. In  $\GAME_{b,6}$, $(E^b_i)_{i \in \CC'}=(\PKE.\Enc(\pp,\pk_i,s'_i))_{i \in \CC'}$ where $(s'_i)_{i\in \CC'}$ are the shares of $s^b$. In $\GAME_{b,5}$ however, the shares $(s'_i)_{i\in \CC'}$ are the shares of $0$ instead. Note that the distribution of $(E^b_i)_{i \in \CC'}$ in both games is identical because the distribution of $(s'_i)_{i\in \CC'}$ in both games are identical due to the $t$-privacy property of  Shamir secret sharing. Thus, the distribution of the transcript $((E^b_i)_{i=1}^{n'},\pi)$ in $\GAME_{b,5}$ and $\GAME_{b,6}$ are identical. 

In conclusion, $|\Pr[\GAME_{b,0}(\Adversary)=1]-\Pr[\GAME_{b,6}(\Adversary)=1]| \leq 3\Adv^{\mathsf{ZK}}(\Adversary')+\Adv^{\mathsf{Gen-IND-CPA}}(\Adversary') \leq \negl(\lambda)$ for all $b \in \{0,1\}$. Since we have proved that $\Pr[\GAME_{0,6}(\Adversary)=1]=\Pr[\GAME_{1,6}(\Adversary)=1]$ at the beginning, we thus conclude that $|\Pr[\GAME_{0,0}(\Adversary)=1]-\Pr[\GAME_{1,0}(\Adversary)=1]| \leq \negl(\lambda)$ as well, as desired.
\end{proof}

\section{Formal Security Proof of the Trapdoor $\Sigma$-Protocols}\label{appendix-nizk-security}
\subsection{Proof of Theorem \ref{theorem-nizk-for-key-generation}} \label{appendix-proof-of-key-generation}

\begin{proof}[Proof of Theorem \ref{theorem-nizk-for-key-generation}]
 The correctness property is trivial: It is easy to verify that if $(\bA,\bb) \in \LL^\Key_{zk}$,  then $\bz^\top\cdot \bA+\bt^\top=\bd+c\cdot \bb \pmod{q}$. Now, according to Lemma \ref{lemma-rejection-sampling}, since $||(\bz~||~\bt)|| \leq \sqrt{(B^\Key_\bs)^2+(B^\Key_\be)^2}$ and with the choice of $\sigma^\Key$, the distribution of $(\bz~||~\bt)$ is statistically close to the distribution of $(\bz~||~\bt)$ when sampled from $D_{\mathbb{Z}^{u+v},\sigma^\Key}$ and returned with probability $1/M$. In addition, due to Lemma \ref{lemma-bound}, for any $(\bz~||~\bt)$ sampled from the latter distribution, it holds that $||(\bz~||~\bt)|| \leq \sqrt{u+v}\cdot \sigma^\Key$ with probability $1-2^{O(-(u+v))}$. Thus, by combining the two lemmas, the verifier accepts with probability $1-2^{O(-(u+v))}$. Now to show zero-knowledge, given a challenge $c \in \{0,1\}$ and statement $(\bb,(\bc_1,\bc_2),\bmm)$, the simulator proceeds as follows.
\begin{itemize}
\item Sample $ (\bz~||~\bt) \leftarrow D_{\mathbb{Z}^{u+v},\sigma^\Key}$. 
\item Compute  $\bd=\bz^\top\cdot \bA+\bt^\top-c\cdot \bb \pmod{q}$.
\item Return $(\bd,\bz,\bt)$ with probability $1/M$.
\end{itemize}

To prove that the simulated transcript is indistinguishable from the real transcript, we consider the distribution $(\bd,\bz,\bt)$. Note that, given $(\bz,\bt)$, the value $\bd$ is uniquely determined from $\bd=\bz^\top\cdot \bA+\bt^\top-c\cdot \bb \pmod{q}$. Thus, it suffices to prove that the transcript $(\bz,\bt)$ in both real and simulated transcripts is indistinguishable. 

Indeed, the component $(\bz~||~\bt)$ in the simulated transcript is sampled from $D_{\mathbb{Z}^{u+v},\sigma^\Key}$ and outputted with probability $1/M$, while the component $(\bz~||~\bt)$ from the real transcript is sampled from  $D_{\mathbb{Z}^{u+v},\sigma^\Key,c\cdot \be}$ and is outputted with probability~$1-K$. Due to the choice of $\sigma^\Key$ Lemma \ref{lemma-rejection-sampling}, note that by letting $\bv$ in the lemma to be $c\cdot \be$, the two transcripts have statistical distance at most $2^{-100}/M$ (Note that in the lemma, the algorithm outputs $((\bz~||~\bt),c\cdot (\bs~||~\be))$, but here we only output $\bt$. But if we let $f(\bz,\bv)=\bz$, then note that statistical distance does not increase when applying any function. Thus, the transcript of $f((\bz~||~\bt),c\cdot (\bs~||~\be))$ in both cases is statistically close as well). Hence, the transcript $(\bz,\bt)$ of both real and simulated cases is statistically close.

Next, to prove special soundness, we prove that, given $\bd$, there cannot be two valid responses $(\bz_c,\bt_c)$ for all $c \in \{0,1\}$. Indeed, suppose otherwise, then it holds that  $\bb=(\bz_1^\top-\bz_0^\top)\cdot \bA+(\bt_1^\top-\bt_0^\top) \pmod{q}$. Since   $||\bz_1-\bz_0|| \leq 2\sqrt{u+v}\cdot \sigma^\Key=B^{\Key\star}_\bs$ and $||\bt_1-\bt_0|| \leq 2\sqrt{u+v}\cdot \sigma^\Key=B^{\Key\star}_\be$, this implies that $\bb \in \LL^\Key_{sound}$, contradiction. Thus, special soundness must hold.

To prove CRS indistinguishability, note that $\Trap\Sigma.\Gen$ returns a uniform matrix~$\bA$ in $\mathbb{Z}_q^{v \times u}$, while $\Trap\Sigma.\TrapGen$ returns a matrix $\bA$ which is statistically close to uniform, thus the two CRS are statistically close, as desired.

Finally, we need to prove that the function $\BadChallenge$ provides the correct output. First, we prove that $\BadChallenge$ must return a bit in $\{0,1\}$. Indeed, the $\BadChallenge$ function returns $\perp$ iff the value $(\bz_c,\bt_c)$ satisfies  and $||(\bz_c~||~\bt_c)||<\sqrt{u+v}\cdot \sigma^\Key$ for all $c$. However, similar to the proof of soundness, this will lead to $\bb \in \LL^{\Key}_{sound},$ a contradiction.
Now, consider $c$ such that $(\bz_c,\bt_c)=\Invert(\bA,\bT,\bd+c\cdot \bb)$, and $||(\bz_c~||~\bt_c)|| > \sqrt{u+v}\cdot \sigma^\Key$. Suppose $c=0$ and values $(\bz_0,\bt_0)$ such that $(\bz_0,\bt_0)=\Invert(\bA,\bT,\bd)$ and $||(\bz_0~||~\bt_0)|| > \sqrt{u+v}\cdot \sigma^\Key$. In this case, the $\BadChallenge$ function returns $1$, and we prove that there is no valid response if the challenge is not equal to $1$. Indeed, suppose there exists a valid response for $,c=0$, then it holds that $\bd=\bz_0^\top\cdot\bA+\bt_0^\top$ for some $||(\bz_0~||~\bt_0)|| \leq \sqrt{u+v}\cdot \sigma^\Key$. In this case, the algorithm $\Invert$ correctly inverts $(\bz_0,\bt_0)$ from $(\bA,\bT,\bd)$ due to Lemma \ref{theorem-invert-lwe} and checks $||(\bz_0~||~\bt_0)||$. This means that $\BadChallenge$ returns $1$ even when  $(\bz_0,\bt_0)$ satisfies $||(\bz_0~||~\bt_0)||<\sqrt{u+v}\cdot \sigma^\Key$.  This contradicts the earlier property that $\BadChallenge$ only returns $1$ if $||(\bz_0~||~\bt_0)|| \geq \sqrt{u+v}\cdot \sigma^\Key$ when $(\bz_0~||~\bt_0)$ is computed from $\Invert$. We can prove similarly in the case $\BadChallenge$ returns $0$. Thus, the $\BadChallenge$ function always correctly returns the correct output.  
\end{proof}

\subsection{Proof of Theorem \ref{theorem-nizk-for-sharingg}} \label{appendix-proof-of-sharing}

\begin{proof}[Proof of Theorem \ref{theorem-nizk-for-sharingg}]
The correctness property is trivial: Indeed, when $(\bb_i,\bc_{1i},\bc_{2i})_{i=1}^n \in \LL^\Enc_{zk}$, then $\bmm \in \LL^\SSS_{t}$ and consequently, $\bmm^\top \cdot \bH^{t}_n=0 \pmod{p}$, thus $\bA\cdot \bz_i=\ba_{1i}+c\cdot \bc_{1i} \pmod{q},$  $\bb_i\cdot \bz_i+h_i+p\cdot t_i=\ba_{2i}+c\cdot \bc_{2i} \pmod{q}$ for all $1 \leq i \leq n$\footnote{We need a subtle property here: If $t_i=u_i+c\cdot m_i \pmod{p}$, then $p \cdot t_i = p\cdot (u_i+c\cdot m_i) \pmod{q}$. This works if and only if $q=p^2$, hence the scheme of \cite{ACPS09} aligns perfectly. When $q$ is coprime to $p$ and the value $p\cdot m_i$ is replaced by $\lfloor q/p\cdot m_i \rfloor$ in several schemes, such as \cite{PVW08}, we do not know how to build the trapdoor $\Sigma$- protocol that works in this case.}, and $\bt^\top\cdot\bH^{t}_n=\bzero \pmod{p}$. Now, consider the component $(\bz~||~\bh)$, where $\bz=( \bz_1~||\bz_2~||~\dots~||~\bz_n )$ and similarly $\bh=( h_1~||~h_2~||~\dots~||~h_n )$. By Lemma \ref{lemma-rejection-sampling}, since the norm of $(\br~||~\be')$ is at most $\sqrt{n\cdot (B^\Enc_e)^2+(B^\Enc_\br)^2}$ and $\sigma^\Enc$ is chosen larger than this norm, the component $(\bz~||~\bh)$ is statistically close to a vector $(\bz~||~\bh)$ sampled from $D_{\mathbb{Z}^{un+u},\sigma^\Enc}$ and is returned with probability $1/M$. By using Lemma \ref{lemma-nextbound} on the vector $(\bz~||~\bh)$ with a set $S \subseteq [nu+n]$ of size $u+1$, the norm of $(\bz_i~||~h_i)$ is at most $\sigma^\Enc \cdot \sqrt{u+1}$ with probability with probability $1-2^{O(-(u+1))}$. Thus by union bound, the probability that  $(\bz_i~||~h_i)$ is at most $\sigma^\Enc \cdot \sqrt{u+1}$ for some $1\leq i \leq n$ is at most  $n\cdot 2^{O(-(u+1))}$. Hence verifier accepts with probability $1-n\cdot 2^{O(-(u+1))}$.

Now to show zero-knowledge, given a challenge $c \in \{0,1\}$ and statement $(\bb_i,\bc_{1i},\bc_{2i})$, the simulator proceeds as follows. 
\begin{itemize}
\item Sample $(\bz~||~\bh) \leftarrow  D_{\mathbb{Z}^{nu+n},\sigma^\Enc}$ and a vector $~\bt \in \mathbb{Z}_q$ conditioning to $\bt^\top\cdot\bH^{t}_n=\bzero$. Parse $\bz=(\bz_1~||~ \bz_2~||~ \dots ~||~\bz_n)$, $\bh= (h_1~||~h_2 ~||~ \dots ~||~ h_n)$ where $\bz_i \in \mathbb{Z}_q^u$.
\item Compute $\ba_{1i}=\bA\cdot \bz_i-c\cdot \bc_{1i} \pmod{q},$ $\ba_{2i}=\bb_i\cdot \bz_i+h_i+p\cdot t_i-c\cdot \bc_{2i} \pmod{q}$.
\item Return $(\ba_{1i},\ba_{2i},t_i,\bz_i,h_i)_{i=1}^n$ to $\VV$ with probability $1/M$.
\end{itemize}

To prove that the simulated transcript is indistinguishable from the real transcript, we consider the distribution $(\ba_{1i},\ba_{2i},t_i,\bz_i,h_i)_{i=1}^n$. Note that, given $(t_i,\bz_i,h_i)_{i=1}^n$, the value $(\ba_{1i},\ba_{2i})_{i=1}^n$ is uniquely determined as $\ba_{1i}=\bA\cdot \bz_i-c\cdot \bc_{1i} \pmod{q},$ $\ba_{2i}=\bb_i\cdot \bz_i+h_i+p\cdot t_i-c\cdot \bc_{2i} \pmod{q}$. Thus, it suffices to prove that the transcript $(t_i,\bz_i,h_i)_{i=1}^n$ in both real and simulated transcripts is indistinguishable. 

Indeed, the component $(\bz~||~\bh)$ in the simulated transcript is sampled from $D_{\mathbb{Z}^{nu+n},\sigma^\Enc}$ and outputted with probability $1/M$, while the component $(\bz~||~\bh)$ from the real transcript is sampled from  $D_{\mathbb{Z}^{nu+n},\sigma^\Enc,c\cdot (\br~||~\be')}$ and is outputted with probability $1-K$. Recall that $||(\br~||~\be')|| \leq \sqrt{n\cdot (B_{e}^\Enc)^2+(B_{\br}^\Enc)^2}$ and due to the choice of $\sigma^\Enc$ and Lemma \ref{lemma-rejection-sampling}, note that by letting $\bv$ in the lemma to be $c\cdot (\br~||~\be')$, the two transcripts have statistical distance at most $2^{-100}/M$. In addition, the component $\bt$ of both transcripts are uniformly distributed in the set of vector $\bmm$ such that $\bmm^\top\cdot\bH^{t}_n=\bzero \pmod{p}$. Consequently, the transcript $(t_i,\bz_i,h_i)_{i=1}^n$  of both real and simulated cases are statistically close, as desired.

To prove CRS indistinguishability, note that $\Trap\Sigma.\Gen$ returns a uniform matrix $\bA$ in $\mathbb{Z}_q^{v \times u}$, while $\Trap\Sigma.\TrapGen$ returns a matrix $\bA$ which is statistically close to uniform, thus the two CRSs are statistically close, as desired.

Next, for special soundness, we prove that, given $(\ba_{1i},\ba_{2i})_{i=1}^n$, there cannot be two valid responses $(\bz_{ci},\bh_{ci},\bt_{ci})_{i=1}^n$ for all $c \in \{0,1\}$. Indeed, suppose otherwise. Because we already have that $(\bA,\bb_i)_{i=1}^n \in \LL^\Key_{sound}$ (which is implied after the key generation process). Thus there exist $\bs_i,\be_i$ such that $\bb_i=\bs_i^\top\cdot \bA+\be_i \pmod{q}$ and $||\be_i|| \leq B^{\Enc\star}_{\be}$ for all $1 \leq i \leq n$. Now it holds that  $\bc_{2i}-\bs_i^\top\cdot\bc_{1i}=p\cdot (\bt_{1i}-\bt_{0i})+h_{1i}-h_{0i}+\be_i^\top(\bz_{1i}-\bz_{0i}) \pmod{q}$ and $(\bt_1^\top-\bt_0^\top)\cdot\bH^{t}_n=\bzero \pmod{p}$. Let $f'_i=h_{1i}-h_{0i}+\be_i^\top(\bz_{1i}-\bz_{0i})$ (as integer, not mod $q$).  Since $|f'_i| \leq 2\sigma^\Enc \cdot \sqrt{u+1}\cdot (B^{\Key\star}_\be+1)=B^{\Enc\star}_f$ and $(\bt_1^\top-\bt_0^\top)\cdot\bH^{t}_n=\bzero \pmod{p}$ is equivalent to $(\bt_1^\top-\bt_0^\top) \in \LL^\SSS_{n,t}$, this implies that $(\bA,n,t,(\bb_i,\bc_{1i},\bc_{2i})_{i=1}^n) \in \LL^\Enc_{sound}$, contradiction. Thus, special soundness must hold. 

Finally, we need to prove that the function $\BadChallenge$ provides the correct output. First, we prove that $\BadChallenge$ must return a bit in $\{0,1\}$. Indeed, the $\BadChallenge$ function returns $\perp$ iff 
for all $c \in \{0,1\}$ and $i \in [n]$, the obtained $f_{ci}$ satisfies $||f_{ci}||<B_f^{\Enc^\star}/2$ and $\bt^\top_c\cdot\bH^{t}_n=\bzero \pmod{p}$. Similar to the proof of soundness and recall that the function $\Invert$ correctly restores $\bs_i,\be_i$ for $\bb_i \in \LL^\Key_{sound}$, this implies $(\bA,n,t,(\bb_i,\bc_{1i},\bc_{2i})_{i=1}^n) \in \LL^{\Enc}_{sound}$, contradiction. 

Now, consider $c$ such that $|f_{ci}|>B^{\Enc\star}_f/2$ for some $i$ or $\bt^\top\cdot \bH^{t}_n\neq \bzero \pmod{p}$. Suppose $c=0$, thus the values  $(\ba_{1i},\ba_{2i})_{i=1}^n$ are such that: If we consider $f_{0i}=\ba_{2i}-\bs^\top\cdot \ba_{1i} \pmod{p}$, then cast $f_{0i}$ as an integer in $[-(p-1)/2,(p-1)/2]$ and compute $t_i=(\ba_{2i}-\bs^\top\cdot \ba_{1i}-f_{0i})/p \pmod{p}$, then there must exist some $i$ such that $|f_{0i}|> B^{\Enc\star}_f/2$ for some $i$, or $\bt^\top\cdot \bH^{t}_n\neq \bzero \pmod{p}$.  In this case, the $\BadChallenge$ function always returns $1$, and we prove that if the verifier's challenge is not equal to $1$, then there is no valid response for the prover.  Indeed, suppose there exists a valid response $(\bz_i,t_i,h_i)_{i=1}^n$ for $c=0$, then it holds that   $\bA\cdot \bz_i=\ba_{1i} \pmod{q},$ $\bb_i\cdot \bz_i+h_i+p\cdot t_i=\ba_{2i} \pmod{q}$, $||(\bz_i~||~h_i)|| \leq \sigma^\Enc\cdot \sqrt{u+1}$ for all $1 \leq i \leq n$ and $\bt^\top\cdot\bH^{t}_n=\bzero \pmod{p}$. In this case, recall that since $\bb_i=\bs_i^\top\cdot \bA+\be^\top_i \pmod{q}$ with $||\be_i|| \leq B^{\Key\star}_{\be}$ (due to our assumption that $(\bA,\bb_i) \in \LL^\Key_{sound}$in the theorem's statement), it holds that $\ba_{2i}-\bs_i^\top\cdot\ba_{1i}=pt_i+h_{i}+\be_i^\top\bz_{i} \pmod{q}$ and $\bt\cdot\bH^{t}_n=\bzero \pmod{p}$. Let $f'_{0i}=h_{i}+\be_i^\top\cdot \bz_{i}$ (as integer) then it holds that $|f'_{0i}| \leq B^{\Enc\star}_f/2<p/2$ for all $i$ and $\bt^\top\cdot \bH^{t}_n= \bzero \pmod{p}$. As a result, from the description of $\BadChallenge$, it will extract the values $f'_{0i},t_i$ above from $(\ba_{1i},\ba_{2i})$ and still returns $1$. However, this contradicts the property that $\BadChallenge$ only returns $1$ if the computed $f'_{0i},t_i$  satisfies $|f'_{0i}|> B^{\Enc\star}_f/2$ for some $i$, or $\bt^\top\cdot \bH^{t}_n\neq \bzero \pmod{p}$. We can prove similarly in the case $\BadChallenge$ returns $0$. Thus, the $\BadChallenge$ function always correctly returns the correct output. 
\end{proof}

\subsection{Proof of Theorem \ref{theorem-nizk-for-decryption}} \label{appendix-proof-of-decryption}

\begin{proof}[Proof of Theorem \ref{theorem-nizk-for-decryption}]
The correctness property is trivial: It is easy to verify that if $(\bb,(\bc_1,\bc_2),\bmm) \in \LL^\Dec_{zk}$,  then $\bz^\top\cdot \bA+\bt^\top=\bd+c\cdot \bb \pmod{q}$ and $\bz^\top\cdot \bc_1+t=h+c\cdot (\bc_2-p\cdot m) \pmod{q}$. In addition, according to Lemma \ref{lemma-rejection-sampling}, the distribution of $(\bz~||~\bt~||~t)$ is statistically close to the distribution of $(\bt~||~t)$ when sampled from $D_{\mathbb{Z}^{u+v+1},\sigma^\Dec}$ and returned with probability $1/M$. In addition, due to Lemma \ref{lemma-bound}, for any $(\bz~||~\bt~||~t)$ sampled from the latter distribution, it holds that $||(\bz~||~\bt~||~t)|| \leq \sqrt{v+u+1}\cdot \sigma^\Dec$ with probability $1-2^{O(-(u+v))}$. Thus, by combining the two lemmas, the verifier accepts with probability $1-2^{O(-(u+v))}$. To show zero-knowledge, given a challenge $c \in \{0,1\}$ and statement $(\bb,(\bc_1,\bc_2),m)$, the simulator proceeds as follows.
\begin{itemize}
\item Sample $(\bz~||~\bt~||~t) \leftarrow D_{\mathbb{Z}^{v+u+1},\sigma^\Dec}$. 
\item Compute  $\bd=\bz^\top \bA+\bt^\top-c \cdot \bb \pmod{q},~h=\bz^\top \bc_1+t-c\cdot  (\bc_2-p\cdot m) \pmod{q}$.
\item Return $(\bd,h,\bz,\bt,t)$ with probability $1/M$.
\end{itemize}

To prove that the simulated transcript is indistinguishable from the real transcript, we consider the distribution $(\bd,h,\bz,\bt,t)$. Note that, given $(\bz,\bt,t)$, the values $\bd,h$ are uniquely determined from $\bd=\bz^\top\cdot \bA+\bt^\top-c\cdot \bb \pmod{q}$ and $h=\bz^\top\cdot \bc_1+t-c\cdot (\bc_2-p\cdot m) \pmod{q}$. Thus, it suffices to prove that the transcript $(\bz,\bt,t)$ in both real and simulated transcripts is indistinguishable. 

Indeed, the component $(\bz,\bt,t)$ in the simulated transcript is sampled from $D_{\mathbb{Z}^{v+u++1},\sigma^\Dec}$ and outputted with probability $1/M$, while the component $(\bz,\bt,t)$ from the real transcript is sampled from  $D_{\mathbb{Z}^{v+u+1},\sigma^\Dec,c\cdot (\be~||~f)}$ and is outputted with probability $1-K$. Since $||(\bs~||~\be~||~f)||<\sqrt{(B^\Dec_\bs)^2+(B^\Dec_\be)^2+(B^\Dec_f)^2}$ and due to the choice of $\sigma^\Dec$ and Lemma \ref{lemma-rejection-sampling}, by letting $\bv$ in the lemma to be $c\cdot (\bs~||~\be~||~f)$, the two transcripts have statistical distance at most $2^{-100}/M$.  Consequently, the transcript $(\bz,\bt,t)$ of both real and simulated cases are statistically close, as desired.

Next, to prove special soundness, we prove that, given $(\bd,h)$, there cannot be two valid responses $(\bz_c,\bt_c,t_c)$ for all $c \in \{0,1\}$. Indeed, suppose otherwise, then it holds that $\bc_2-p\cdot m=(\bz_1^\top-\bz_0^\top)\cdot \bc_1+t_1-t_0 \pmod{q}$ and $\bb=(\bz_1^\top-\bz_0^\top)\cdot \bA+(\bt_1^\top-\bt_0^\top) \pmod{q}$. Since $||t_1-t_0|| \leq 2\sqrt{v+u+1}\cdot \sigma^\Dec=B^{\Dec\star}_f$, $||\bz_1-\bz_0|| \leq 2\sqrt{v+u+1}\cdot \sigma^\Dec=B^{\Dec\star}_\bs$ and $||\bt_1-\bt_0|| \leq 2 \sqrt{v+u+1}\cdot \sigma^\Dec=B^{\Dec\star}_\be$, we easily see that $(\bb,(\bc_1,\bc_2),\bmm) \in \LL^\Dec_{sound}$, contradiction. Thus, special soundness must hold.

To prove CRS indistinguishability, note that $\Trap\Sigma.\Gen$ returns a uniform matrix $\bA$ in $\mathbb{Z}_q^{v \times u}$, while $\Trap\Sigma.\TrapGen$ returns a matrix $\bA$ which is statistically close to uniform, thus the two CRS are statistically close, as desired.

Finally, we need to prove that the function $\BadChallenge$ provides the correct output. First, we prove that $\BadChallenge$ must return a bit in $\{0,1\}$. Indeed, the $\BadChallenge$ function returns $\perp$ iff 
for all $c \in \{0,1\}$, the value $(\bz_c,\bt_c,t_c)$ satisfies $||\bz_c||< B^{\Dec\star}_{\bs}/2$, $||\bt_c||< B^{\Dec\star}_{\be}/2$ and $|t_c|<B^{\Dec\star}_{f}/2$. In this case, similar to the proof of soundness, it leads to  $(\bA,(\bb,(\bc_1,\bc_2),m)) \in \LL^\Dec_{sound}$, contradiction. Thus, at least one bad challenge exists. Now, consider $c$ such that $(||\bs_c|| > B^{\Dec\star}_{\bs}/2)~\lor~(||\bt_c|| > B^{\Dec\star}_{\be}/2)~\lor~(|t_c|>B^{\Dec\star}_{f}/2)$.  Suppose $c=0$ and values $(\bz_0,\bt_0)$ such that $(\bz_0,\bt_0)=\Invert(\bA,\bT,\bd)$ and $t_0=h-\bz_0^\top \cdot \bc_1 \pmod{p}$ and $(||\bz_0|| > B^{\Dec\star}_{\bs}/2)~\lor~||\bt_0|| > B^{\Dec\star}_{\be}/2)~\lor~(|t_0|>B^{\Dec\star}_{f}/2)$.  In this case, the $\BadChallenge$ function always returns $1$, and we prove that if the verifier's challenge is not equal to $1$, then there is no valid prover's response. Indeed, suppose there exists a valid response for $c=0$, then it holds that $\bd=\bz_0^\top\cdot\bA+\bt_0^\top$ and   $h-\bz_0^\top \cdot \bc_1=t_0+p\cdot m_0 \pmod{q}$ for some $||(\bt_0~||~t_0)|| \leq \sqrt{u+1}\cdot \sigma^\Dec$ and $m \in \mathbb{Z}_p$. In this case, note that the algorithm $\Invert$ correctly inverts $(\bz_0,\bt_0)$ from $(\bA,\bT,\bd)$, then correctly computes $t_0$ such that $(||\bz_0|| < B^{\Dec\star}_{\bs}/2)~\land~(||\bt_0|| < B^{\Dec\star}_{\be}/2)~\land~(|t_0|<B^{\Dec\star}_{f}/2)$. Thus, from the description of the $\BadChallenge$ function, it would get the values $(\bz_0,\bt_0,t_0)$, above, then checks their norm and still returns $1$.  This contradicts the earlier property that $\BadChallenge$ only returns $1$ if $(||\bz_0|| > B^{\Dec\star}_{\bs}/2)~\lor~(||\bt_0|| > B^{\Dec\star}_{\be}/2)~\lor~(|t_0|>B^{\Dec\star}_{f}/2)$. We can prove similarly in the case $\BadChallenge$ returns $0$. Thus, the $\BadChallenge$ function always correctly returns the correct output. 
\end{proof}

\section{On Adaptive Soundness of Libert et al.'s NIZK (Asiacrypt 2020)}\label{appendix-adaptive-soundness}
Recall that in Subsection \ref{subsection-compiler}, we need to prove that the NIZK of \cite{LNPT20} provides adaptive soundness for languages $\LL$ with trapdoor $\tau$ to efficiently check whether an element is in $\LL_{sound}$ with probability $1$, and we also need to prove that, adaptive soundness also holds when we generate $\bA \uniformly \mathbb{Z}_q^{v \times u}$ instead from $\TrapGen$ (because when using the NIZKs for languages $\LL^\Key,\LL^\Enc,\LL^\Dec$ in Section \ref{section-generic-pvss}, we need them to satisfy the adaptive soundness property for these languages).  First, we sketch the proof that the simulation soundness of the NIZK in \cite{LNPT20} achieves adaptive soundness for such languages.  For any $\Adversary$, let 
$$P_1(\Adversary)=\condprob{x \not \in \LL_{sound}~\land 
\NIZK.\Verify(\crs,x,\pi)=1}{\crs\leftarrow \NIZK.\Setup(1^\lambda,\LL,\tau),\\ (x,\pi) \leftarrow \Adversary(\crs,\tau)},$$ $$P_2(\Adversary)=\condprob{x \not \in \LL_{sound}~\land 
\NIZK.\Verify(\crs,x,\pi)=1}{(\crs,\rho)\leftarrow \Simulator_{\crs},\\ (x,\pi) \leftarrow \Adversary(\crs,\tau)}.$$ We need to prove that $|P_1(\Adversary)-P_2(\Adversary)|$ is negligible. Indeed, suppose the contrary, consider the adversary $\Adversary'$ as follows: $\Adversary'$ receives $\crs$ from the ZK challenger, which is the CRS from $\NIZK.\Setup$ or $\Simulator_{\crs}$. $\Adversary'$ gives $\crs$ to $\Adversary$ and  receives $(x,\pi)$ from $\Adversary$ and then uses $\tau$ to efficiently check if $x \in \LL_{sound}$ and $\NIZK\Verify(\crs,x,\pi)=1$. It outputs $1$ iff $x \not \in \LL_{sound}$ and $\NIZK.\Verify(\crs,x,\pi)=1$. 

First, it is known that due to the simulation soundness property (see \cite[Definition A.3]{LNPT20}) of the NIZK in \cite[Theorem 3.4 and Theorem 4.5]{LNPT20}, $P_2(\Adversary)$ is negligible. So, in order to show that $P_1(\Adversary)$ is negligible it suffices to prove that $|P_1(\Adversary)-P_2(\Adversary)|$ is negligible. Indeed, consider the adversary $\Adversary'$ as follows: $\Adversary'$ receives $\crs$ from the ZK challenger, which is the CRS from $\NIZK.\Setup$ or $\Simulator_{\crs}$. $\Adversary'$ gives $\crs$ to $\Adversary$ and  receives $(x,\pi)$ from $\Adversary$ and then uses $\tau$ to efficiently check if $x \in \LL_{sound}$ and $\NIZK\Verify(\crs,x,\pi)=1$. It outputs $1$ iff $x \not \in \LL_{sound}$ and $\NIZK.\Verify(\crs,x,\pi)=1$. We easily see that $|P_1(\Adversary)-P_2(\Adversary)|\leq \Adv^{\mathsf{ZK}}(\Adversary')$. So it holds that $|P_1(\Adversary)-P_2(\Adversary)|$ is negligible because $\Adv^{\mathsf{ZK}}(\Adversary')$ is negligible. Therefore, $P_1(\Adversary)$ must be negligible.

We note that  $P_1(\Adversary)$ is negligible does not immediately imply adaptive soundness because $\Adversary$ is additionally the trapdoor $\tau$ of $\LL$, but in the real execution, $\Adversary$ will not be given $\tau$. However, we can easily prove that it indeed implies soundness as follows: Suppose an adversary $\Adversary$ can break the adaptive soundness property in Definition \ref{definition-adaptive-soundness}, then there exists some $\Adversary'$ that makes $P_1(\Adversary')$ non-negligible. $\Adversary'$, on input the CRS and trapdoor $(\crs,\tau)$, simply provides $\crs$ to $\Adversary$, and outputs whatever $\Adversary$ outputs. It is easy to see that the probability that $\Adversary$ breaks the adaptive soundness property is the same as $P_1(\Adversary')$. So, adaptive soundness is achieved for trapdoor languages.

But we are still not done. The adaptive soundness property only holds when we have a trapdoor $\tau=\bT$ in the languages $\LL^\Key,\LL^\Enc,\LL^\Dec$, so adaptive soundness only holds if the matrix $\bA$ from the PKE is generated from $\TrapGen(1^\lambda,u,v)$ (so we have to modify the algorithm   $\PKE.\Setup$ a bit to use the NIZKs), while in the real PVSS, $\bA$ is a uniformly distributed matrix in $\mathbb{Z}_q^{v \times u}$ without such $\bT$. Fortunately, the distribution of both matrices $\bA$ are statistically close (and so are the $\crs$ of the NIZKs), so with the same argument above, we see that, the distribution of $(x,\pi)$ outputted by $\Adversary$ when given a uniform $\bA \uniformly \mathbb{Z}_q^{u \times v}$ and a matrix $\bA$ generated from $\TrapGen$ are statistically close. Hence, the probability that $x \not \in \LL_{sound}$ and $\NIZK.\Verify(\crs,x,\pi)=1$  in two cases are negligibly close. As a result, if the probability that $x \not \in \LL_{sound}$ and $\NIZK.\Verify(\crs,x,\pi)=1$ is negligible when  $\bA$ is generated from $\TrapGen$, then the same also holds when  $\bA \uniformly \mathbb{Z}_q^{u \times v}$ (thus all the properties of the NIZKs and PVSS are preserved), as desired.

\begin{spacing}{0.9}
\bibliographystyle{plainnat}
\bibliography{refs}
\end{spacing}
\newpage



\end{document}